\documentclass[onecolumn]{IEEEtran}
\usepackage{soul}
\usepackage[mathscr]{eucal}
\usepackage[cmex10]{amsmath}
\usepackage{epsfig,epsf,psfrag}
\usepackage{amssymb,amsmath,amsthm,amsfonts,latexsym}
\usepackage{amsmath,graphicx,bm,xcolor,url,overpic}
\usepackage[caption=false]{subfig} 
\usepackage{fixltx2e}
\usepackage{array}
\usepackage{verbatim}
\usepackage{bm}
\usepackage{algorithmic}
\usepackage{algorithm}
\usepackage{verbatim}
\usepackage{textcomp}
\usepackage{mathrsfs}
\usepackage{epstopdf}

\newcommand{\openone}{\leavevmode\hbox{\small1\normalsize\kern-.33em1}}

\catcode`~=11 \def\UrlSpecials{\do\~{\kern -.15em\lower .7ex\hbox{~}\kern .04em}} \catcode`~=13 

\allowdisplaybreaks[4]

\newcommand{\nn}{\nonumber}

\newcommand{\calA}{\mathcal{A}}
\newcommand{\calB}{\mathcal{B}}
\newcommand{\calC}{\mathcal{C}}

\newcommand{\calE}{\mathcal{E}}
\newcommand{\calF}{\mathcal{F}}

\newcommand{\calM}{\mathcal{M}}

\newcommand{\calP}{\mathcal{P}}

\newcommand{\calS}{\mathcal{S}}
\newcommand{\calT}{\mathcal{T}}

\newcommand{\calV}{\mathcal{V}}

\newcommand{\calX}{\mathcal{X}}
\newcommand{\calY}{\mathcal{Y}}

\newcommand{\bx}{\mathbf{x}}
\newcommand{\bX}{\mathbf{X}}
\newcommand{\by}{\mathbf{y}}
\newcommand{\bY}{\mathbf{Y}}

\newcommand{\rmb}{\mathrm{b}}

\newcommand{\rmc}{\mathrm{c}}

\newcommand{\rmj}{\mathrm{j}}

\newcommand{\rmS}{\mathrm{S}}

\newcommand{\rmV}{\mathrm{V}}

\newcommand{\bbN}{\mathbb{N}}

\newcommand{\bbR}{\mathbb{R}}

\DeclareMathAlphabet{\mathbsf}{OT1}{cmss}{bx}{n}
\DeclareMathAlphabet{\mathssf}{OT1}{cmss}{m}{sl}

\DeclareSymbolFont{bsfletters}{OT1}{cmss}{bx}{n}  
\DeclareSymbolFont{ssfletters}{OT1}{cmss}{m}{n}
\DeclareMathSymbol{\bsfGamma}{0}{bsfletters}{'000}
\DeclareMathSymbol{\ssfGamma}{0}{ssfletters}{'000}
\DeclareMathSymbol{\bsfDelta}{0}{bsfletters}{'001}
\DeclareMathSymbol{\ssfDelta}{0}{ssfletters}{'001}
\DeclareMathSymbol{\bsfTheta}{0}{bsfletters}{'002}
\DeclareMathSymbol{\ssfTheta}{0}{ssfletters}{'002}
\DeclareMathSymbol{\bsfLambda}{0}{bsfletters}{'003}
\DeclareMathSymbol{\ssfLambda}{0}{ssfletters}{'003}
\DeclareMathSymbol{\bsfXi}{0}{bsfletters}{'004}
\DeclareMathSymbol{\ssfXi}{0}{ssfletters}{'004}
\DeclareMathSymbol{\bsfPi}{0}{bsfletters}{'005}
\DeclareMathSymbol{\ssfPi}{0}{ssfletters}{'005}
\DeclareMathSymbol{\bsfSigma}{0}{bsfletters}{'006}
\DeclareMathSymbol{\ssfSigma}{0}{ssfletters}{'006}
\DeclareMathSymbol{\bsfUpsilon}{0}{bsfletters}{'007}
\DeclareMathSymbol{\ssfUpsilon}{0}{ssfletters}{'007}
\DeclareMathSymbol{\bsfPhi}{0}{bsfletters}{'010}
\DeclareMathSymbol{\ssfPhi}{0}{ssfletters}{'010}
\DeclareMathSymbol{\bsfPsi}{0}{bsfletters}{'011}
\DeclareMathSymbol{\ssfPsi}{0}{ssfletters}{'011}
\DeclareMathSymbol{\bsfOmega}{0}{bsfletters}{'012}
\DeclareMathSymbol{\ssfOmega}{0}{ssfletters}{'012}

\newcommand{\hatH}{\hat{H}}

\newcommand{\hatT}{\hat{T}}

\DeclareMathOperator*{\argmin}{arg\,min}

\newtheorem{theorem}{Theorem} 
\newtheorem{lemma}[theorem]{Lemma}

\newtheorem{proposition}[theorem]{Proposition}
\newtheorem{corollary}[theorem]{Corollary}
\newtheorem{definition}{Definition}

\usepackage{picture,xcolor}

\graphicspath{{./figures/}}
\usepackage[ colorlinks = true,
             linkcolor = blue,
             urlcolor  = blue,
             citecolor = red,
             anchorcolor = green,]{hyperref}
\usepackage{soul,color}
\usepackage{epstopdf,cite}

\newcommand{\myfoot}[1]{\footnote{\color{red}\bf #1}}
\soulregister\em7
\soulregister\cite7
\soulregister\ref7
\soulregister\eqref7
\soulregister\underline7
\soulregister\emph7
\soulregister\footnote7
\soulregister\myfoot7

\begin{document}
\title{Achievable Moderate Deviations Asymptotics for Streaming Compression of Correlated Sources}

\author{\IEEEauthorblockN{Lin Zhou, Vincent Y.\ F.\ Tan, and Mehul Motani} \thanks{The authors are with the Department of Electrical and Computer Engineering, National University of Singapore  (emails: lzhou@u.nus.edu, vtan@nus.edu.sg, motani@nus.edu.sg). V.~Y.~F.~Tan is also with the Department of Mathematics, NUS. }
\thanks{The work of the authors is supported in part by a Ministry of Education Tier 2 grant (grant number R-263-000-B61-112).} 
\thanks{Part of this paper was presented at ISIT 2017~\cite{zhou2017achievable}.}
}

\maketitle

\flushbottom
\allowdisplaybreaks[1]

\begin{abstract}
Motivated by streaming multi-view video coding and wireless sensor networks, we consider the problem of blockwise streaming compression of a pair of correlated sources, which we term streaming Slepian-Wolf coding. We study the moderate deviations regime in which the rate pairs of a sequence of codes converge, along a straight line, to various points on the boundary of the Slepian-Wolf region at a speed slower than the inverse square root of the blocklength $n$, while the error probability decays subexponentially fast in $n$. Our main result focuses on directions of approaches to corner points of the Slepian-Wolf region. It states that for each correlated source and all corner points, there exists a non-empty subset of directions of approaches such that the moderate deviations constant (the constant of proportionality for the subexponential decay of the error probability) is enhanced (over the non-streaming case) by at least a factor of $T$, the block delay of decoding source block pairs. We specialize our main result to the setting of streaming lossless source coding and generalize this result  to the setting where we have different delay requirements for each of the two source blocks. The proof of our main result involves the use of various analytical tools and amalgamates several ideas from the recent information-theoretic streaming literature. We adapt the  so-called truncated memory encoding idea from Draper and Khisti (2011) and Lee, Tan, and Khisti (2016) to ensure that the effect of error accumulation is nullified in the limit of large blocklengths. We also adapt the use of the so-called minimum weighted empirical suffix entropy decoder which was used by Draper, Chang, and Sahai (2014) to derive achievable error exponents for symbolwise streaming Slepian-Wolf coding.
\end{abstract}


\section{Introduction}\label{sec:intro}

Multi-view video coding (MVC) has found important applications in 3D television and surveillance~\cite{flierl2007mvc,guo2008wyner, Wang00, Wilburn04}. In MVC, we usually have multiple cameras recording the same scene from different locations and/or angles. The video frames from each camera are then compressed separately by the encoders and transmitted to a control center.  The video frames from the different cameras are highly correlated since the cameras are recording the same scene. Each view is  also corrupted by noise that emanates from various sources in the environment.   The control center (decoder) then aims to recover the video frames sequentially by tolerating a small pre-specified delay. Another motivation for the present work arises from wireless sensor networks \cite{Sohraby07} which are deployed to monitor some ambient environment. Consider a scenario where we have multiple sensors monitoring the temperature and humidity in a given location. The data monitored by each sensor is compressed by an encoder and transmitted to a control center periodically.  Environmental data typically exhibits both temporal and spatial correlation~\cite{Piegorsch}.  The control center aims to recover the measurements, e.g., temperature and humidity, accurately but can tolerate some small delay. Note that, in both cases, the encoder has access to the data (i.e, video frames or environmental measurements) in an incremental manner.

In an effort to  characterize the fundamental performance limits of streaming MVC and wireless sensor network applications, we propose a \emph{streaming} version of  the Slepian-Wolf (SW) source coding problem~\cite{slepian1973noiseless} as an information-theoretic model. Our setting is shown pictorially in Figure~\ref{systemmodelsw}. In this setting, the correlated source generates one source block pair (each of length $n$) and the $k$-th pair of  encoders has access to {\em all} $k$ source block pairs generated up to and including the current time. The decoders incur a block delay $T$ when decoding each pair of source blocks. We would like to evaluate the maximum of the error probabilities over {\em all} source blocks under the so-called moderate deviations regime. We note our setting is, in general, \emph{different} from the settings in \cite{chang2007streaming} and \cite{draper2010lossless} where the authors restricted the encoders to have access to accumulated {\em symbol pairs}  and not {\em blocks} of  symbol pairs  at each time. Hence, we term the  settings  in \cite{chang2007streaming} and \cite{draper2010lossless} as \emph{symbolwise} streaming and our setting as \emph{blockwise} streaming. 

There are three asymptotic regimes in the study of Shannon-theoretic problems when we want to establish the relationships and tradeoffs between the blocklength, coding rate(s) and the error probability. They are, respectively, the error exponent (large deviations), second-order (central limit) and moderate deviations regimes. Here we use fixed-length lossless source coding as an example to discuss these three regimes. In the traditional study of error exponents for source coding, the rate of the code $R$ is fixed at a value strictly above entropy $H(P_X)$ (which is the first-order fundamental limit) and the exponential rate of the decay of the error probability is sought~\cite{csiszar2011information,gallagerIT}. In second-order or normal approximation analysis~\cite{TanBook, Hayashi08,strassen1962asymptotische}, the rate $R_n$ converges to the first-order fundamental limit at a rate of the order $1/\sqrt{n}$ and the error probability converges to a constant between $0$ and $1$. In contrast, in the study of moderate deviations~\cite{altugwagner2014,polyanskiy2010channel,altug2013lossless}, the rate of the code $R_n$ depends on the blocklength $n$ and converges to the entropy $H(P_X)$  (the first-order fundamental limit) at a speed slower than $1/\sqrt{n}$ while the error probability decays to zero at a sub-exponential speed of roughly $\exp(-\nu n^{1-2t})$ for $t\in (0,1/2)$. The constant $\nu$ is known as the moderate deviations constant and is the object of study in moderate deviations analysis here. The moderate deviations regime can be seen as a bridge between the error exponent and second-order regimes.

\subsection{Related Work}

The papers that are most  related to the present one are those by Lee, Tan and Khisti~\cite{lee2015streaming} and Draper, Chang and Sahai~\cite{draper2010lossless}. In \cite{lee2015streaming}, the authors study the information-theoretic limits of the blockwise streaming version of channel coding in  the moderate deviations and central limit regimes. In \cite{draper2010lossless}, the authors derived lower bounds for the error exponent (reliability function) of symbolwise streaming SW coding by using random binning, minimum weighted empirical entropy decoding, and maximum likelihood decoding. 
 In other works on information-theoretic limits of streaming compression and transmission, Chang and Sahai~\cite{chang2006error} derived bounds on the error exponent for symbolwise streaming of lossless compression. They demonstrated similar results for the case with both encoder and decoder side information in \cite{chang2007price}. They also extended the feedforward decoder idea (which originated from Pinsker~\cite{pinsker1967bounds}) to the case with decoder side information  to derive an upper (converse) bound on the error exponent. Other works by Chang on information-theoretic limits in symbolwise streaming are summarized in~\cite{chang2007streaming}.

Concerning other works on streaming and source coding with delayed decoding, Palaiyanur  \cite{palaiyanur2007sequential} studied lossless streaming compression of a source with side information with and without a discrete memoryless channel between the encoder and the decoder. 
  Matsuta and Uyematsu \cite{matsuta2014wyner} considered the lossy source coding problem with delayed side information. Ma and Ishwar \cite{ma2011delayed} focused on delayed sequential coding of correlated video sources. Zhang, Vellambi and Nguyen~\cite{zhang2014delay} analyzed the error exponent of lossless streaming compression of a single source using variable-length sequential random binning. Finally, Etezadi, Khisti, and Chen~\cite{etezadi2014truncated} recently considered the sequential transmission of a stream of Gauss-Markov sources over erasure channels with zero decoding delay.

Here we also mention a few works on  moderate deviations analyses in information theory. Chen \emph{et al.}~\cite{chen2007redundancy} and He \emph{et al.}~\cite{he2009redundancy} initiated the study of moderate deviations by studying lossless source coding with decoder side information.  Chen \emph{et al.}~\cite{chen2007redundancy} also used duality between source and channel coding to study the moderate deviation asymptotics of so-called cyclic symmetric channels. Altu\u{g} and Wagner~\cite{altugwagner2014} studied moderate deviations for discrete memoryless channels. Polyanksiy and Verd\'u~\cite{polyanskiy2010channel} relaxed some assumptions in the conference version of Altu\u{g} and Wagner's work~\cite{altug2010moderate} and they also considered moderate deviations for AWGN channels. Altu\u{g}, Wagner  and Kontoyiannis~\cite{altug2013lossless} considered moderate deviations for lossless source coding. Other works on moderate deviations in information theory include~\cite{tan2012moderate, tan2014moderate}. Note, however, that all these cited works on moderate deviations analysis in information  theory pertain to point-to-point systems with a {\em single} rate parameter. In this paper, we perform moderate deviations analysis on  a multi-terminal problem involving {\em two} rates and, additionally, we consider the {\em streaming} scenario.

\subsection{Main Contributions}

In this paper, we derive an achievable moderate deviations constant for the blockwise streaming version of SW coding. We show that for each correlated source, there exists a non-empty subset of directions of approaches to the boundary of the SW region for which the moderate deviations constant is enhanced by a factor of $T$, the block delay of decoding source block pairs, compared with non-streaming case. Furthermore, we specialize our results to streaming lossless source coding and generalize our results to the setting where we have different delay requirements for each of the two source blocks.

Despite being analogous to \cite{lee2015streaming} in terms of the main result, our proof is significantly different. We adapt two key ideas to our blockwise streaming setting: the {\em truncated memory encoding} idea in \cite{lee2015streaming} (originated from \cite{drapertruncated} and also related to the notion of tree codes in \cite{wozencraft1957sequential,schulman,sahai2001anytime,sukhavasi2011,khisti2014dmt}) and the {\em  minimum weighted empirical suffix entropy} decoder in \cite{draper2010lossless}. The analysis of this paper is different from both \cite{lee2015streaming} and \cite{draper2010lossless}. In \cite{lee2015streaming}, the truncated memory encoding idea was used in the central limit regime for streaming channel coding. In this paper, we adapt the idea and apply it to the moderate deviations regime of our streaming SW setting. We argue that truncation is, in general, necessary for us to prove our main result (otherwise the accumulation of error probabilities will result). In \cite{draper2010lossless}, the minimum weighted empirical suffix entropy  decoder was used to derive an achievable error exponent for symbolwise streaming SW coding. We adapt the decoding rule to be suited to our moderate deviations setting  so as to obtain an analogous term as in the exponent calculation for  \cite[Theorem 6, Case (iii)]{draper2010lossless}. Subsequently, we Taylor expand the exponent to obtain an achievable moderate deviations constant. We remark that this final Taylor expansion step  is not straightforward, requiring some analytical techniques inspired by Polyanskiy \cite[Lemma 48]{polyanskiy2010thesis}. This is because the exponent involves an additional minimization over a scalar parameter (see~\cite[Theorem 6, Case (iii)]{draper2010lossless} and Lemma \ref{mdcasymp}). 

We emphasize that in \cite{draper2010lossless}, a quantification of the improvement of the error exponent in the streaming case vis-\`a-vis the block coding setting was absent. In this paper, under our proposed coding scheme, we provide a definitive answer to the question of how much we gain in the moderate deviations regime from streaming setup with a block delay $T$. We show analytically that there is a multiplicative gain of $T$ in the moderate deviations constant over the non-streaming setting in many scenarios and provide intuition for why this is the case. We analyze several sources  in Section \ref{sec:num}  and calculate the various directions of approaches in which we can attain this gain.   

\subsection{Organization of the Paper}
The rest of the paper is organized as follows. In Section \ref{sec:sw}, we set up the notation, formulate the problem of streaming SW coding and present our main result--an achievable moderate deviations constant for streaming SW coding. In Section \ref{sec:num}, we provide  three numerical examples to illustrate our results. In particular, we delineate the set of directions for which we can achieve a gain of at least the block delay $T$ in the moderate deviations constant.  In Section \ref{sec:gen}, we specialize our results to streaming lossless source coding with and without decoder side information and generalize our results to the scenario where we impose different delay requirements on different source blocks. In Section \ref{sec:swmdc}, we present the proof of our achievable moderate deviations constant. Finally, in Section \ref{conc}, we conclude the paper and propose future research topics. Auxiliary lemmata are proved in the Appendices.

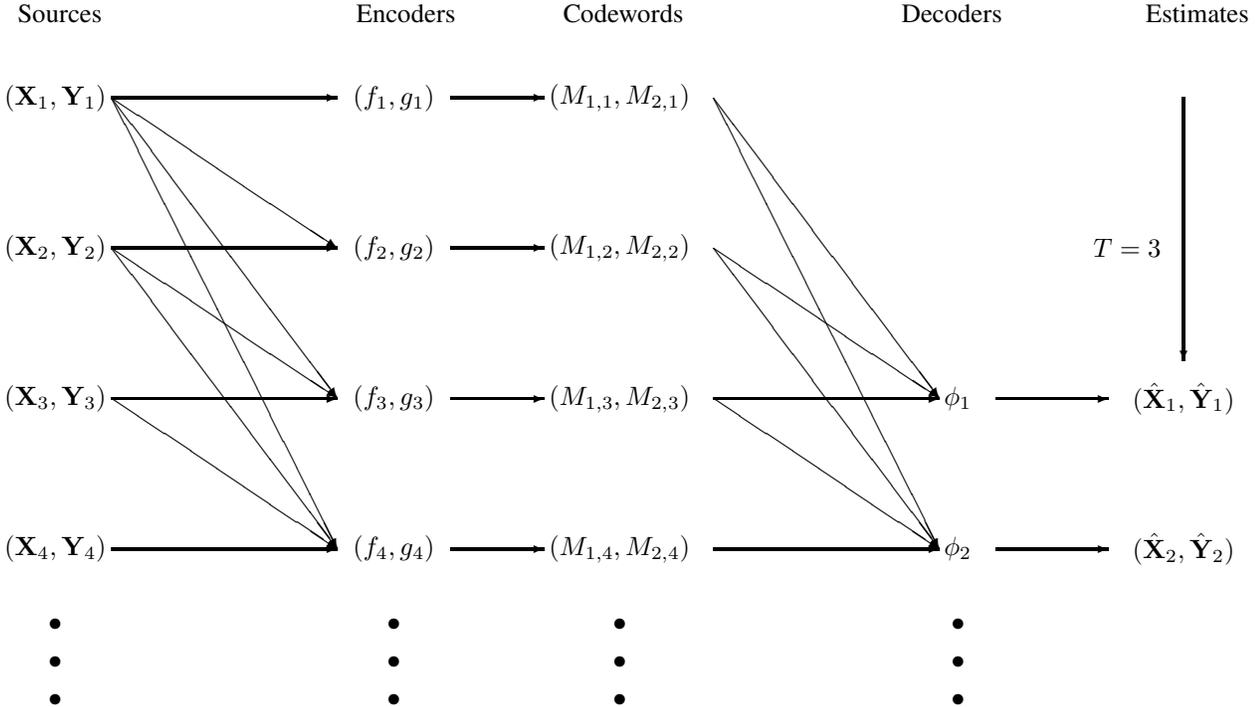
\begin{figure}[t]
\centering
\setlength{\unitlength}{0.5cm}
\begin{picture}(35,18)
\linethickness{1pt}
\put(0,18){\makebox{Sources}}
\put(9,18){\makebox{Encoders}}
\put(14.5,18){\makebox{Codewords}}
\put(23.5,18){\makebox{Decoders}}
\put(30,18){\makebox{Estimates}}
\put(1,16){\makebox(0,0){$(\bX_1,\bY_1)$}}
\put(1,12){\makebox(0,0){$(\bX_2,\bY_2)$}}
\put(1,8){\makebox(0,0){$(\bX_3,\bY_3)$}}
\put(1,4){\makebox(0,0){$(\bX_4,\bY_4)$}}
\put(1,2){\circle*{0.3}}
\put(1,1){\circle*{0.3}}
\put(1,0){\circle*{0.3}}

\put(10,2){\circle*{0.3}}
\put(10,1){\circle*{0.3}}
\put(10,0){\circle*{0.3}}

\put(16,2){\circle*{0.3}}
\put(16,1){\circle*{0.3}}
\put(16,0){\circle*{0.3}}

\put(25,2){\circle*{0.3}}
\put(25,1){\circle*{0.3}}
\put(25,0){\circle*{0.3}}
\put(10,16){\makebox(0,0){$(f_1,g_1)$}}
\put(10,12){\makebox(0,0){$(f_2,g_2)$}}
\put(10,8){\makebox(0,0){$(f_3,g_3)$}}
\put(10,4){\makebox(0,0){$(f_4,g_4)$}}
\put(2.5,16){\vector(1,0){6}}
\put(2.5,16){\vector(3,-2){6}}
\put(2.5,12){\vector(1,0){6}}
\put(2.5,16){\vector(3,-4){6}}
\put(2.5,12){\vector(3,-2){6}}
\put(2.5,8){\vector(1,0){6}}
\put(2.5,16){\vector(1,-2){6}}
\put(2.5,12){\vector(3,-4){6}}
\put(2.5,8){\vector(3,-2){6}}
\put(2.5,4){\vector(1,0){6}}
\put(16,16){\makebox(0,0){$(M_{1,1},M_{2,1})$}}
\put(16,12){\makebox(0,0){$(M_{1,2},M_{2,2})$}}
\put(16,8){\makebox(0,0){$(M_{1,3},M_{2,3})$}}
\put(16,4){\makebox(0,0){$(M_{1,4},M_{2,4})$}}
\put(11.5,16){\vector(1,0){2.5}}
\put(11.5,12){\vector(1,0){2.5}}
\put(11.5,8){\vector(1,0){2.5}}
\put(11.5,4){\vector(1,0){2.5}}
\put(25,8){\makebox(0,0){$\phi_1$}}
\put(25,4){\makebox(0,0){$\phi_2$}}
\put(18.5,16){\vector(3,-4){6}}
\put(18.5,12){\vector(3,-2){6}}
\put(18.5,8){\vector(1,0){6}}
\put(18.5,16){\vector(1,-2){6}}
\put(18.5,12){\vector(3,-4){6}}
\put(18.5,8){\vector(3,-2){6}}
\put(18.5,4){\vector(1,0){6}}
\put(31,8){\makebox(0,0){$(\hat{\bX}_1,\hat{\bY}_1)$}}
\put(31,4){\makebox(0,0){$(\hat{\bX}_2,\hat{\bY}_2)$}}
\put(26,8){\vector(1,0){3}}
\put(26,4){\vector(1,0){3}}
\put(31,16){\vector(0,-1){7}}
\put(29.5,12){\makebox(0,0){$T=3$}}
\end{picture}
\caption{Lossless streaming compression of correlated sources with $T=3$. At each time $k$, a new source block pair $(\bX_k,\bY_k)$ is fed into encoders $(f_k,g_k)$. The encoders then produce output as codewords $(M_{1,k},M_{2,k})$ using the accumulated source blocks $(\bX^k,\bY^k)$. For the case of $T=3$, the decoder is required  to estimate $(X_k,Y_k)$ after receiving the codewords with indices from $1$ to $k+2=k+T-1$.}
\label{systemmodelsw}
\end{figure}

\section{Streaming Slepian Wolf Coding}
\label{sec:sw}
\subsection{Notation}
Random variables and their realizations are in capital (e.g.,\ $X$) and lower case (e.g.,\ $x$) respectively. All sets are denoted in calligraphic font (e.g.,\ $\mathcal{X}$). Let $X^n:=(X_1,\ldots,X_n)$ be a random vector of length $n$. We use $\bX$ to denote $X^n$ and $\bX_a^b$ to denote $X_{n(a-1)+1}^{bn}$. All logarithms are base $e$  (natural logarithm).  As usual,  for any $k\in\bbN$,  $p\log^k p=0$ if $p=0$. Given two integers $a$ and $b$, we use $[a:b]$ to denote the set  $\{a, a+1, \ldots, b-1,b\}$. We use standard asymptotic notation such as $O(\cdot)$  and $o(\cdot)$~\cite{Cor03}. 

The set of all probability distributions on $\calX$ is denoted as $\calP(\calX)$ and the set of all conditional probability distributions from $\calX$ to $\calY$ is denoted as $\calP(\calY|\calX)$. Given $P\in\calP(\calX)$ and $V\in\calP(\calY|\calX)$, we use $P\times V$ to denote the joint distribution induced by $P$ and $V$. We use the method of types extensively and we follow the notation in~\cite{TanBook}. Given a sequence $x^n$, the empirical distribution (type) is denoted as $\hat{T}_{x^n}$. The set of types formed from length-$n$ sequences in $\calX$ is denoted as $\calP_{n}(\calX)=\{\hat{T}_{x^n} \in \calP(\calX):x^n\in\calX^n\}$. Given $P\in\calP_{n}(\calX)$, the set of all sequences of length-$n$ with type $P$ (the type class) is denoted as $\calT_{P}$. Given $x^n\in\calT_{P}$, the set of all sequences $y^n\in\calY^n$ such that the joint type of $(x^n,y^n)$ is $P\times V$ is denoted as $\calT_{V}(x^n)$, the $V$-shell. The set of all stochastic matrices $V\in\calP(\calY|\calX)$ for which the $V$-shell of a sequence  of type $P$ in $\calX^n$ is not empty is denoted as $\calV_{n}(\calY;P)$.

For  information-theoretic quantities, we interchangeably use $H(P_X)$ and $H(X)$ to denote the entropy of a random variable $X$ with distribution $P_{X}$. Similarly, we interchangeably use $H(P_{X|Y}|P_Y)$  or $H(X|Y)$ to denote the conditional entropy. The mutual information and relative entropy are denoted in a similar manner.

\subsection{System Model and Problem Formulation}

We consider a streaming version of Slepian-Wolf problem, which is termed \emph{streaming Slepian-Wolf coding}. The standard Slepian-Wolf problem was solved in~\cite{slepian1973noiseless} and \cite{cover1975proof}. Similarly to~\cite{slepian1973noiseless}, we have two correlated sources to be separately compressed and jointly reconstructed. However, our streaming model differs from~\cite{slepian1973noiseless} in three main  aspects:
\begin{enumerate}
\item We have a {\em sequence} (countably infinite) of source blocks as inputs for each encoder;
\item We have a {\em sequence} (countably infinite) of encoders and decoders;
\item We incur a {\em  delay} at the decoder to decode a specific  source block pair.
\end{enumerate}
Our streaming setting is also different from \cite{draper2010lossless} in the sense that we allow encoders to have access to one more source {\em block} per unit time instead of one more {\em symbol} as in \cite{draper2010lossless}.

Consider a discrete memoryless source (DMS) with joint probability mass function (pmf) $P_{XY}$ on a finite alphabet $\calX\times\calY$. We define a streaming SW code formally as follows.

\begin{definition}[Streaming Code]
\label{def:streamingcode}
An $(n,N_1,N_2,T,\epsilon_n)$-code for streaming SW coding consists of
\begin{enumerate}
\item A sequence of correlated source blocks $\{(\bX_k,\bY_k)\}_{k\geq 1}$, where $(\bX_k,\bY_k)\in\calX^n\times\calY^n$ is an i.i.d.\ sequence with common distribution $P_{XY}\in\calP(\calX\times\calY)$;
\item A sequence of encoding function pairs $f_k:\mathcal{X}^{kn}\to \calM_1$ that maps the accumulated source blocks $\bX^k\in\calX^{kn}$ to a codeword $M_{1,k}\in\calM_1$ and  $g_k:\mathcal{Y}^{kn}\to \calM_2$ that maps the accumulated source blocks $\by^k\in\calY^{kn}$ to a codeword $M_{2,k}\in\calM_2$.
Here $|\calM_1|=N_1$ and $|\calM_2|=N_2$;
\item A sequence of decoding functions $\phi_k:\calM_1^{k+T-1}\times \calM_2^{k+T-1}\to\calX^n\times\calY^n$ that maps the accumulated codewords $\{(M_{1,j},M_{2,j})\}_{j=1}^{k+T-1}\in\calM_1^{k+T-1}\times\in\calM_2^{k+T-1}$ to a source block pair $(\hat{\bX}_k,\hat{\bY}_k)$,
\end{enumerate}
which satisfies 
\begin{align}
\sup_{k\in\bbN}\Pr\big((\hat{\bX}_k,\hat{\bY}_k)\neq (\bX_k,\bY_k)\big)\leq \epsilon_n\label{errorreq}.
\end{align}
Furthermore, we assume that common randomness is shared between the encoder and decoder.
\end{definition}

Let us say a few words about the availability of  common randomness at the encoder and decoder. This assumption ensures the existence of a deterministic code~\cite[Definition~1]{draper2010lossless}. We remark that, if common randomness is not present, the existence of deterministic codes can also be shown by changing the criterion in \eqref{errorreq} to  one of the following two forms:
\begin{enumerate}
	\item The average error criterion (averaged over infinitely many source blocks):
	\begin{align}
	\limsup_{N\to\infty}\frac{1}{N}\sum_{k=1}^N\Pr\big((\hat{\bX}_k,\hat{\bY}_k)\neq (\bX_k,\bY_k)\big)\leq \epsilon_n\label{errorreqa1}.
	\end{align}
	This error criterion was employed in the achievability work for streaming data transmission by Lee, Tan, and Khisti~\cite{lee2015streaming}. In \cite{lee2015streaming} it was asserted that deterministic codes exist under this criterion. 
	\item The maximum error criterion, similar  to \eqref{errorreq} but limiting the total number of source blocks (that are required to satisfy the error probability bound of $\epsilon_n$) to be  $L=L_n=\exp(o(n\xi_n^2))$, where $\xi_n = \omega\big(\sqrt{\frac{\log n}{n}} \big) \cap o(1)$ is a vanishing  sequence used to define the moderate deviations regime (cf.~Definition \ref{defmdcsw}). More precisely,
	\begin{align}
	\sup_{k \in [1:L]}\Pr\big((\hat{\bX}_k,\hat{\bY}_k)\neq(\bX_k,\bY_k)\big)\leq \epsilon_n\label{errorreqa2}.
	\end{align}
	Note that $L$ can be large and can also grow with the blocklength, so for practical applications \eqref{errorreqa2} is also useful. If we adopt~\eqref{errorreqa2}, we can assert the existence of deterministic codes by invoking the union bound and Markov's inequality as in \cite{hayashi2015asym}.\footnote{In more detail,  averaged over the random code code $\calC_n$, the error probabilities $\mathbb{E}_{\calC_n}[\Pr(\calE_k|\calC_n)]$ for each source block indexed by $k$ can be shown to decay as $\exp( - c n \xi_n^2)$ for some constant $c>0$ (this is, in fact, the moderate deviations constant). By the union bound and Markov's inequality, the probability, over the random code,  that any one of $N$ error events $\{ \Pr(\calE_k|\calC_n) > 2N e^{-cn\xi_n^2 }\}, k\in [1:L]$ occurs is $\le 1/2$. Thus, with probability   $>1/2$, there must exist a deterministic code, say $\calC_n^*$, satisfying $\Pr(\calE_k|\calC_n^*) \le 2N e^{-cn\xi_n^2 }\approx e^{-cn\xi_n^2 }$ for all $k \in [1:L]$.}  This error criterion was  also used by Lee, Tan, and Khisti in~\cite{LeeTanKhisti2016}. 
\end{enumerate} 

We remark that given fixed blocklength $n$, an $(n,N_1,N_2,T,\epsilon_n)$-code for streaming SW coding consists of a {\em sequence} of encoding and decoding functions.  Specifically, when $T=1$ and $k=1$ (only {\em one} source block is considered), we recover the standard SW coding scenario~\cite{slepian1973noiseless}.  To illustrate our streaming setup, we show in Figure~\ref{systemmodelsw} the case of $T=3$.

Further, we remark that the optimal rate region for our streaming setting coincides with the optimal rate region for the standard SW coding problem~\cite{slepian1973noiseless}. The converse proof of the optimal rate region for our streaming model follows by combining the techniques in \cite{slepian1973noiseless} and \cite[Chapter 5.7.1]{etezadi2015streaming}.

To define the next concept succinctly, let $R_X^*$ and $R_Y^*$ be two fixed rates. The rate pair $(R_X^*, R_Y^*)$  will be taken to be on the boundary of the optimal rate region for SW  coding~\cite{slepian1973noiseless} in Figure \ref{rateregion}. There are five cases in total, but Cases (iv) and (v) are symmetric to Cases (ii) and (i) respectively. Hence, in the presentation of definition and main result, we illustrate only the first three cases.

\begin{definition}[Moderate Deviations Constant]
\label{defmdcsw} 
{\em
Consider any correlated source with joint pmf $P_{XY}$ and any positive sequence $\{\xi_n\}_{n\in\bbN}$ satisfying $\xi_n\to 0$ and $\frac{\log n}{n\xi_n^2}\to 0$ as $n\to\infty$.\footnote{As an archetypical example, $\xi_n=n^{-t}$ for any $t\in (0,1/2)$ satisfies the two  conditions.} Let $\bm{\theta}=(\theta_1,\theta_2)\in\bbR^2$ be a real vector. A number $\nu$ is said to be a {\em $(R_X^*,R_Y^*,\bm{\theta},T)$-achievable moderate deviations constant} if there exists a sequence of $(n,N_1,N_2,T,\epsilon_n)$-codes such that 
\begin{align}
\label{ratemdc}
  \limsup_{n\to\infty} \frac{\log N_1-nR_X^*}{n\xi_n}&\leq \theta_1\\
  \limsup_{n\to\infty} \frac{\log N_2-nR_Y^*}{n\xi_n}&\leq \theta_2.
\end{align}
and
\begin{align}
\liminf_{n\to\infty} -\frac{\log\epsilon_n}{n\xi_n^2}\geq \nu.
\end{align}
The supremum of all $(R_X^*,R_Y^*,\bm{\theta},T)$-achievable moderate deviations constants is denoted as $\nu^*(R_X^*,R_Y^*,\bm{\theta},T)$.
}
\end{definition}

We remark that compared with \cite{draper2010lossless}, our goal to characterize $\nu^*(R_X^*,R_Y^*,\bm{\theta},T)$ is significantly different in the sense that i) the authors in~\cite{draper2010lossless} consider   large deviations while we consider moderate deviations; ii) their exponent that was derived in \cite{draper2010lossless} is with respect to the delay while our moderate deviations constant is with respect to the blocklength $n$ of each source block and the sequence $\xi_n$ which controls the speed of convergence of rates to a particular rate pair; and iii) we can show a quantitative improvement of  (at least) $T$  in the moderate deviations constant over the non-streaming SW setting.

We remark that in the above definition, in order to approach a {boundary} rate pair from a sequence of {non-boundary} rate pairs {\em inside}  (in the interior of) the SW coding region, we need to impose different conditions on $\bm{\theta}$ for the different cases as follows:
\begin{itemize}
\item Case (i):  The feasible set of $\bm{\theta}$ is
\begin{align}
\Theta_{(\mathrm{i})}:=\left\{\bm{\theta}\in\bbR^2:\theta_1>0,~-\infty<\theta_2<\infty\right\}.
\end{align}
\item Case (ii): The feasible set of $\bm{\theta}$ is
\begin{align}
\Theta_{(\mathrm{ii})}:=\left\{\bm{\theta}\in\bbR^2:\theta_1>0,~-\theta_1<\theta_2<\infty\right\}\label{def:Theta2}.
\end{align}
\item Case (iii): The feasible set of $\bm{\theta}$ is
\begin{align}
\Theta_{(\mathrm{iii})}
&:=\left\{\bm{\theta}\in\bbR^2:-\infty<\theta_1<\infty,~-\theta_1<\theta_2<\infty\right\}\\
&=\Theta_{(\mathrm{ii})}\bigcup \left\{\bm{\theta}\in\bbR^2:\theta_1<0,~-\theta_1<\theta_2<\infty\right\}.
\end{align}
\end{itemize}
The conditions on $\bm{\theta}$ for Cases (iv) and (v) are omitted since they are similar to the conditions for Cases (ii) and (i).

\subsection{Preliminaries and An Assumption}
\label{preliminary}

\begin{figure}[t]
\centering
\begin{picture}(115, 135)
\setlength{\unitlength}{.47mm}
\put(30,50){\circle*{4}}
\put(17,52){\footnotesize (ii)}
\put(19,66){\footnotesize (i)}
\put(30,68){\circle*{4}}
\put(40,40){\circle*{4}}
\put(42,40){\footnotesize (iii)}
\put(50,30){\circle*{4}}
\put(51,24){\footnotesize (iv)}
\put(68,30){\circle*{4}}
\put(69,24){\footnotesize (v)}
\input{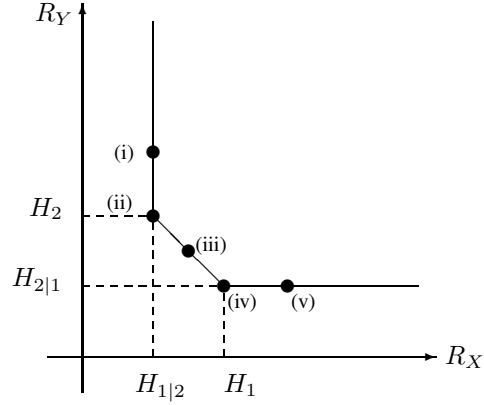}
\end{picture}
\caption{Illustration of the different cases in Definition \ref{defmdcsw}, Theorems \ref{swmdc} and \ref{streamingswmdc} where $H_1=H(P_X)$ and $H_{1|2}={H(P_{X|Y}|P_{Y})}$ etc. }
\label{rateregion}
\end{figure}

For a given source with pmf $P_{XY}$ on alphabet $\calX\times\calY$, the {\em joint and conditional source dispersions} or {\em varentropies} \cite{verdu14} are respectively defined as
\begin{align}
\rmV(P_{XY})&:=\sum_{ x,y  }P_{XY}(x,y)\left(-\log P_{XY}(x,y)-H(P_{XY})\right)^2,\label{defvarpxy}\\
\rmV(P_{X|Y}|P_{Y})&:=\sum_{ x,y  } P_{XY}(x,y)\left(-\log P_{X|Y}(x|y)-H(P_{X|Y}|P_Y)\right)^2\label{defvarpxgy},
\end{align}
and $\rmV(P_{Y|X}|P_X)$ is defined in a similar manner as $\rmV(P_{X|Y}|P_Y)$ with $X$ and $Y$ interchanged.  Note that in \eqref{defvarpxy} and \eqref{defvarpxgy}, it suffices to only sum over elements $(x,y)\in\calX\times\calY$ such that $P_{XY}(x,y)>0$.

We assume that for the correlated source with joint pmf $P_{XY}$ on the finite alphabet $\calX\times\calY$,  the three source dispersions
$\rmV(P_{XY})$, $\rmV(P_{X|Y}|P_Y)$, and $\rmV(P_{Y|X}|P_X)$ are positive.

\subsection{Moderate Deviations Asymptotics for Standard Slepian-Wolf Coding}

In this section, we present the moderate deviations constant for the (standard) Slepian-Wolf problem, i.e., the \emph{non-streaming case}. Let $(R_X^*,R_Y^*)$ be fixed as a boundary rate pair of the Slepian Wolf rate region in Figure \ref{rateregion}. Recall that when $T=1$ and $k=1$, the streaming setting in Definition \ref{def:streamingcode} reduces to the traditional Slepian-Wolf coding. Similarly as Definition \ref{defmdcsw}, we define the optimal moderate deviations constant for non-streaming SW coding and denote it as $\nu^*_{\mathrm{sw}}(R_X^*,R_Y^*,\bm{\theta})$.

\begin{theorem}[Non-Streaming Moderate Deviations Constant]
\label{swmdc}
Depending on $(R_X^*,R_Y^*)$, there are five cases, of which we present three here. The moderate deviations constant for non-streaming SW source coding is
\begin{enumerate}
\item Case (i): $R_X^*=H(P_{X|Y}|P_Y)$ and $R_Y^*>H(P_Y)$
\begin{align}
\nu^*_{\mathrm{sw}}(R_X^*,R_Y^*,\bm{\theta})=\frac{\theta_1^2}{2\rmV(P_{X|Y}|P_Y)}.
\end{align}
\item Case (ii): $R_X^*=H(P_{X|Y}|P_Y)$ and $R_Y^*=H(P_Y)$
\begin{align}
\nu^*_{\mathrm{sw}}(R_X^*,R_Y^*,\bm{\theta})=\min\left\{\frac{\theta_1^2}{2\rmV(P_{X|Y}|P_Y)},\frac{(\theta_1+\theta_2)^2}{2\rmV(P_{XY})}\right\}\label{caseiinonstreaming}
\end{align}
\item Case (iii): $R_X^*+R_Y^*=H(P_{XY})$, $H(P_{X|Y}|P_Y)<R_X^*<H(P_X)$ and $H(P_{Y|X}|P_X)<R_Y^*<H(P_Y)$
\begin{align}
\nu^*_{\mathrm{sw}}(R_X^*,R_Y^*,\bm{\theta})=\frac{(\theta_1+\theta_2)^2}{2\rmV(P_{XY})}.
\end{align}
\end{enumerate}
\end{theorem}

For lossless source coding with decoder side information, Chen \emph{et. al} \cite{chen2007redundancy} and He \emph{et. al}~\cite{he2009redundancy} derived the optimal moderate deviations constant. For SW coding, Hayashi and Matsumoto derived an achievability result for the moderate deviations constants similar to Theorem \ref{swmdc} in \cite[Lemma 89]{hayashi2016sw}. However, no corresponding converse result was proved for SW coding.

We remark that there are at least two ways to prove Theorem \ref{swmdc} for both achievability and converse parts. One way is to Taylor expand the error exponents for SW coding in a similar manner as in~\cite{altugwagner2014}. In Appendix \ref{appprelim}, we provide some preliminary lemmas for this calculation. The second way is to leverage information spectrum bounds \cite[Lemmas 7.2.1 and 7.2.2]{han2003information} and the moderate deviations theorem \cite[Theorem 3.7.1]{dembo2009large}. We believe that  the result in Theorem \ref{swmdc} (sans the achievability part~\cite{hayashi2016sw}) does not appear explicitly in previous works but it is straightforward and thus we omit its proof.

\subsection{Moderate Deviations Asymptotics for Streaming Slepian-Wolf Coding}
We now state the main result of this paper.
\begin{theorem}[Streaming Moderate Deviations Constant]
\label{streamingswmdc}
The moderate deviations constant for streaming SW coding satisfies
\begin{enumerate}
\item Case (i): $R_X^*=H(P_{X|Y}|P_Y)$ and $R_Y^*>H(P_Y)$
\begin{align}
\nu^*(R_X^*,R_Y^*,\bm{\theta},T)\geq \frac{T\theta_1^2}{2\rmV(P_{X|Y}|P_Y)}.
\end{align}
\item Case (ii): $R_X^*=H(P_{X|Y}|P_Y)$ and $R_Y^*=H(P_Y)$
\begin{align}
\nu^*(R_X^*,R_Y^*,\bm{\theta},T)\geq T\min\left\{\inf_{\gamma\in[0,1]}\frac{(\theta_1+(1-\gamma)\theta_2)^2}{2\left(\gamma \rmV(P_{X|Y}|P_Y)+(1-\gamma)\rmV(P_{XY})\right)},
\frac{(\theta_1+\theta_2)^2}{2\rmV(P_{XY})}\right\}\label{caseii}
\end{align}
\item Case (iii): $R_X^*+R_Y^*=H(P_{XY})$, $H(P_{X|Y}|P_Y)<R_X^*<H(P_X)$ and $H(P_{Y|X}|P_X)<R_Y^*<H(P_Y)$
\begin{align}
\nu^*(R_X^*,R_Y^*,\bm{\theta},T)\geq \frac{T(\theta_1+\theta_2)^2}{2\rmV(P_{XY})}.
\end{align}
\end{enumerate}
\end{theorem}
The proof of Theorem \ref{streamingswmdc} is provided in Section \ref{sec:swmdc}.  Several remarks are now in order. 
\begin{enumerate}

\item At a high level, our proof proceeds by combining the truncated memory idea for second-order analysis from \cite{lee2015streaming} and the so called minimum weighted empirical suffix entropy decoding idea from~\cite{draper2010lossless}. Subsequently, we Taylor expand the resultant  exponents at rates near the first-order fundamental limit and invoke some continuity arguments from \cite{polyanskiy2010thesis}. We remark that, unlike the moderate deviations analysis in \cite{lee2015streaming}, the truncated memory idea appears to be necessary. Otherwise,  the decoding error probability of a particular source block pair is upper bounded by the sum of the probabilities of $k$ ``dominant'' error events. In general, $k$ (the index of the source block we wish to decode) can be much larger than any exponential function of the blocklength so if truncation is not performed, this upper bound on the error probability would be vacuous as it would exceed one (cf. Section \ref{sec:needtrun} for a detailed explanation). However, if we adopt the maximum error criterion over $L = \exp(o(n\xi_n^2))$ source blocks (cf.~\eqref{errorreqa2}), then the truncated memory scheme can be shown to be no longer necessary. We believe that in order to extend the single-user streaming result  in \cite{lee2015streaming} to  multi-terminal  settings, ideas similar to using truncated memory encoding are required.

\item  The first term of the minimization in \eqref{caseii} is somewhat unusual so we comment on its significance here. This term results from the analysis of our streaming setup. Essentially, we perform standard random binning \cite{cover1975proof} and in the decoding procedure, several dominant error events result  from the use of the minimum weighted empirical suffix entropy decoder (see \eqref{suffixentropy}) due to the streaming scenario. More precisely, let $l$ be the first index where  the true source sequence $\bx^{T_k}$ differs from a competitor source sequence $\tilde{\bx}^{T_k}$. Similarly, let $m$ denote the first index where the true source sequence $\by^{T_k}$ differs from a competitor source sequence $\tilde{\by}^{T_k}$. The minimization over $\gamma\in[0,1]$ in \eqref{caseii} results from identifying the dominant error events over all feasible choices of $l$ and $m$. 

\item To achieve the moderate deviations constant in Theorem \ref{swmdc} using our coding scheme, we need the more stringent condition on the backoff sequence $\{\xi_n\}_{n\in\bbN}$, namely $\frac{n\xi_n^2}{\log n}\to\infty$ (compared to $n\xi_n^2\to\infty$ which is standard in moderate deviations analysis \cite{altugwagner2014, polyanskiy2010channel}).  However, we believe that this  condition cannot be easily relaxed in the current problem as well as other multi-terminal coding streaming problems. Furthermore when we use a truncated memory code for multi-terminal problems, we have roughly $O(n^{\frac{1}{2}})$ dominant error events, and thus  we need the additional logarithm to nullify these error events. This  stringent condition on $\{\xi_n\}_{n\in\bbN}$ \emph{cannot} be relaxed even if we employ an analogue of the (non-universal) maximum likelihood decoder instead of the minimum weighted empirical suffix entropy decoder.

\item We also discuss briefly on the difficulties faced to derive a matching converse result for streaming source coding problems. The only tight result (with matching achievability and converse) for streaming source coding problems was proved by Chang and Sahai in \cite{chang2006error}, where they derived the optimal delay exponent for symbolwise streaming point-to-point lossless source coding problem. In \cite{chang2006error}, they proved the achievability part by using a {\em fixed-to-variable} coding idea coupled with a ``first in, first out'' (FIFO) encoder. We consider fixed-to-fixed-length coding in this paper. Chang \emph{et al.}~\cite{chang2007price} made attempts to establish a converse result in the large deviations regime for lossless streaming source coding with decoder side information  using feedforward decoding. However, the derived bounds on optimal error exponents differ significantly (from the achievability) except for some very pathological sources. Adopting the idea of feedforward decoders to our setting~\cite{sahai2008block,chang2007streaming}, we can derive a converse moderate deviations result. However, because of the suboptimality of the bounds on the error exponents, the result turns out to yield a moderate deviations constant of infinity, which is vacuous.

\item In the streaming SW setting, for Cases (i) and (iii), we can  clearly achieve \emph{at least} $T$ times of the moderate deviations constant compared to the non-streaming case (see Theorem \ref{swmdc}). However, for Case (ii), which is most interesting, we cannot guarantee a gain of $T$ in moderate deviations constant  for all possible values of $\bm{\theta}\in\Theta_{\mathrm{(ii)}}$ (recall \eqref{def:Theta2}). In Proposition~\ref{ttime}, we show that for each correlated source $P_{XY}$, there exists a non-empty set of values of $\bm{\theta}\in\Theta_{\mathrm{(ii)}}$ such that we have at least a multiplicative gain of $T$ in the moderate deviations constant compared to Case (ii) of the non-streaming setting. Recall that $\bm{\theta}$ designates the direction of approach of a pair of rates towards the boundary of the SW region.

\end{enumerate}

For simplicity in notation, let the objective function in \eqref{caseii} be denoted as
\begin{align}
f(P_{XY},\gamma,\bm{\theta}):=\frac{(\theta_1+(1-\gamma)\theta_2)^2}{2\left(\gamma \rmV(P_{X|Y}|P_Y)+(1-\gamma)\rmV(P_{XY})\right)}.
\end{align}
Also define  the functions 
\begin{align}
g_1(P_{XY})&:=\frac{\rmV(P_{XY})-\rmV(P_{X|Y}|P_Y)}{2\rmV(P_{X|Y}|P_Y)},\label{def:g1} \qquad\mbox{and}\\
g_2(P_{XY})&:= \min\left\{\sqrt{\frac{\rmV(P_{XY})}{\rmV(P_{X|Y}|P_Y)}}-1,\frac{\rmV(P_{XY})-\rmV(P_{X|Y}|P_Y)}{\rmV(P_{XY})+\rmV(P_{X|Y}|P_Y)}\right\}\label{def:g2}.
\end{align}

We remark that $g_1(P_{XY})\geq g_2(P_{XY})$ for all sources $P_{XY}$ with $\rmV(P_{X|Y}|P_Y)>0$. This can be verified by considering the different relationships between $\rmV(P_{XY})$ and $\rmV(P_{X|Y}|P_Y)$. For the case in which $\rmV(P_{XY})=\rmV(P_{X|Y}|P_Y)$, we have $g_1(P_{XY})=g_2(P_{XY})=0$. When $\rmV(P_{XY})>\rmV(P_{X|Y}|P_Y)$, we have
\begin{align}
g_1(P_{XY})
> \frac{\rmV(P_{XY})-\rmV(P_{X|Y}|P_Y)}{\rmV(P_{XY})+\rmV(P_{X|Y}|P_Y)}\geq g_2(P_{XY}). \label{eqn:g1g2}
\end{align}
Lastly, when $\rmV(P_{XY})<\rmV(P_{X|Y}|P_Y)$, we have 
\begin{equation}
g_1(P_{XY})=\frac{1}{2}\left( \frac{ \rmV(P_{XY})}{ \rmV(P_{X|Y}|P_Y) } -1\right) >  \sqrt{\frac{\rmV(P_{XY})}{\rmV(P_{X|Y}|P_Y)}}-1\ge g_2(P_{XY}) \label{eqn:g1g22}
\end{equation}
where the strict inequality holds because $\frac{1}{2}(t-1) > \sqrt{t}-1$ for all $t\in [0,1)$ (due to the strict concavity of $t\mapsto\sqrt{t}-1$).  From the bounds in \eqref{eqn:g1g2} and \eqref{eqn:g1g22}, we also infer that for all sources $P_{XY}$, $g_1(P_{XY})=g_2(P_{XY})$ if and only if $\rmV(P_{XY})=\rmV(P_{X|Y}|P_Y)$.

We now state necessary and sufficient conditions on $\bm{\theta}\in\Theta_{\mathrm{(ii)}}$ (recall \eqref{def:Theta2}) for obtaining a multiplicative gain of $T$ in the moderate deviations constant under our proposed scheme.

\begin{proposition}[Conditions for Obtaining a Multiplicative Gain of $T$]
\label{ttime}
For each source distribution $P_{XY}$, the moderate deviations constant in the streaming setting is at least a factor of $T$ larger than the non-streaming counterpart under our proposed coding scheme, i.e.,
\begin{align}
\min\left\{\inf_{\gamma\in[0,1]}f(P_{XY},\gamma,\bm{\theta}),\frac{(\theta_1+\theta_2)^2}{2\rmV(P_{X|Y}|P_Y)}\right\}&=\eqref{caseiinonstreaming}\label{tgain},
\end{align}
if and only if $\bm{\theta}\in\Theta_{\mathrm{(ii)}}$ satisfies 
\begin{align}
\frac{\theta_2}{\theta_1}\geq g_1(P_{XY}),\label{condtheta1}
\end{align}
or
\begin{align}
-1<\frac{\theta_2}{\theta_1}\leq g_2(P_{XY})\label{condtheta2}.
\end{align}
\end{proposition}
The proof of Proposition \ref{ttime} is provided in Appendix \ref{proofttime}.

Note that for any source $P_{XY}$, $g_1(P_{XY})\in\bbR$. Thus, Proposition \ref{ttime}, and in particular \eqref{condtheta1}, allows us to conclude that for each correlated source, there exists a non-empty set of directions of approaching a boundary rate pair $(R_X^*,R_Y^*)$, parametrized by $\bm{\theta}$, such that we can achieve a multiplicative  gain of $T$ in the moderate deviations constant compared with Case (ii) of the non-streaming SW coding. In addition, from Proposition \ref{ttime} and the observation after \eqref{eqn:g1g22} (namely that  $g_1(P_{XY}) = g_2(P_{XY}) \Longleftrightarrow\rmV(P_{XY})=\rmV(P_{X|Y}|P_Y)$), we notice that for sources $P_{XY}$ such that  $\rmV(P_{XY}) \ne \rmV(P_{X|Y}|P_Y)$, we  conclude that there will also be a non-empty set of directions of approaching the boundary in which we {\em cannot}, in general, achieve a gain of $T$ in the streaming scenario (using the coding scheme and analyses delineated in Section \ref{sec:swmdc}).

One may conceive that the multiplicative gain of $T$ is achieved easily via a coding scheme for the vanilla SW problem with an effective blocklength $nT$. However, this is \emph{not} true because such a  scheme requires the encoders to have access to {\em future} source blocks\footnote{That is,  encoder $f_k$ (resp.\ $g_k$) has access to the first $T_k=T+k-1$
source blocks and compresses a superblock $\bX^{T_k}$ (resp. $\bY^{T_k}$) with blocklength $nT$ using $T$ codewords of rate $R$ and only one source block is decoded using these $T$ codewords.}  beyond the current time and is thus inconsistent with our streaming model (cf.\ Figure \ref{systemmodelsw} and Definition \ref{def:streamingcode}).

\section{Numerical Examples} \label{sec:num}

In this section, we present three sources satisfying the conditions as stated in Section \ref{preliminary} and illustrate conditions on $\bm{\theta}$ for which we can achieve a multiplicative gain of $T$ in moderate deviations constant compared to non-streaming setting. Throughout, we focus only on Case (ii) of Theorem \ref{streamingswmdc}.

\subsection{Doubly Symmetric Binary Source}

We consider the example of a doubly symmetric binary source (DSBS) where $\calX=\calY=\{0,1\}$, $P_{XY}(0,0)=P_{XY}(1,1)=\frac{1-p}{2}$ and $P_{XY}(0,1)=P_{XY}(1,0)=\frac{p}{2}$ for some $p\in\left[0,\frac{1}{2}\right]$. Note that $P_{X}(0)=P_Y(0)=\frac{1}{2}$, $P_{X|Y}(x|y)=1-p$ if $x=y$ and $P_{X|Y}(x|y)=p$ if $x
\neq y$. The joint and conditional entropies are
\begin{align}
H(P_{XY})&=2\left(\frac{1-p}{2}\log \frac{2}{1-p}+\frac{p}{2}\log \frac{2}{p}\right)=\log 2+h_{\rmb}(p),\\
H(P_{X|Y}|P_Y)&=h_{\rmb}(p).
\end{align}
where $h_{\rmb}(p)=-p\log p-(1-p)\log (1-p)$ is the binary entropy function.

The joint source dispersion is
\begin{align}
\rmV(P_{XY})
&=2\left(\frac{1-p}{2}\left(\log\frac{2}{1-p}-h_{\rmb}(p)-\log 2\right)^2+\frac{p}{2}\left(\log\frac{2}{p}-h_{\rmb}(p)-\log 2\right)^2\right)\\
&=(1-p)\left(-\log(1-p)-h_{\rmb}(p)\right)^2+p\left(-\log p-h_{\rmb}(p)\right)^2,
\end{align}
and similarly we obtain the conditional source dispersion
\begin{align}
\rmV(P_{X|Y}|P_Y)
&=(1-p)\left(-\log(1-p)-h_{\rmb}(p)\right)^2+p\left(-\log p-h_{\rmb}(p)\right)^2.
\end{align}
Hence, $\rmV(P_{XY})=\rmV(P_{X|Y}|P_Y)$ for this class of sources and we expect to achieve a gain of $T$. 

Indeed, using the definitions of $g_1$ in~\eqref{def:g1} and  $g_2$ in~\eqref{def:g2}, we obtain
\begin{align}
g_1(P_{XY}) =
g_2(P_{XY}) =0.
\end{align}
Now invoking Proposition~\ref{ttime} and \eqref{def:Theta2}, we conclude that for any DSBS, for all $\bm{\theta}\in\Theta_{\mathrm{(ii)}}$, we can achieve a multiplicative  gain of $T$ in the moderate deviations constant.

\begin{figure}[t]
\centering
\includegraphics[width=12cm]{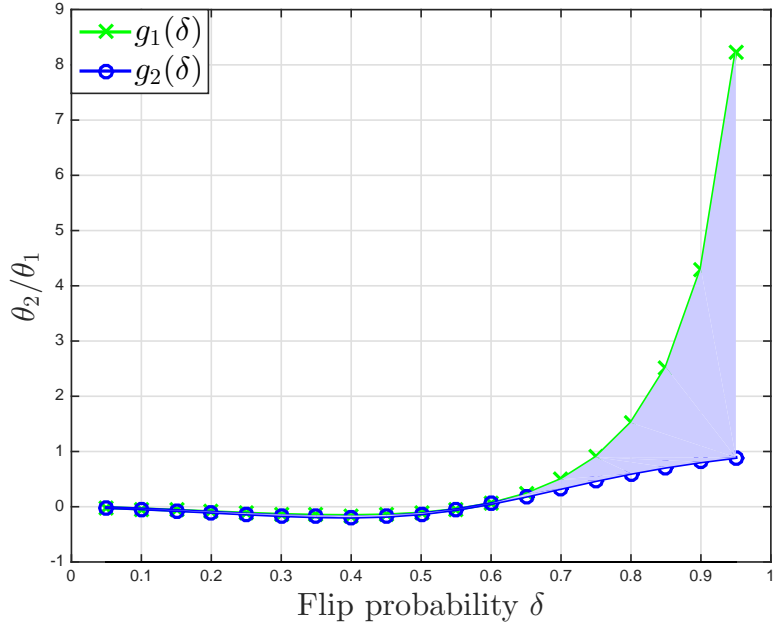}
\caption{Bounds on $\theta_2/\theta_1$ for the correlated binary source  generated by a Z-channel determined by $\delta$.  Proposition~\ref{ttime} says that the set of $\theta_2/\theta_1$ such that we obtain a multiplicative gain of $T$ in the moderate deviations constant are precisely contained in those areas  {\em excluding the shaded region}, i.e., $\theta_2/\theta_1$ above the green line ($\theta_2/\theta_1\geq g_1(\delta)$) or between the red line and the blue line ($-1<\theta_2/\theta_1\leq g_2(\delta)$). As the flip probability $\delta$ increases towards one (very noisy Z-channel), the gap between the green and blue lines increases.}
\label{gdelta}
\end{figure}

\subsection{Correlated Binary Source Generated via a Z-channel}

Let $\calX=\calY=\{0,1\}$ and $X\sim\mathrm{Bern}(\frac{1}{2})$. The random variable $Y$ is generated by transmitting $X$ through a Z-channel with flip probability $\delta\in(0,1)$~\cite[Figure 1.7]{palaiyanur2011impact}, i.e., $P_{Y|X}(0|0)=1$ and $P_{Y|X}(0|1)=\delta$. For this source, it can be verified that (using the convention that $0\log 0=0$)
\begin{align}
H(P_{XY})&=\log 2+\frac{1}{2} h_{\rmb}(\delta),\\
H(P_{X|Y}|P_Y)&=\frac{1+\delta}{2}h_{\rmb}\left(\frac{1}{1+\delta}\right),
\end{align}
and
\begin{align}
\rmV(P_{XY})&=\frac{1}{8}h_{\rmb}(\delta)^2+\frac{\delta}{2}\left( \log \frac{1}{\delta}-\frac{1}{2}h_{\rmb}(\delta)\right)^2+\frac{1-\delta}{2}\left( \log \frac{1}{1-\delta}-\frac{1}{2}h_{\rmb}(\delta)\right)^2,\\
\rmV(P_{X|Y}|P_Y)\nn&=\frac{1}{2}\left(\log (1+\delta)-\frac{1+\delta}{2}h_{\rmb}\left(\frac{1}{1+\delta}\right)\right)^2+\frac{\delta}{2}\left(\log\frac{1+\delta}{\delta}-\frac{1+\delta}{2}h_{\rmb}\left(\frac{1}{1+\delta}\right)\right)^2\\
&\qquad+\frac{(1+\delta)(1-\delta^2)}{8}\left(h_{\rmb}\left(\frac{1}{1+\delta}\right)\right)^2.
\end{align}
Note that $\rmV(P_{XY})$ and $\rmV(P_{X|Y}|P_Y)$ are functions of only $\delta$. Hence, we use $g_1(\delta)$ and $g_2(\delta)$ in place of  \eqref{def:g1} and \eqref{def:g2} for this correlated source.

Invoking Proposition~\ref{ttime}, we know that for any $\bm{\theta}\in\Theta_{\mathrm{(ii)}}$ (recall \eqref{def:Theta2}) such that $\frac{\theta_2}{\theta_1}\geq g_1(\delta)$ or $-1<\frac{\theta_2}{\theta_1}\leq g_2(\delta)$, we obtain a gain of $T$ in the moderate deviations constant compared with Case (ii) of  the non-streaming setting in Theorem \ref{swmdc}. In Figure \ref{gdelta}, we plot the two functions $g_1(\delta)$ and $g_2(\delta)$ against the flip probability $\delta$. From Figure \ref{gdelta}, we observe that  when $\delta=0.6$, we obtain a multiplicative gain of $T$ in the moderate deviations constant except when $\bm{\theta}\in\Theta_{\mathrm{(ii)}}$ is such that $0.056<\frac{\theta_2}{\theta_1}<0.058$ (which is a small interval).

\subsection{Asymmetric Correlated Binary Source }
In this subsection, we consider an asymmetric correlated binary sources where $\calX=\calY=\{0,1\}$, $P_{XY}(0,0)=1-3p$ and $P_{XY}(0,1)=P_{XY}(1,0)=P_{XY}(1,1)=p$ for some $p\in\left(0,\frac{1}{4}\right)\cup\left(\frac{1}{4},\frac{1}{3}\right)$. We do not allow $p$ to be $0$ (constant/degenerate distribution) or $1/4$ (uniform distribution), otherwise at least one of the joint and conditional source dispersions (varentropies) is zero, which is not permitted.

The joint and conditional entropies are 
\begin{align}
H(P_{XY})&=-(1-3p)\log (1-3p)-3p\log p,\\
H(P_{X|Y}|P_Y)&=(1-2p)h_{\rmb}\left(\frac{1-3p}{1-2p}\right)+2p \log 2.
\end{align}
The joint and conditional dispersions are
\begin{align}
\rmV(P_{XY})&=(1-3p)\left(-\log(1-3p)-H(P_{XY})\right)^2+3p\left(-\log p-H(P_{XY})\right)^2,\\
\nn\rmV(P_{X|Y}|P_Y)&=(1-3p)\left(\log\frac{1-2p}{1-3p}-H(P_{X|Y}|P_Y)\right)^2+p\left(\log\frac{1-2p}{p}-H(P_{X|Y}|P_Y)\right)^2\\
&\qquad+2p\left(\log2-H(P_{X|Y}|P_Y)\right)^2
\end{align}

Define $g_1(p),g_2(p)$ similarly as \eqref{def:g1} and \eqref{def:g2}. We plot these two functions $g_1(p)$ and $g_2(p)$ in Figure \ref{asymmetric}. Note that when $p=0.1$, for all $\bm{\theta}\in\Theta_{\mathrm{(ii)}}$ (recall \eqref{def:Theta2}) such that 
$\frac{\theta_2}{\theta_1}\in(-1,0.40]\cup[0.67,\infty)$, we can achieve a multiplicative gain of $T$ in the moderate deviations constant. When $\frac{\theta_2}{\theta_1} \in (0.40, 0.67)$ we cannot guarantee that we can achieve this gain using the proposed coding scheme. 

\begin{figure}[t]
\centering
\includegraphics[width=12cm]{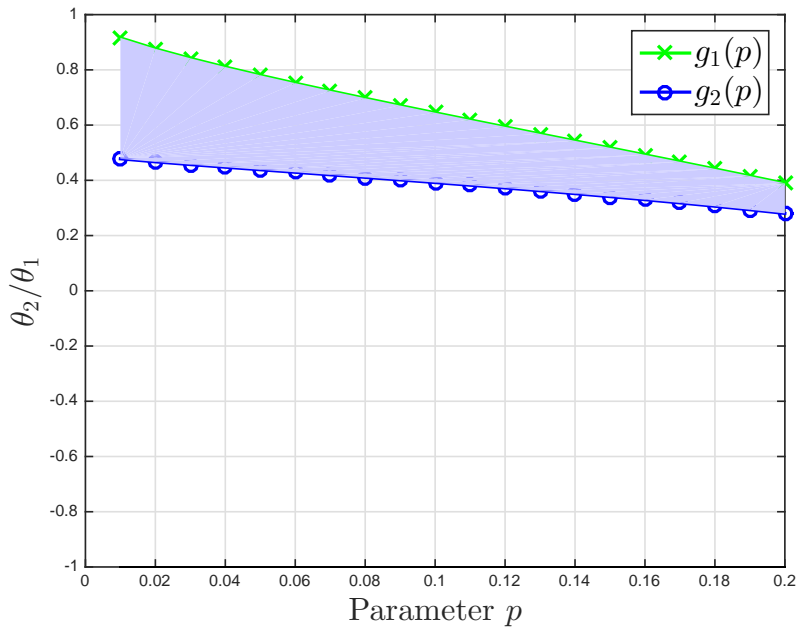}
\caption{Bounds on $\theta_2/\theta_1$ for the asymmetric correlated binary source  determined by $p$.  See the caption of Figure \ref{gdelta} for the set of $\theta_2/\theta_1$ such that we obtain a gain of $T$.}
\label{asymmetric}
\end{figure}

\section{Some Specializations and Generalizations} \label{sec:gen}
In this section, we specialize Theorem \ref{streamingswmdc} to point-to-point lossless source coding with/without decoder side information and we also generalize our result to the case where the delay requirements on both sources are different.

\subsection{Specialization to Streaming Lossless Source Coding with and without Decoder Side Information}

In this subsection, we first specialize the result in Theorem \ref{streamingswmdc} to the one encoder and one decoder with causal decoder side information setting. We term this as \emph{streaming lossless source coding with decoder side information.} An $(n,N_1,T,\epsilon_n)$-streaming code can be defined similarly as Definition \ref{def:streamingcode} by removing the encoding functions $g_k$ and replacing the decoding functions by $\phi_k:\calM_1^{k+T-1}\times\calY^{n(k+T-1)}\to\calX^n$. Each of these  decoding functions maps the accumulated codewords $\{ M_{1,j}\}_{j=1}^{k+T-1}$ and causal side information $\{\bY_j\}_{j=1}^{k+T-1}$ into a source block $\bX_k$. Similarly as Definition \ref{defmdcsw}, we can define the moderate deviations constant $\nu^*_{X|Y}(R_X^*,T)$ for this setting. Note that we let $\theta_1=1$ and ignore the requirement on the rate of second encoder.

For non-streaming  lossless source coding with decoder side information, it was shown in Chen {\em et al.}~\cite[Theorems~4 and~5]{chen2007redundancy} that the optimal moderate deviations constant is $( {2\rmV(P_{X|Y}|P_Y)})^{-1}$. By specializing Case (i) in Theorem \ref{streamingswmdc}, we obtain the following result, which says that under our setting in which we allow for a block delay of $T$, the optimal moderate deviations constant is enhanced by at least $T$.
\begin{corollary}
\label{streamingdsi}
The optimal  moderate deviations constant for streaming lossless source coding with decoder side information satisfies
\begin{align}
\nu^*_{X|Y}(H(P_{X|Y} | P_Y),T)\geq \frac{T}{2\rmV(P_{X|Y}|P_Y)}.
\end{align}
When the side information is not available, the optimal moderate deviations constant satisfies
\begin{equation}
\nu_X^*(H(P_X),T)\geq \frac{T}{2\rmV(P_X)}.
\end{equation}
\end{corollary}

Thus for streaming lossless source coding with or without decoder side information, we obtain a multiplicative gain of at least $T$ in the moderate deviations constant. We remark that Corollary~\ref{streamingdsi}  can be proved similarly as \cite[Theorem 1]{lee2015streaming} using the information spectrum method~\cite[Lemma 7.2.1]{han2003information} and the moderate deviations theorem~\cite[Theorem 3.7.1]{dembo2009large}. The proof requires that the decoder has knowledge of  the source distribution and thus it is \emph{non-universal}. However, the benefit is that we only require the usual condition on $\{\xi_n\}_{n\in\bbN}$ for the moderate deviations regime \cite{altugwagner2014}  (i.e., $n\xi_n^2\to\infty$ as $n\to\infty$)  instead of the more stringent condition that $\frac{n\xi_n^2}{\log n}\to \infty$ as $n\to\infty$ in Definition \ref{defmdcsw}.

\subsection{Generalization to Streaming with Different Delay Requirements}

In this subsection, we consider a generalization of Theorem \ref{streamingswmdc}. Recall that in Theorem \ref{streamingswmdc}, we require that  both source blocks are reconstructed with the same delay $T$. Now assume that we incur a delay  of $T_1$  blocks for source $\bX$ and a delay of $T_2$  blocks  for source $\bY$. We are now interested in characterizing the streaming performance in the moderate deviations regime. The formal definition of this setup is as follows.
\begin{definition}[Streaming Code with Different Delay Requirements]
An $(n,N_1,N_2,T_1,T_2,\epsilon_n)$-code for streaming SW coding consists of
\begin{enumerate}
\item A sequence of correlated source blocks $\{(\bX_k,\bY_k)\}_{k\geq 1}$, where $(\bX_k,\bY_k)\in\calX^n\times\calY^n$ is an i.i.d. sequence;
\item A sequence of encoding function pairs $f_k:\mathcal{X}^{kn}\to \calM_1$ that maps the accumulated source blocks $\bX^k\in\calX^{kn}$ to a codeword $M_{1,k}\in\calM_1$ and  $g_k:\mathcal{Y}^{kn}\to \calM_2$ that maps the accumulated source blocks $\by^k\in\calY^{kn}$ to a codeword $M_{2,k}\in\calM_2$ where $|\calM_i|=N_i,~i=1,2$.
\item A sequence of decoding functions $\phi_k:\calM_1^{k+T_1-1}\times \calM_2^{k+T_1-1}\to\calX^n\times\calY^n$ that maps the accumulated codewords $(M_1^{k+T_1-1},M_2^{k+T_1-1})\in\calM_1^{k+T_1-1}\times\calM_2^{k+T_1-1}$ to a source block pair $(\hat{\bX}_{k_1},\hat{\bY}_{k_2})$,
\end{enumerate}
which satisfies that for each $k$ satisfying $\max\{k_1,k_2\}\geq 1$,
\begin{align}
\Pr\big((\hat{\bX}_{k_1},\hat{\bY}_{k_2})\neq(\bX_{k_1},\bY_{k_2})\big)\leq \epsilon_n\label{tt2errorreq},
\end{align}
where $k_1=k$ and $k_2=k+T_1-T_2$.

When $k_1>0$, $k_2\leq 0$, \eqref{tt2errorreq} reduces to 
\begin{align}
\Pr\left(\hat{\bX}_{k_1}\neq \bX_{k_1}\right)\leq \epsilon_n, \label{eqn:X_only}
\end{align}
and similarly when $k_1\leq 0,k_2>0$, \eqref{tt2errorreq} reduces to 
\begin{equation}
\Pr\left(\hat{\bY}_{k_2}\neq \bY_{k_2}\right)\leq \epsilon_n.
\end{equation}
Similarly to Definition \ref{def:streamingcode}, we assume that common randomness is shared between the encoder and decoder.
\end{definition}

The optimal moderate deviations constant for this setting is defined similarly to that in Definition \ref{defmdcsw} and we denote it as  $\nu^*(R_X^*,R_Y^*,\bm{\theta},T_1,T_2)$.

\begin{theorem}
\label{swstreamingmdctt2}
The moderate deviations constant for streaming compression of correlated sources with two decoding delay requirements satisfies the same lower bounds as in Theorem \ref{streamingswmdc} with $T$ replaced by $\min\{T_1, T_2\}$. 
\end{theorem}
\begin{proof}
The proof of Theorem \ref{swstreamingmdctt2} follows by invoking Theorem \ref{streamingswmdc} twice. Suppose $T_1\leq T_2$. Then $k_2\leq k_1$. At each time $T_1+k_1-1$, we can first decode $(\bX_{k_2},\bY_{k_2})$ using codewords with indices from $1$ to $T_1+k_2-1$ and then decode $(\bX_{k_1},\bY_{k_1})$ using all the received $T_1+k_1-1$ codewords. An application of  Theorem \ref{streamingswmdc} yields that for each $k_1\geq 1$ and $k_2\geq 1$,
\begin{align}
\Pr\big((\hat{\bX}_{k_2},\hat{\bY}_{k_2})\neq (\bX_{k_2},\bY_{k_2})\big)
&\leq \exp\big(-n\xi_n^2\nu^*(R_X^*,R_Y^*,\bm{\theta},T_1)  + o(n\xi_n^2) \big),\\
\Pr\big((\hat{\bX}_{k_1},\hat{\bY}_{k_1})\neq (\bX_{k_1},\bY_{k_1})\big)
&\leq \exp\big(-n\xi_n^2\nu^*(R_X^*,R_Y^*,\bm{\theta},T_1)+ o(n\xi_n^2) \big).
\end{align}
Therefore, we have that for   for each $k_1\geq 1$ and $k_2\geq 1$,
\begin{align}
\Pr\big((\hat{\bX}_{k_1},\hat{\bY}_{k_2})\neq(\bX_{k_1},\bY_{k_2})\big)
&\leq \Pr\big((\hat{\bX}_{k_2},\hat{\bY}_{k_2})\neq (\bX_{k_2},\bY_{k_2})\big)+\Pr\big((\hat{\bX}_{k_1},\hat{\bY}_{k_1})\neq (\bX_{k_1},\bY_{k_1})\big)\\
&\leq 2\exp\big(-n\xi_n^2\nu^*(R_X^*,R_Y^*,\bm{\theta},T_1)+ o(n\xi_n^2) \big).
\end{align}
This completes the proof for the case $T_1\le T_2$. The argument for $T_2\leq T_1$ is completely analogous. 
\end{proof}

\section{Proof of Theorem \ref{streamingswmdc}}
\label{sec:swmdc}
In this section, we present the proof of Theorem \ref{streamingswmdc}. We combine the truncated memory encoding idea in \cite{lee2015streaming} and the minimum weighted empirical suffix entropy decoding in \cite{draper2010lossless} in the proof. The analysis is, however, rather \emph{different} from both works.

\subsection{Justification for the Use of Truncated Memory Encoding}
\label{sec:needtrun}

We explain why we choose to employ truncated memory coding (cf.~\cite{drapertruncated} and \cite[Theorem 2]{lee2015streaming})  for our problem. Let us first consider the point-to-point streaming source coding problem using  a simple tree code. Under this achievability scheme, when we aim to decode $\bX_k$, we actually need to decode $\bX_1,\ldots,\bX_k$ sequentially (cf.~\cite{wozencraft1957sequential,schulman,sahai2001anytime,sukhavasi2011,draper2010lossless,khisti2014dmt} and \cite[Theorem 1]{lee2015streaming}). The error in decoding any $\bX_j$ for $j=1,\ldots,k$ can potentially lead to a decoding error. Fortunately, the error probability of decoding $\bX_j$ with $j<k$ decreases exponentially fast and one can upper bound the total error probability as in \cite[Theorem~1]{lee2015streaming} to obtain the desired result. However, this is not the case for the streaming SW coding problem. Suppose a simple  tree code is used. When we aim to decode $(\bX_k,\bY_k)$, we then need to decode $(\bX_1,\bY_1),\ldots,(\bX_k,\bY_k)$ sequentially. The error in decoding can occur in any $(\bX_l,\bY_m)$ where $l,m\in[1:k]^2$, i.e., errors in the $\calX$-blocks and $\calY$-blocks are   {\em interleaved}. Let $p_{l,m}$ be the probability of decoding $(\bX_l,\bY_m)$ wrongly at time $k+T-1$. Then the error probability can be upper bounded as $\sum_{l=1}^k\sum_{m=1}^kp_{l,m}$. The inner sum $\sum_{m=1}^kp_{l,m}$ can be calculated in a similar manner as \cite{lee2015streaming} and shown to scale as $ \exp(-O(nT\xi_n^2\nu))$ for some $\nu>0$. When the maximum value of $k$ is fixed as in \cite{draper2010lossless}, the outer sum does not affect the exponent. However, in our setting, {\em any} $k\in\bbN$ (scaling with $n$) is allowed. Hence, the maximum value of $k$ is infinity and the outer sum is unbounded if we do not use the truncated memory idea to {\em limit the terms involved in outer sum}.

\subsection{Truncated Memory Encoding: Basic Idea and One Example}

In truncated memory encoding, we set a buffer to store source blocks at each encoder with maximum and minimum sizes which depend on the number of source blocks. Let the maximum and minimum sizes be denoted as $\Psi$ and $\Omega$ respectively. By setting these two values, we ensure that the maximum and minimum number of codewords that each source sequence is encoded into is $\Psi$ and  $\Omega$ (instead of infinity in a simple tree code) respectively. To achieve a multiplicative gain of $T$, under our coding scheme, we assume that $\Omega\geq T$ and $\Psi>2\Omega$ and we will ensure these two conditions hold in the sequel. The choices of $\Psi$ and $\Omega$ (as functions of $n$) form the crux of the   proof of Theorem \ref{streamingswmdc}.

The buffer for each encoder adheres to the following two rules:
\begin{enumerate}
\item At each time, only one new source block enters the buffer;
\item Once the buffer is full, in the next time slot, only the new source block and the most recent $\Omega-1$ source blocks remain in the buffer.
\end{enumerate}

To illustrate the idea behind truncated memory encoding, we first consider an example where the maximum and minimum sizes of the encoder buffer are set to be $\Psi=8$ and $\Omega=3$ respectively. We will describe the encoding procedure for encoders $\{f_k\}_{k\in\bbN}$ only since encoders $\{g_k\}_{k\in\bbN}$ operate similarly. At each time $k\in[1:\Psi=8]$, the buffer is not full and thus one new source block enters the buffer according to rule 1). Correspondingly, the encoder $f_k$ maps the first $k$ source blocks $\bX^k$ into a codeword $M_{1,k}$. At time $k=\Psi+1=9$, the buffer is full for the first time. According to rule 2), we keep the new source block and the most recent $\Omega-1=2$ source blocks in the encoder buffer, i.e., only $(\bX_7,\bX_8,\bX_9)$. Thus, the encoder $f_k$ maps $\bX_7^9$ into a codeword $M_{1,9}$. Subsequently, for $k\in[10:14]$, the encoder is not full and one new source block enters the buffer per time. Hence, the encoder $f_k$ maps $\bX_7^k$ in to a codeword $M_{1,k}$. For $k>14$, the encoding procedure continues by following the above two rules to select a subset of the accumulated source blocks to be encoded into codeword $M_{1,k}$. We illustrate the truncated memory encoding idea for this example in Figure \ref{encodingsw} for $k \in [1:20]$.

\begin{figure}[t]
\centering
\setlength{\unitlength}{0.5cm}
\begin{picture}(28,32)
\linethickness{1pt}
\input{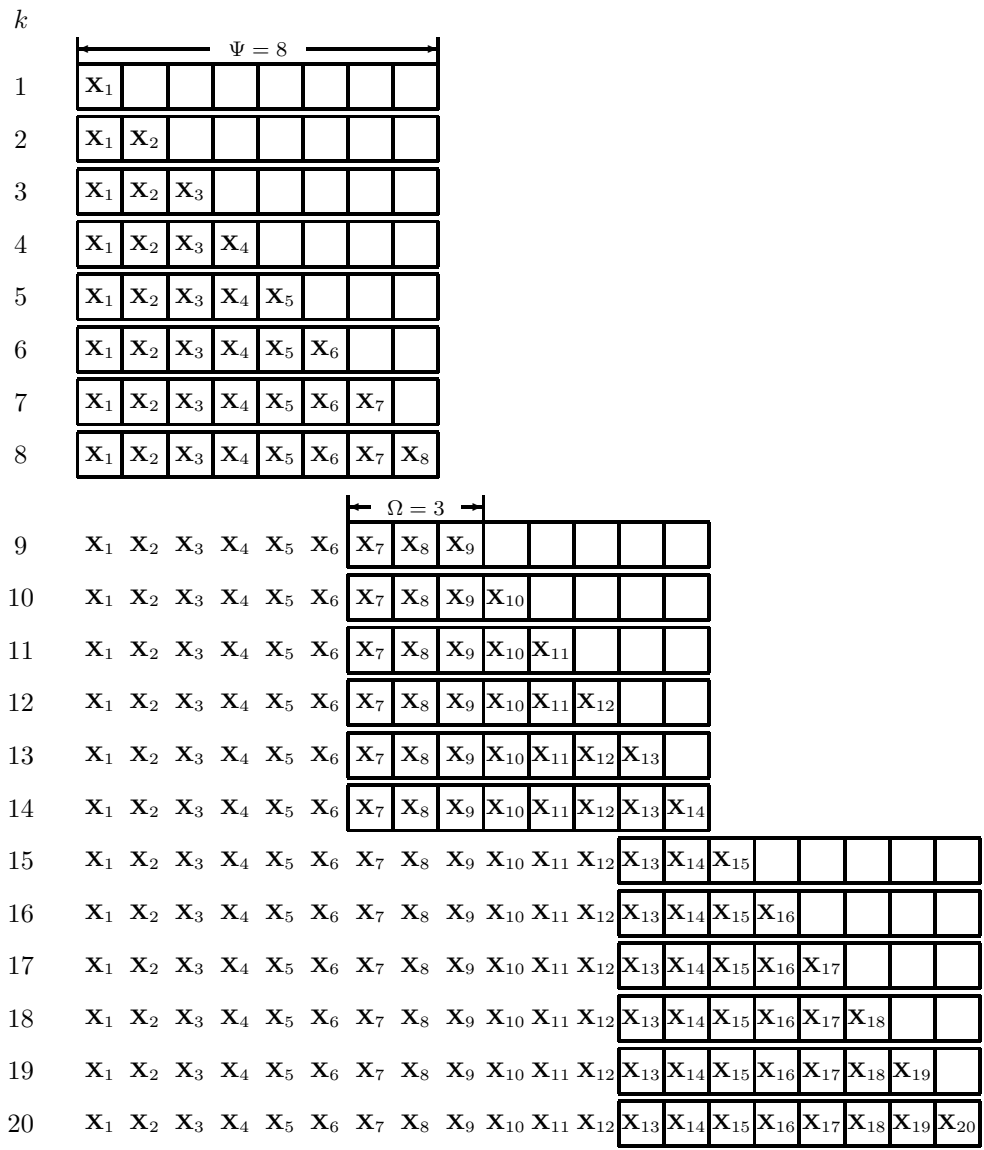}
\end{picture}
\caption{Illustration of truncated memory encoding for encoder $f_k$ of an example where $\Omega=3$ and $\Psi=8$. For any $k\in\bbN$, only the source blocks inside the encoder buffer (denoted by squares) will be encoded into the codeword $M_{1,k}$.}
\label{encodingsw}
\end{figure}

We state several observations from this numerical example (cf.\ Figure \ref{encodingsw}). We divide the encoding process  into two phases: the \emph{Initialization Phase} and the \emph{Periodic Phase}. In the initialization phase (i.e., $k\in[1:\Psi]$), the encoder $f_k$ encodes all the accumulated source blocks $\bX^k$ into a codeword $M_{1,k}$. In the periodic phase (i.e.,  $k>\Psi$), the encoder $f_k$ encodes only a subset of all available source blocks into a codeword $M_{1,k}$. In order to illustrate the rule for the  periodic phase and generalize the current example to arbitrary values of $\Psi$ and $\Omega$, we need the following definitions:
\begin{align}
\alpha_q&:=(\Psi-\Omega+1)q+\Omega,\label{def:alphaq}\\
\beta_q&:=(\Psi-\Omega+1)(q+1)+\Omega-1 \label{def:betaq},\\
t_q&:=\alpha_q-(\Omega-1)=(\Psi-\Omega+1)q+1\label{def:tq},\\
\calS(q)&:=\big[(\Psi-\Omega+1)q+\Omega : (\Psi-\Omega+1)(q+1)+\Omega-1 \big]\label{def:calsq}.
\end{align}
By specializing \eqref{def:alphaq}, \eqref{def:betaq}, and  \eqref{def:tq} with $\Psi=8$ and $\Omega=3$ in our example, we obtain that
\begin{align}
\alpha_q&=6q+3,\label{alpha:eg}\\
\beta_q&=6(q+1)+2,\label{beta:eg}\\
t_q&=6q+1\label{tq:eg},\\
\calS(q)&=\big[6q+3:6(q+1)+2\big]\label{calsq:eg}.
\end{align}
We also  state several observations for this encoding rule when $k>\Psi$. From Figure \ref{encodingsw}, for $k\in[\alpha_1=9:\beta_1=14]=\calS(1)$ (cf.\ \eqref{alpha:eg} to \eqref{calsq:eg}), the encoder $f_k$ encodes $\bX_{t_1=7}^k$ into a codeword $M_{1,k}$. In general, for $k\in\calS(q)$ with $q\in\bbN$, the encoder $f_k$ encodes $\bX_{t_q}^k$ into a codeword $M_{1,k}$. These observations pave the way for us to introduce the truncated memory encoding idea for arbitrary values of $\Psi$ and $\Omega$ in  the next section.

\subsection{Codebook Generation and Encoding}
\label{swencoding}

We now present the codebook generation and the truncated memory encoding for arbitrary values of $\Psi$ and $\Omega$. Again we only consider encoders $\{f_k\}_{k\in\bbN}$ since encoders $\{g_k\}_{k\in\bbN}$ operate similarly.

\textbf{Codebook Generation:} 
\begin{enumerate}
\item Initialization Phase: for $k\in[1:\Psi]$, we randomly and independently generate a codeword $M_{1,k}$ for each $\bx^k\in\calX^{nk}$ according to a uniform distribution over $\calM_1$. The codebook $\calC_k$ consists of all codewords $M_{1,k}$ for each $\bx^k\in\calX^{nk}$;
\item Periodic Phase: For each $k\in \calS(q)$, we randomly and independently generate a codeword $M_{1,k}$ for each $\bx_{t_q}^k\in\calX^{n(k-t_q+1)}$ according to a uniform distribution over $\calM_1$. The codebook $
\calC_k$ consists of all codewords $M_{1,k}$ for each $\bx_{t_q}^k\in\calX^{n(k-t_q+1)}$.
\end{enumerate}
We assume  that   the codebooks $\{\calC_k\}_{k\in\bbN}$ are known to the encoders and decoders.

\textbf{Encoding:}  
\begin{enumerate}
\item For $k\in[1:\Psi]$, given $\bx^k$, we send $f_k(\bx^k)=M_{1,k}$;
\item For $k\in\calS(q)$, since the output of $f_k$ depends only on $\bx_{t_q}^k$, we may define a function $\tilde{f}_k:\calX^{n(k-t_q+1)}\to\calM_1$ satisfying $\tilde{f}_k(\bx_{t_q}^k)=f_k(\bx^k)$. Similarly, we define $\tilde{g}_k(\by_{t_q}^k)=M_{2,k}$ for $k\in\calS(q)$.
\end{enumerate}
Throughout this section, we assume that source block pairs $\{(\bX_k,\bY_k)\}_{k\in\bbN}$ are generated and the corresponding codewords produced by the encoders are $\{(M_{1,k},M_{2,k})\}_{k\in\bbN}$.

\subsection{Decoding: Basic Idea and One Example}
\label{decode:example}

Our decoding strategy is similar to that used for the central limit analysis for streaming  channel coding \cite{lee2015streaming}  and is done in correspondence to the truncated memory encoding strategy. We will first consider an example where $T=2$, $\Omega=3$ and $\Psi=8$. For brevity, let
\begin{align}
T_k:=k+T-1\label{def:tk}.
\end{align}
Note that $T_k$ denotes the time to decode the source block pair $\{(\bX_k,\bY_k)\}$ when a delay of $T$ blocks is tolerable at the decoder $\phi_k$~(cf. Figure \ref{systemmodelsw}). Let us consider how to decode the source block pair $(X_{16},Y_{16})$ at time $T_{16}=17$ using the codeword pairs $\{(M_{1,\tau},M_{2,\tau})\}_{\tau\in[1:17]}$.

Recall the encoding procedure in Figure \ref{encodingsw}. Out of all available codeword pairs, only $\{(M_{1,\tau},M_{2,\tau})\}_{\tau\in\{16,17\}}$ are functions of the source block $(\bX_{16},\bY_{16})$. However, this does not mean that it suffices to use  only $\{(M_{1,\tau},M_{2,\tau})\}_{\tau\in\{16,17\}}$ to decode since these codeword pairs are also a function of the {\em past} source blocks $(\bX_{13}^{15},\bY_{13}^{15})$ and the future source block pair $(\bX_{17},\bY_{17})$. Since the future source block cannot lead to an error in decoding the current source block $(\bX_k,\bY_k)$ (cf.~\cite{lee2015streaming,LeeTanKhisti2016}), it suffices to first decode past source blocks $(\bX_{13}^{15},\bY_{13}^{15})$ in order to remove uncertainty in using $\{(M_{1,\tau},M_{2,\tau})\}_{\tau\in\{16,17\}}$ to decode the current source block pair. To decode $(\bX_{13},\bY_{13})$, we need to use codewords $\{(M_{1,\tau},M_{2,\tau})\}_{\tau\in[13:17]}$. We observe from Figure~\ref{encodingsw} that $\{(M_{1,\tau},M_{2,\tau})\}_{\tau\in\{13,14\}}$ is also a function of $(\bX_7^{12},\bY_7^{12})$. Hence, we need to backtrack and decode all source blocks $(\bX_7^{12},\bY_7^{12})$ in order to remove uncertainty of the past source blocks $(\bX_7^{12},\bY_7^{12})$ when using $\{(M_{1,\tau},M_{2,\tau})\}_{\tau\in[13:17]}$ to decode $(\bX_{13},\bY_{13})$.

To decode any source block $(\bX_k,\bY_k)$, one may potentially continue backtracking until the first source block pair $(\bX_1,\bY_1)$. However,  to reduce the ``complexity'' of decoding, we stop backtracking until we can conclude that we need to decode a source block pair with index $t_{q-1}$ (cf.\  \eqref{def:tq}) if the index of the source block we aim to decode lies in $\calS(q)$ for some integer $q\geq 2$. In our example, $k=16\in\calS(2)$, so we stop backtracking until we conclude that we need to decode $(\bX_7,\bY_7)$ since $t_1=7$.

To stop backtracking at the index $t_1=7$, we should only make use of the codewords $\{(M_{1,\tau},M_{2,\tau})\}_{\tau\in[\alpha_2=9:T_{16}=17]}$ in  decoding the source block $(\bX_{16},\bY_{16})$ since otherwise previous source blocks $(\bX^6,\bY^6)$ will cause uncertainty (i.e., a potential error) in our decoding of the past source blocks $(\bX_7^8,\bY_7^8)$ (cf.\ Figure \ref{encodingsw}). We illustrate the process of   backtracking to decode $(\bX_{16},\bY_{16})$ in Figure \ref{decodingsw}.

\begin{figure}[t]
\centering
\setlength{\unitlength}{0.5cm}
\begin{picture}(30,8)
\linethickness{1pt}
\input{sw_decoding}
\end{picture}
\caption{Illustration of the backtracking decoding idea to decode $(\bX_{16},\bY_{16})$ when $\Psi=8$, $\Omega=3$, and $T=2$.}
\label{decodingsw}
\end{figure}

After stating which source blocks to decode when we aim to decode $(\bX_{16},\bY_{16})$, we are now in a position to describe how the decoding is done. Prior to explaining the detail in decoding, we need the following additional notation. For a pair of sequences $(\bx,\by) \in\calX^r\times\calY^r$ (for some $r\in\bbN$), we let $\hatH(\bx,\by)$ be the joint empirical entropy, i.e., the entropy of the joint type $\hatT_{\bx,\by}$. Similarly, we let $\hatH(\bx|\by)= H( \hatT_{\bx|\by}|\hatT_\by)$ be the conditional empirical  entropy. Given two pairs of sequences $(\tilde{\bx}_a^b,\tilde{\by}_a^b)$ and $(\bar{\bx}_a^b,\bar{\by}_a^b)$, let $l$ be the index of the block where $\tilde{\bx}^{T_k}$ and $\bar{\bx}^{T_k}$ first differs and similarly let $m$ be the index of the block where $\tilde{\by}^{T_k}$ and $\bar{\by}^{T_k}$ first differs.

We make use of the weighted empirical suffix entropy in \cite[Eqn.~(54)]{draper2010lossless}, i.e., for $l,m\in[a:b]$,
\begin{align}
 \hat{H}_{\rmS}(l,m,\tilde{\bx}_a^b,\tilde{\by}_a^b)
 &=\left\{
\begin{array}{cc}
\hat{H}(\tilde{\bx}_l^b,\tilde{\by}_l^b) & l=m,\\
\\
\frac{m-l}{b-l+1}\hat{H}(\tilde{\bx}_l^{m-1}|\tilde{\by}_l^{m-1})+\frac{b-m+1}{b-l+1}\hat{H}(\tilde{\bx}_m^b,\tilde{\by}_m^b) &l<m,\\
\\
\frac{l-m}{b-l+1}\hat{H}(\tilde{\bx}_m^{l-1}|\tilde{\by}_m^{l-1})+\frac{b-l+1}{b-m+1}\hat{H}(\tilde{\bx}_l^b,\tilde{\by}_l^b) &l>m.
\end{array}
 \right.\label{suffixentropy}
\end{align}

We will now describe how to decode source block pairs $(\bX_{t_1=7}^{k=16},\bY_7^{16})$ sequentially. Recall the definitions of $\alpha_q$ in \eqref{def:alphaq}, $\beta_q$ in \eqref{def:betaq} and $t_q$ in \eqref{def:tq}. The detail of the decoding is presented as follows:
\begin{enumerate}
\item Decoding $(\bX_{t_1=7}^{\alpha_1=9},\bY_7^9)$ jointly using the codewords $\{(M_{1,\tau},M_{2,\tau})\}_{\tau\in[\alpha_1=9:\beta_1=14]}$.  Let the set of all source blocks which are binned to  $\{(M_{1,\tau},M_{2,\tau})\}_{\tau\in[\alpha_1=9:\beta_1=14]}$ be defined as
\begin{align}
\calA_1:=\{(\bar{\bx}_7^{14},\bar{\by}_7^{14}):f_\tau(\bar{\bx}_7^\tau)=M_{1,\tau},~g_\tau(\bar{\bx}_7^\tau)=M_{2,\tau},~\forall~\tau\in[9:14]\}\label{def:cala1eg}.
\end{align}
The estimate $(\hat{\bX}_7^9,\hat{\bY}_7^9)$ is chosen as $(\tilde{\bx}_7^9,\tilde{\by}_7^9)$ if $(\tilde{\bx}_7^{14},\tilde{\by}_7^{14})\in\calA_1$ satisfies that for all $(\bar{\bx}_7^{14},\bar{\by}_7^{14})$, we have
\begin{align}
\hat{H}_{\rmS}(l,m,\tilde{\bx}_7^{14},\tilde{\by}_7^{14})
 &\leq \hat{H}_{\rmS}(l,m,\bar{\bx}_7^{14},\bar{\by}_7^{14})\label{kin1psi},
\end{align}
where $l$ and $m$ in \eqref{kin1psi} are uniquely determined by $(\tilde{\bx}_7^{14},\tilde{\by}_7^{14})$ and $(\bar{\bx}_7^{14},\bar{\by}_7^{14})$ according to the rule as described prior to~\eqref{suffixentropy}. This rule for choosing $l$ and $m$ applies verbatim whenever we employ minimum weighted empirical suffix entropy decoding in the sequel.
\item For $j\in[\alpha_2+1=10:t_2-1=12]$, decode $(\bX_j,\bY_j)$ using codewords $\{(M_{1,\tau},M_{2,\tau})\}_{\tau\in[j:14]}$ and the previous estimates $(\hat{\bX}_7^{j-1},\hat{\bY}_7^{j-1})$. Similarly to~\eqref{def:cala1eg}, we define the set of source blocks which are binned to $\{(M_{1,\tau},M_{2,\tau})\}_{\tau\in[j:14]}$  as follows:
\begin{align}
\calA_2:=\{(\bar{\bx}_7^{14},\bar{\by}_7^{14}):f_\tau(\bar{\bx}_7^\tau)=M_{1,\tau},~g_\tau(\bar{\bx}_7^\tau)=M_{2,\tau},~\forall~\tau\in[j:14]\}\label{def:cala2eg}.
\end{align}
The estimate $(\hat{\bX}_j,\hat{\bY}_j)$ is chosen as $(\tilde{\bx}_j,\tilde{\by}_j)$ if $(\tilde{\bx}_7^{14},\tilde{\by}_7^{14})\in\calA_2$ satisfies 
\begin{enumerate}
\item $(\tilde{\bx}_7^{j-1},\tilde{\by}_7^{j-1})=(\hat{\bX}_7^{j-1},\hat{\bY}_7^{j-1})$;
\item for all $(\bar{\bx}_7^{14},\bar{\by}_7^{14})\in\calA_2$ satisfying $(\bar{\bx}_7^{j-1},\bar{\by}_7^{j-1})=(\hat{\bX}_7^{j-1},\hat{\bY}_7^{j-1})$
\begin{align}
\hat{H}_{\rmS}(l,m,\tilde{\bx}_7^{14},\tilde{\by}_7^{14})
 &\leq \hat{H}_{\rmS}(l,m,\bar{\bx}_7^{14},\bar{\by}_7^{14}).
\end{align}
\end{enumerate}
The method to decode $(\bX_j,\bY_j)$ for $j\in[13:16]$ is similar to the current case and thus omitted for simplicity.
\item For $j\in[t_2=13:\beta_1=14]$, decode $(\bX_j,\bY_j)$ using  $\{(M_{1,\tau},M_{2,\tau})\}_{\tau\in[j:T_{16}=17]}$ and the previous estimates $(\hat{\bX}_7^{j-1},\hat{\bY}_7^{j-1})$;
\item For $j\in[\alpha_2=15:k=16]$, decode $(\bX_j,\bY_j)$ using   $\{(M_{1,\tau},M_{2,\tau})\}_{\tau\in[j:T_{16}=17]}$ and the previous estimates $(\hat{\bX}_{13}^{j-1},\hat{\bY}_{13}^{j-1})$.
\end{enumerate}
In the next subsection, we generalize the decoding rule introduced in this particular example to arbitrary values of $\Psi$, $\Omega$ and $T$ such that $\Psi>2\Omega$ and $\Omega\geq T$.

\subsection{Decoding: General Case}
\label{decode:general}

For simplicity, we only consider decoding $(\bX_k,\bY_k)$ at time $T_k$ (cf. \eqref{def:tk}) for some $k\in\calS(q)$ (cf. \eqref{def:calsq}) with $q\geq 2$   since other cases can be done similarly. Readers may refer to our arxiv preprint (cf.~\cite[version 1]{zhou2016moderate}) for a complete description of the decoding details for all $k\in\bbN$.

By generalizing the decoding rule presented in the example in Section \ref{decode:example}, we conclude that in order to decode $(\bX_k,\bY_k)$, under our coding scheme, we should sequentially decode source block pairs $(\bX_{t_{q-1}}^k,\bY_{t_{q-1}}^k)$. To be specific, we first jointly decode $(\bX_{t_{q-1}}^{\alpha_{q-1}},\bY_{t_{q-1}}^{\alpha_{q-1}})$. Then we sequentially decode $(\bX_j,\bY_j)$ for $j\in\calS(q-1)$. Finally, we sequentially decode $(\bX_j,\bY_j)$ for $j\leq k$ and $j\in\calS(q)$. Recall the notation used in the codebook generation and truncated memory encoding in Section~\ref{swencoding}. The details on how to decode each source block pair(s) are presented as follows.

\begin{enumerate}

\item \label{joint}Joint decoding of $\bX_{t_{q-1}}^{\alpha_{q-1}}$: Define the set of source blocks in the bin indexed by the codewords as
\begin{align}
\!\!\calB_1(\bX_{t_{q-1}}^{\beta_{q-1}},\bY_{t_{q-1}}^{\beta_{q-1}})
:=\left\{(\tilde{\bx}_{t_{q-1}}^{\beta_{q-1}},\tilde{\by}_{t_{q-1}}^{\beta_{q-1}}):\forall \tau\in[\alpha_{q-1}:\beta_{q-1}],~\tilde{f}_\tau(\tilde{\bx}_{t_{q-1}}^\tau)=\tilde{f}_\tau(\bX_{t_{q-1}}^\tau),~\tilde{g}_\tau(\tilde{\by}_{t_{q-1}}^\tau)=\tilde{g}_\tau(\bY_{t_{q-1}}^\tau)\right\}\label{def:b1}.
\end{align}

We declare $(\bX_{t_{q-1}}^{\alpha_{q-1}},\bY_{t_{q-1}}^{\alpha_{q-1}})=(\tilde{\bx}_{t_{q-1}}^{\alpha_{q-1}},\tilde{\by}_{t_{q-1}}^{\alpha_{q-1}})$
if there exists $(\tilde{\bx}_{t_{q-1}}^{\beta_{q-1}},\tilde{\by}_{t_{q-1}}^{\beta_{q-1}})\in
\calB_1(\bX_{t_{q-1}}^{\beta_{q-1}},\bY_{t_{q-1}}^{\beta_{q-1}})$ such that for all $(\bar{\bx}_{t_{q-1}}^{\beta_{q-1}},\bar{\by}_{t_{q-1}}^{\beta_{q-1}})\in
\calB_1(\bX_{t_{q-1}}^{\beta_{q-1}},\bY_{t_{q-1}}^{\beta_{q-1}})$, we have
\begin{align}
\hat{H}_{\rmS}(l,m,\tilde{\bx}_{t_{q-1}}^{\beta_{q-1}},\tilde{\by}_{t_{q-1}}^{\beta_{q-1}})\leq \hat{H}_{\rmS}(l,m,\bar{\bx}_{t_{q-1}}^{\beta_{q-1}},\bar{\by}_{t_{q-1}}^{\beta_{q-1}})\label{def:db1}.
\end{align}

\item \label{seq1} Decode $(\bX_j,\bY_j)$ for each $j\in[\alpha_{q-1}+1:\beta_{q-1}]$ sequentially, i.e., the remaining source blocks with indices in $\calS(q-1)$.

Let $\lambda_k=\min\{T_k,\beta_q\}$. The decoding rule  differs  depending  on $j$.
\begin{enumerate}
\item $\alpha_{q-1}+1\leq j\leq t_q-1$: Define 
\begin{align}
\calB_2(\bX_{t_{q-1}}^{\beta_{q-1}},\bY_{t_{q-1}}^{\beta_{q-1}},j)
:=\left\{(\tilde{\bx}_{t_{q-1}}^{\beta_{q-1}},\tilde{\by}_{t_{q-1}}^{\beta_{q-1}}):\forall \tau\in[j:\beta_{q-1}],~\tilde{f}_\tau(\tilde{\bx}_{t_{q-1}}^\tau)=\tilde{f}_\tau(\bX_{t_{q-1}}^\tau),~\tilde{g}_\tau(\tilde{\by}_{t_{q-1}}^\tau)=\tilde{g}_\tau(\bY_{t_{q-1}}^\tau)\right\}\label{def:b2}.
\end{align}

Given $(\hat{\bX}_{t_{q-1}}^{j-1},\hat{\bY}_{t_{q-1}}^{j-1})$, we declare $(\hat{\bX}_j,\hat{\bY}_j)=(\tilde{\bx}_j,\tilde{\by}_j)$ if there exists  $(\tilde{\bx}_{t_{q-1}}^{\beta_{q-1}},\tilde{\by}_{t_{q-1}}^{\beta_{q-1}})\in\calB_2(\bX_{t_{q-1}}^{\beta_{q-1}},\bY_{t_{q-1}}^{\beta_{q-1}},j)
$ satisfying $(\tilde{\bx}_{t_{q-1}}^{j-1},\tilde{\by}_{t_{q-1}}^{j-1})=(\hat{\bX}_{t_{q-1}}^{j-1},\hat{\bY}_{t_{q-1}}^{j-1})$ such that for all $(\bar{\bx}_{t_{q-1}}^{\beta_{q-1}},\bar{\by}_{t_{q-1}}^{\beta_{q-1}})\in\calB_2(\bX_{t_{q-1}}^{\beta_{q-1}},\bY_{t_{q-1}}^{\beta_{q-1}},j)
$ satisfying $(\bar{\bx}_{t_{q-1}}^{j-1},\bar{\by}_{t_{q-1}}^{j-1})=(\hat{\bX}_{t_{q-1}}^{j-1},\hat{\bY}_{t_{q-1}}^{j-1})$, we have
\begin{align}
\hat{H}_{\rmS}(l,m,\tilde{\bx}_{t_{q-1}}^{\beta_{q-1}},\tilde{\by}_{t_{q-1}}^{\beta_{q-1}})\leq \hat{H}_{\rmS}(l,m,\bar{\bx}_{t_{q-1}}^{\beta_{q-1}},\bar{\by}_{t_{q-1}}^{\beta_{q-1}})\label{def:db2}.
\end{align}

\item $t_q\leq j\leq \beta_{q-1}$: Define 
\begin{align}
\calB_3(\bX_{t_{q-1}}^{\lambda_k},\bY_{t_{q-1}}^{\lambda_k},j)
\nn:=\Big\{(\tilde{\bx}_{t_{q-1}}^{\lambda_k},\tilde{\by}_{t_{q-1}}^{\lambda_k}):&\forall \tau\in[j:\beta_{q-1}],~\tilde{f}_\tau(\tilde{\bx}_{t_{q-1}}^\tau)=\tilde{f}_\tau(\bX_{t_{q-1}}^\tau),~\tilde{g}_\tau(\tilde{\by}_{t_{q-1}}^\tau)=\tilde{g}_\tau(\bY_{t_{q-1}}^\tau),\\
&\forall \tau\in[\alpha_q:\lambda_k],~\tilde{f}_\tau(\tilde{\bx}_{t_q}^\tau)=\tilde{f}_\tau(\bX_{t_q}^{\tau}),~\tilde{g}_\tau(\tilde{\by}_{t_q}^\tau)=\tilde{g}_\tau(\bY_{t_q}^\tau)\Big\}\label{defb3}.
\end{align}

Given $(\hat{\bX}_{t_{q-1}}^{j-1},\hat{\bY}_{t_{q-1}}^{j-1})$, we declare $(\hat{\bX}_j,\hat{\bY}_j)=(\tilde{\bx}_j,\tilde{\by}_j)$ if there exists  $(\tilde{\bx}_{t_{q-1}}^{\lambda_k},\tilde{\by}_{t_{q-1}}^{\lambda_k})\in\calB_3(\bX_{t_{q-1}}^{\lambda_k},\bY_{t_{q-1}}^{\lambda_k},j)
$ satisfying $(\tilde{\bx}_{t_{q-1}}^{j-1},\tilde{\by}_{t_{q-1}}^{j-1})=(\hat{\bX}_{t_{q-1}}^{j-1},\hat{\bY}_{t_{q-1}}^{j-1})$ such that for all $(\bar{\bx}_{t_{q-1}}^{\lambda_k},\bar{\by}_{t_{q-1}}^{\lambda_k})\in\calB_3(\bX_{t_{q-1}}^{\lambda_k},\bY_{t_{q-1}}^{\lambda_k},j)
$ where $(\bar{\bx}_{t_{q-1}}^{j-1},\bar{\by}_{t_{q-1}}^{j-1})=(\hat{\bX}_{t_{q-1}}^{j-1},\hat{\bY}_{t_{q-1}}^{j-1})$, we have
\begin{align}
\hat{H}_{\rmS}(l,m,\tilde{\bx}_{t_{q-1}}^{\lambda_k},\tilde{\by}_{t_{q-1}}^{\lambda_k})\leq \hat{H}_{\rmS}(l,m,\bar{\bx}_{t_{q-1}}^{\lambda_k},\bar{\by}_{t_{q-1}}^{\lambda_k})\label{def:db3}.
\end{align}
\end{enumerate}

\item Decode $(\bX_j,\bY_j)$ for $j\in[\alpha_q,k]$ sequentially, i.e., the source blocks with indices in $\calS(q)$ which are no larger than $k$. We consider two scenarios which differ in the use of the codewords.
\begin{enumerate}
\item \label{case1}$T_k\leq \beta_q$: Define 
\begin{align}
\calB_4(\bX_{t_q}^{T_k},\bY_{t_q}^{T_k},j)
:=\left\{(\tilde{\bx}_{t_q}^{T_k},\tilde{\by}_{t_q}^{T_k}):\forall \tau\in[j:T_k],~\tilde{f}_\tau(\tilde{\bx}_{t_q}^\tau)=\tilde{f}_\tau(\bX_{t_q}^\tau),~\tilde{g}_\tau(\tilde{\by}_{t_q}^\tau)=\tilde{g}_\tau(\bY_{t_q}^\tau)\right\}\label{def:b4}.
\end{align}

Given $(\hat{\bX}_{t_{q-1}}^{j-1},\hat{\bY}_{t_{q-1}}^{j-1})$, we declare $(\hat{\bX}_j,\hat{\bY}_j)=(\tilde{\bx}_j,\tilde{\by}_j)$ if there exists $(\tilde{\bx}_{t_q}^{T_k},\tilde{\by}_{t_q}^{T_k})\in\calB_4(\bX_{t_q}^{T_k},\bY_{t_q}^{T_k},j)
$ satisfying $(\tilde{\bx}_{t_q}^{j-1},\tilde{\by}_{t_q}^{j-1})=(\hat{\bX}_{t_q}^{j-1},\hat{\bY}_{t_q}^{j-1})$ such that for all $(\bar{\bx}_{t_q}^{T_k},\bar{\by}_{t_q}^{T_k})\in\calB_4(\bX_{t_q}^{T_k},\bY_{t_q}^{T_k},j)
$ where $(\bar{\bx}_{t_q}^{j-1},\bar{\by}_{t_q}^{j-1})=(\hat{\bX}_{t_q}^{j-1},\hat{\bY}_{t_q}^{j-1})$, we have
\begin{align}
\hat{H}_{\rmS}(l,m,\tilde{\bx}_{t_q}^{T_k},\tilde{\by}_{t_q}^{T_k})\leq \hat{H}_{\rmS}(l,m,\bar{\bx}_{t_q}^{T_k},\bar{\by}_{t_q}^{T_k})\label{def:db4}.
\end{align}

\item $T_k>\beta_q$: Depending on the index $j$ of the source block pair to decode, the decoding rules differ.
\begin{enumerate}
\item $\alpha_q \leq j\leq t_{q+1}-1$: Define 
\begin{align}
\calB_5(\bX_{t_q}^{\beta_q},\bY_{t_q}^{\beta_q},j)
:=\left\{(\tilde{\bx}_{t_q}^{\beta_q},\tilde{\by}_{t_q}^{\beta_q}):\forall \tau\in[j:\beta_q],~\tilde{f}_\tau(\tilde{\bx}_{t_q}^\tau)=\tilde{f}_\tau(\bX_{t_q}^\tau),~\tilde{g}_\tau(\tilde{\by}_{t_q}^\tau)=\tilde{g}_\tau(\bY_{t_q}^\tau)\right\}\label{defb5}.
\end{align}

Given $(\hat{\bX}_{t_{q-1}}^{j-1},\hat{\bY}_{t_{q-1}}^{j-1})$, we declare $(\hat{\bX}_j,\hat{\bY}_j)=(\tilde{\bx}_j,\tilde{\by}_j)$ if there exists  $(\tilde{\bx}_{t_q}^{\beta_q},\tilde{\by}_{t_q}^{\beta_q})\in\calB_5(\bX_{t_q}^{\beta_q},\bY_{t_q}^{\beta_q},j)
$ satisfying $(\tilde{\bx}_{t_q}^{j-1},\tilde{\by}_{t_q}^{j-1})=(\hat{\bX}_{t_q}^{j-1},\hat{\bY}_{t_q}^{j-1})$ such that for all $(\bar{\bx}_{t_q}^{\beta_q},\bar{\by}_{t_q}^{\beta_q})\in\calB_5(\bX_{t_q}^{\beta_q},\bY_{t_q}^{\beta_q},j)
$ where $(\bar{\bx}_{t_q}^{j-1},\bar{\by}_{t_q}^{j-1})=(\hat{\bX}_{t_q}^{j-1},\hat{\bY}_{t_q}^{j-1})$, we have
\begin{align}
\hat{H}_{\rmS}(l,m,\tilde{\bx}_{t_q}^{\beta_q},\tilde{\by}_{t_q}^{\beta_q})\leq \hat{H}_{\rmS}(l,m,\bar{\bx}_{t_q}^{\beta_q},\bar{\by}_{t_q}^{\beta_q})\label{def:db5}.
\end{align}

\item $t_{q+1}\leq j\leq k$: Define 
\begin{align}
\calB_6(\bX_{t_q}^{T_k},\bY_{t_q}^{T_k},j)
\nn:=\Big\{(\tilde{\bx}_{t_q}^{T_k},\tilde{\by}_{t_q}^{T_k}):&\forall \tau\in[j:\beta_q],~\tilde{f}_\tau(\tilde{\bx}_{t_q}^\tau)=\tilde{f}_\tau(\bX_{t_q}^\tau),~\tilde{g}_\tau(\tilde{\by}_{t_q}^\tau)=\tilde{g}_\tau(\bY_{t_q}^\tau)\\
&\forall \tau\in[\alpha_{q+1}:T_k],~\tilde{f}_\tau(\tilde{\bx}_{t_{q+1}}^\tau)=\tilde{f}_\tau(\bX_{t_{q+1}}^\tau),~\tilde{g}_\tau(\tilde{\by}_{t_{q+1}}^\tau)=\tilde{g}_\tau(\bY_{t_{q+1}}^\tau)\Big\}\label{defb6}.
\end{align}

Given $(\hat{\bX}_{t_{q-1}}^{j-1},\hat{\bY}_{t_{q-1}}^{j-1})$, we declare $(\hat{\bX}_j,\hat{\bY}_j)=(\tilde{\bx}_j,\tilde{\by}_j)$ if there exists  $(\tilde{\bx}_{t_q}^{T_k},\tilde{\by}_{t_q}^{T_k})\in\calB_6(\bX_{t_q}^{T_k},\bY_{t_q}^{T_k},j)
$ satisfying $(\tilde{\bx}_{t_q}^{j-1},\tilde{\by}_{t_q}^{j-1})=(\hat{\bX}_{t_q}^{j-1},\hat{\bY}_{t_q}^{j-1})$ such that for all $(\bar{\bx}_{t_q}^{T_k},\bar{\by}_{t_q}^{T_k})\in\calB_6(\bX_{t_q}^{T_k},\bY_{t_q}^{T_k},j)
$ where $(\bar{\bx}_{t_q}^{j-1},\bar{\by}_{t_q}^{j-1})=(\hat{\bX}_{t_q}^{j-1},\hat{\bY}_{t_q}^{j-1})$, we have
\begin{align}
\hat{H}_{\rmS}(l,m,\tilde{\bx}_{t_q}^{T_k},\tilde{\by}_{t_q}^{T_k})\leq \hat{H}_{\rmS}(l,m,\bar{\bx}_{t_q}^{T_k},\bar{\by}_{t_q}^{T_k})\label{def:db6}.
\end{align}
\end{enumerate}
\end{enumerate}
\end{enumerate}

\subsection{Error Events}
\label{analysisep}

In this subsection, we present the error events in our coding scheme. We consider the case where $q\geq 2$, $k\in\calS(q)$ and $T_k>\beta_q$ only since the analyses for other cases can be done similarly.

For subsequent analyses, define the set
\begin{align}
\calF(l,m,\bx_a^b,\by_a^b)
:&=\left\{(\tilde{\bx}_a^b,\tilde{\by}_a^b)\in\calX^{n(b-a+1)}\times\calY^{n(b-a+1)}:~\tilde{\bx}_a^{l-1}=\bx_a^{l-1},~\tilde{\bx}_l\neq \bx_l,~\tilde{\by}_a^{m-1}=\by_a^{m-1},~\tilde{\by}_m\neq \by_m\right\}\label{firstdiffer}.
\end{align}

Note that $\calF(l,m,\bx_a^b,\by_a^b)$ is a collection of source blocks $(\tilde{\bx}_a^b,\tilde{\by}_a^b)$ such that $\bx_a^b$ and $\tilde{\bx}_a^b$ differs first from $l$-th block and $\by_a^b$ and $\tilde{\by}_a^b$ differs first from $m$-th block. Hence, we have that for any $(\bx_a^b,\by_a^b)$,
\begin{align}
\calX^{n(b-a+1)}\times\calY^{n(b-a+1)}=\bigcup_{l=a}^b\bigcup_{m=a}^b 
\calF(l,m,\bx_a^b,\by_a^b).
\end{align}

The error in decoding $(\bX_k,\bY_k)$ occurs if one of the following events occur:
\begin{enumerate}
\item $l,m\in[t_{q-1}:\beta_{q-1}]$ and $\min\{l,m\}\leq \alpha_{q-1}$
\begin{align}
\calE_{\mathrm{sw}}^1(l,m)
\nn&:=\Bigg\{\exists (\tilde{\bx}_{t_{q-1}}^{\beta_{q-1}},\tilde{\by}_{t_{q-1}}^{\beta_{q-1}})\in\calB_1(\bX_{t_{q-1}}^{\beta_{q-1}},\bY_{t_{q-1}}^{\beta_{q-1}})\bigcap\calF(l,m,\bX_{t_{q-1}}^{\beta_{q-1}},\bY_{t_{q-1}}^{\beta_{q-1}}):\\
&\qquad\qquad\qquad\qquad\hat{H}_{\rmS}(l,m,\tilde{\bx}_{t_{q-1}}^{\beta_{q-1}},\tilde{\by}_{t_{q-1}}^{\beta_{q-1}})\leq \hat{H}_{\rmS}(l,m,\bX_{t_{q-1}}^{\beta_{q-1}},\bY_{t_{q-1}}^{\beta_{q-1}})\Bigg\}\label{calsw1}.
\end{align}
These error events correspond to the joint decoding of $(\bX_{t_{q-1}}^{\alpha_{q-1}},\bY_{t_{q-1}}^{\alpha_{q-1}})$ (cf. \eqref{def:b1} and \eqref{def:db1}).

\item $j\in[\alpha_{q-1}+1:t_q-1]$, $l,m\in[j:\beta_{q-1}]$ and $\min\{l,m\}=j$.
\begin{align}
\calE_{\mathrm{sw},j}^2(l,m)
\nn&:=\Bigg\{\exists (\tilde{\bx}_{t_{q-1}}^{\beta_{q-1}},\tilde{\by}_{t_{q-1}}^{\beta_{q-1}})\in\calB_2 (\bX_{t_{q-1}}^{\beta_{q-1}},\bY_{t_{q-1}}^{\beta_{q-1}})\bigcap\calF(l,m,\bX_{t_{q-1}}^{\beta_{q-1}},\bY_{t_{q-1}}^{\beta_{q-1}}):\\
&\qquad\qquad\qquad\qquad\hat{H}_{\rmS}(l,m,\tilde{\bx}_{t_{q-1}}^{\beta_{q-1}},\tilde{\by}_{t_{q-1}}^{\beta_{q-1}})\leq \hat{H}_{\rmS}(l,m,\bX_{t_{q-1}}^{\beta_{q-1}},\bY_{t_{q-1}}^{\beta_{q-1}})\Bigg\}\label{calsw2}.
\end{align}
These error events correspond to the sequential decoding of $(\bX_j,\bY_j)$ for $j\in[\alpha_{q-1}+1:t_q-1]$ (cf. \eqref{def:b2} and \eqref{def:db2}).

\item $j\in[t_q:\beta_{q-1}]$, $l,m\in[j:\beta_q]$ and $\min\{l,m\}=j$.
\begin{align}
\calE_{\mathrm{sw},j}^3(l,m)
\nn&:=\Bigg\{\exists (\tilde{\bx}_{t_{q-1}}^{\beta_q},\tilde{\by}_{t_{q-1}}^{\beta_q})\in\calB_3(\bX_{t_{q-1}}^{\beta_q},\bY_{t_{q-1}}^{\beta_q})\bigcap\calF(l,m,\bX_{t_{q-1}}^{\beta_q},\bY_{t_{q-1}}^{\beta_q}):\\
&\qquad\qquad\qquad\qquad\hat{H}_{\rmS}(l,m,\tilde{\bx}_{t_{q-1}}^{\beta_q},\tilde{\by}_{t_{q-1}}^{\beta_q})\leq \hat{H}_{\rmS}(l,m,\bX_{t_{q-1}}^{\beta_q},\bY_{t_{q-1}}^{\beta_q})\Bigg\}\label{calsw3}.
\end{align}

These error events correspond to the sequential decoding of $(\bX_j,\bY_j)$ for $j\in[t_q:\beta_{q-1}]$ (cf. \eqref{defb3} and \eqref{def:db3}).

\item $j\in[\alpha_q:t_{q+1}-1]$, $l,m\in[j:\beta_q]$ and $\min\{l,m\}=j$.
\begin{align}
\calE_{\mathrm{sw},j}^5(l,m)
\nn&:=\Bigg\{\exists (\tilde{\bx}_{t_q}^{\beta_q},\tilde{\by}_{t_q}^{\beta_q})\in\calB_5(\bX_{t_q}^{\beta_q},\bY_{t_q}^{\beta_q},j)\bigcap\calF(l,m,\bX_{t_q}^{\beta_q},\bY_{t_q}^{\beta_q}):\\
&\qquad\qquad\qquad\qquad\hat{H}_{\rmS}(l,m,\tilde{\bx}_{t_q}^{\beta_q},\tilde{\by}_{t_q}^{\beta_q})\leq \hat{H}_{\rmS}(l,m,\bX_{t_q}^{\beta_q},\bY_{t_q}^{\beta_q})\Bigg\}\label{calsw5}.
\end{align}
These error events correspond to the sequential decoding of $(\bX_j,\bY_j)$ for $j\in[\alpha_q:t_{q+1}-1]$ (cf. \eqref{defb5} and \eqref{def:db5}).

\item $j\in[t_{q+1}:k]$, $l,m\in[j:T_k]$ and $\min\{l,m\}=j$.

\begin{align}
\calE_{\mathrm{sw},j}^6(l,m)
\nn&:=\Bigg\{\exists (\tilde{\bx}_{t_q}^{T_k},\tilde{\by}_{t_q}^{T_k})\in\calB_6(\bX_{t_q}^{T_k},\bY_{t_q}^{T_k},j)\bigcap\calF(l,m,\bX_{t_q}^{T_k},\bY_{t_q}^{T_k}):\\
&\qquad\qquad\qquad\qquad\hat{H}_{\rmS}(l,m,\tilde{\bx}_{t_q}^{T_k},\tilde{\by}_{t_q}^{T_k})\leq \hat{H}_{\rmS}(l,m,\bX_{t_q}^{T_k},\bY_{t_q}^{T_k})\Bigg\}\label{calsw6}.
\end{align}
These error events correspond to the sequential decoding of $(\bX_j,\bY_j)$ for $j\in[t_{q+1}:k]$ (cf. \eqref{defb6} and \eqref{def:db6}).
\end{enumerate}

\subsection{Preliminaries for the Evaluation of the Error Probabilities}
\label{parametervalues}

In this subsection, we present some preliminaries which will be used in the analysis of the probabilities of the error events. 

We first recall some important quantities that are used in the derivation of the error exponent in~\cite{draper2010lossless} for distributed streaming compression, i.e., $E_X\left(R_X,R_Y,\gamma\right)$ and $E_Y\left(R_X,R_Y,\gamma\right)$ (cf. \cite[Eqn.~(29)]{draper2010lossless}). We recap the expression in both the Gallager~\cite{gallagerIT} and Csisz\'ar-K\"orner~\cite{csiszar2011information} forms. In the proof in Section \ref{asympep}, we use the fact that these two forms for the error exponent are equivalent~(cf. \cite[Lemma 5]{draper2010lossless}). Define
\begin{align}
  E_X(R_X,R_Y,\gamma)
  &:=\max_{\rho\in[0,1]} \Big[\gamma E_{X|Y}(R_X,\rho)+(1-\gamma)E_{XY}(R_X,R_Y,\rho)\Big]\label{gammax}\\
  \nn&:=\inf_{Q_{XY},\tilde{Q}_{XY}} \gamma D(Q_{XY}\|P_{XY})+(1-\gamma)D(\tilde{Q}_{XY}\|P_{XY})\\
  &\qquad+\left|\gamma\left(R_X-H(Q_{X|Y}|Q_Y)\right)+(1-\gamma)\left(R_X+R_Y-H(\tilde{Q}_{XY})\right)\right|^+\label{gammax2},
\end{align}
where the {\em Gallager functions} are
\begin{align}
E_{X|Y}(R_X,\rho)
&:=\rho R_X-\log\sum_{y}P_{Y}(y)\left(\sum_{x}P_{X|Y}(x|y)^{\frac{1}{1+\rho}}\right)^{1+\rho},\\
E_{XY}(R_X,R_Y,\rho)
&:=\rho(R_X+R_Y)-(1+\rho)\log\sum_{x,y}P_{XY}(x,y)^{\frac{1}{1+\rho}},
\end{align}
and $E_Y(R_X,R_Y,\gamma)$ is similar to $E_X(R_X,R_Y,\gamma)
$ with $X$ and $Y$ interchanged.

Now, we present the definitions of the coding rates. 
Let $(R_X^*,R_Y^*)$ be a   rate pair on the boundary of the SW region (see Cases (i)-(v) in Figure \ref{rateregion}).
We choose $(n,N_1,N_2)$ such that
\begin{align}   
R_{X,n}:=\frac{1}{n}\log N_1&=R_X^*+\theta_1\xi_n\label{m1size}\\
R_{Y,n}:=\frac{1}{n}\log N_2&=R_Y^*+\theta_2\xi_n\label{m2size}.
\end{align}
We remark that this choice satisfies the conditions in Definition \ref{defmdcsw}. 

Additionally, define the rates 
\begin{align}
R_{X,n}^l&: =\frac{\Psi-\Omega+1}{n(\beta_{q-1}-l+1)}\log N_1,\label{rxnl}\\*
R_{Y,n}^m&: =\frac{\Psi-\Omega+1}{n(\beta_{q-1}-m+1)}\log N_2\label{rynm}.
\end{align}
We remark that $R_{X,n}^l$ and $R_{Y,n}^m$ are only used in the analysis of the probability of the  error event  $\calE_{\mathrm{sw}}^1(l,m)$ (cf.~\eqref{calsw1} and the upper bounding in the steps leading to~\eqref{upcalesw1}) where $l,m\in[t_{q-1}:\alpha_{q-1}]$.

Finally, we present an alternative form of the weighted empirical suffix entropy in~\eqref{suffixentropy} in terms of types and conditional types. Given $P_1,P_2\in\calP(\calY)$ and  $V_1 ,V_2\in\calP(\calX|\calY)$, when $l\leq m$, define
\begin{align}
H_{\rmS}(P_1,P_2,V_1,V_2):=(m-l)H(V_1|P_1)+(\beta_{q-1}-m+1)H(P_2\times V_2).
\end{align}

Using the definition of the weighted empirical suffix entropy in~\eqref{suffixentropy}, we obtain that for $a\leq l\leq m\leq b$,
\begin{align}
H_{\rmS}(P_1,P_2,V_1,V_2)=\hatH_{\rmS}(l,m,\tilde{\bx}_a^b,\tilde{\by}_a^b)\label{suffixtype}.
\end{align}
holds if $\tilde{\by}_l^{m-1}\in\calT_{P_1}$, $\tilde{\by}_m^b\in\calT_{P_2}$, $\tilde{\bx}_l^{m-1}\in\calT_{V_1}(\tilde{\by}_l^{m-1})$ and $\tilde{\bx}_m^b\in\calT_{V_2}(\tilde{\by}_m^b)$ for some types $P_1\in\calP_{n(m-l)}(\calY)$ and $P_2\in\calP_{n(b-m+1)}(\calY)$ and some conditional types $V_1\in\calV_{n(m-l)}(\calX;P_1)$  and $V_2\in\calV_{n(b-m+1)}(\calX;P_2)$. Recall that given  an $n$-type $P\in\calP_n(\calY)$,   $\calV_n(\calX; P):= \{V \in \calP( \calX|\calY): \calT_V(y^n)\ne\emptyset\;\mbox{for some}\; y^n\in\calT_P\}$ is the set of all conditional types ``compatible with'' $P$.  We use the facts that $|\calP_n(\calY)|\le (n+1)^{|\calY|}$ and $|\calV_n(\calX; P)|\le (n+1)^{ |\calX| |\calY|}$ extensively in the following.

\subsection{Evaluation of the Error Probabilities}
\label{asympep}
We are now ready to upper bound the probability of each error event using the definitions in Section \ref{parametervalues}.

\begin{enumerate}
\item $\calE_{\mathrm{sw}}^1(l,m)$ defined in~\eqref{calsw1}:
\begin{enumerate}
\item $l,m\in[t_{q-1}:\alpha_{q-1}]$ and $l\leq m$.

We can upper bound the probability of $\calE_{\mathrm{sw}}^1(l,m)$ as follows:
\begin{align}
&\nn\Pr\left(\calE_{\mathrm{sw}}^1(l,m)\right)\\
&\leq \sum_{(\bx_{t_{q-1}}^{\beta_{q-1}},\by_{t_{q-1}}^{\beta_{q-1}})}P_{XY}^{n\Psi}(\bx_{t_{q-1}}^{\beta_{q-1}},\by_{t_{q-1}}^{\beta_{q-1}})\min\left\{1,\Pr\left(\calE_{\mathrm{sw}}^1(l,m)|(\bX_{t_{q-1}}^{\beta_{q-1}},\bY_{t_{q-1}}^{\beta_{q-1}})=(\bx_{t_{q-1}}^{\beta_{q-1}},\by_{t_{q-1}}^{\beta_{q-1}})\right)\right\}\label{upperstep1},
\end{align}
where $\beta_{q-1}-t_{q-1}+1=\Psi$ by invoking the definitions of $\beta_q$ \eqref{def:betaq} and $t_q$ in \eqref{def:tq}.

Define 
\begin{align}
q(\bx_{t_{q-1}}^{\beta_{q-1}},\by_{t_{q-1}}^{\beta_{q-1}})):=\Pr\left(\calE_{\mathrm{sw}}^1(l,m)|(\bX_{t_{q-1}}^{\beta_{q-1}},\bY_{t_{q-1}}^{\beta_{q-1}})=(\bx_{t_{q-1}}^{\beta_{q-1}},\by_{t_{q-1}}^{\beta_{q-1}})\right).
\end{align}

Then, given types $P_1\in\calP_{n(m-l)}(\calY)$, $P_2\in\calP_{n(\beta_{q-1}-m+1)}(\calY)$ and conditional types, $V_1\in\calV_{n(m-l)}(\calX;P_1)$, $V_2\in\calV_{n(\beta_{q-1}-m+1)}(\calX;P_2)$, if $\by_l^{m-1}\in\calT_{P_1}$, $\by_m^{\beta_{q-1}}\in\calT_{P_2}$, $\bx_l^{m-1}\in\calT_{V_1}(\by_l^{m-1})$ and $\bx_m^{\beta_{q-1}}\in\calT_{V_2}(\by_m^{\beta_{q-1}})$, we obtain
\begin{align}
q(\bx_{t_{q-1}}^{\beta_{q-1}},\by_{t_{q-1}}^{\beta_{q-1}}))
&\leq 
\sum_{\substack{(\tilde{\bx}_{t_{q-1}}^{\beta_{q-1}},\tilde{\by}_{t_{q-1}}^{\beta_{q-1}})\in\calF(l,m,\bx_{t_{q-1}}^{\beta_{q-1}},\by_{t_{q-1}}^{\beta_{q-1}}):\\\hat{H}_{\rmS}(l,m,\tilde{\bx}_{t_{q-1}}^{\beta_{q-1}},\tilde{\by}_{t_{q-1}}^{\beta_{q-1}})\leq \hat{H}_{\rmS}(l,m,\bx_{t_{q-1}}^{\beta_{q-1}},\by_{t_{q-1}}^{\beta_{q-1}})}}\!\!\!\!\!\!\!\!\!\!\!\!\!\!
\Pr\left((\tilde{\bx}_{t_{q-1}}^{\beta_{q-1}},\tilde{\by}_{t_{q-1}}^{\beta_{q-1}})\in\calB_1(\bx_{t_{q-1}}^{\beta_{q-1}},\by_{t_{q-1}}^{\beta_{q-1}})\right)\\
&\leq 
\sum_{\substack{(\tilde{\bx}_{t_{q-1}}^{\beta_{q-1}},\tilde{\by}_{t_{q-1}}^{\beta_{q-1}})\in\calF(l,m,\bx_{t_{q-1}}^{\beta_{q-1}},\by_{t_{q-1}}^{\beta_{q-1}}):\\\hat{H}_{\rmS}(l,m,\tilde{\bx}_{t_{q-1}}^{\beta_{q-1}},\tilde{\by}_{t_{q-1}}^{\beta_{q-1}})\leq \hat{H}_{\rmS}(l,m,\bx_{t_{q-1}}^{\beta_{q-1}},\by_{t_{q-1}}^{\beta_{q-1}})}}\!\!\!\!\frac{1}{(N_1N_2)^{\Psi-\Omega+1}}\label{randombinning}\\
&=\sum_{\substack{\tilde{P}_2\in\calP_{n(\beta_{q-1}-m+1)}(\calY),\\\tilde{V}_1\in\calV_{n(m-l)}(\calX;P_1),\\\tilde{V_2}\in\calV_{n(\beta_{q-1}-m+1)}(\calX;\tilde{P_2}):\\H_{\rmS}(P_1,\tilde{P}_2,\tilde{V}_1,\tilde{V_2})\leq H_{\rmS}(P_1,P_2,V_1,V_2)}}\sum_{\tilde{y}_m^{\beta_{q-1}}\in\calT_{\tilde{P}_2}}\sum_{\substack{\tilde{\bx}_l^{m-1}\in\calP_{\tilde{V}_1}(\by_l^{m-1})\\\tilde{\bx}_m^{\beta_{q-1}}\in\calP_{\tilde{V}_2}(\tilde{\by}_m^{\beta_{q-1}})}}\frac{1}{(N_1N_2)^{\Psi-\Omega+1}}\label{usesetf}\\
&\leq \frac{ (n(\beta_{q-1}-l+1)+1)^{3|\calX||\calY|}\exp(n(\beta_{q-1}-l+1)H_{\rmS}(P_1,P_2,V_1,V_2))}{(N_1N_2)^{\Psi-\Omega+1}}\label{draper8287},
\end{align}
where \eqref{randombinning} follows because in the joint decoding of $(\bX_{t_{q-1}}^{\alpha_{q-1}},\bY_{t_{q-1}}^{\alpha_{q-1}})$, we make use of $\beta_{q-1}-\alpha_{q-1}+1=\Psi-\Omega+1$ binning codewords (see decoding procedure in \eqref{def:b1} and \eqref{def:db1}, corresponding error events in \eqref{calsw1} and the definitions of $\alpha_q$ and $\beta_q$ in \eqref{def:alphaq} and \eqref{def:betaq} respectively); \eqref{usesetf} follows by invoking the definition of $\calF(l,m,\bx_{t_{q-1}}^{\beta_{q-1}},\by_{t_{q-1}}^{\beta_{q-1}})$ in \eqref{firstdiffer}, expressing the sum of sequences as sum over (conditional) type classes and noting the equivalent expression for the score function in terms of types and conditional types in \eqref{suffixtype}; \eqref{draper8287} follows from calculations involving types~\cite{csiszar2011information}  and the condition $H_{\rmS}(P_1,\tilde{P}_2,\tilde{V}_1,\tilde{V_2})\leq H_{\rmS}(P_1,P_2,V_1,V_2)$ (also  refer to~\cite[Eqns.~(82)-(87)]{draper2010lossless}). To be specific, we obtain the polynomial factor in \eqref{draper8287} by noting that for $l\leq m\leq \alpha_{q-1}<\beta_{q-1}$,
\begin{align}
\left|\calP_{n(\beta_{q-1}-m+1)}(\calY)\right|
&\leq \left(n(\beta_{q-1}-m+1)+1\right)^{|\calY|} \label{eqn:num_types1}\\
&\leq \left(n(\beta_{q-1}-l+1)+1\right)^{|\calX||\calY|}, \label{eqn:num_types2}\\
\left|\calV_{n(m-l)}(\calX;P_1)\right|\left|\calV_{n(\beta_{q-1}-m+1)}(\calX;\tilde{P}_2)\right|
&\leq \left(n(m-l)+1\right)^{|\calX||\calY|}\left(n(\beta_{q-1}-m+1)+1\right)^{|\calX||\calY|}\label{eqn:num_types3}\\
&\leq \left(n(\beta_{q-1}-l+1)+1\right)^{2|\calX||\calY|}.\label{eqn:num_types4}
\end{align}


Combining \eqref{upperstep1} and \eqref{draper8287}, we obtain
\begin{align}
&\nn\Pr\left(\calE_{\mathrm{sw}}^1(l,m)\right)\\
&\leq \sum_{(\bx_{t_{q-1}}^{\beta_{q-1}},\by_{t_{q-1}}^{\beta_{q-1}})}P_{XY}^{n\Psi}(\bx_{t_{q-1}}^{\beta_{q-1}},\by_{t_{q-1}}^{\beta_{q-1}})
\min\left\{1,q(\bx_{t_{q-1}}^{\beta_{q-1}},\by_{t_{q-1}}^{\beta_{q-1}})\right\}\\
&=\sum_{(\bx_{t_{q-1}}^{l-1},\by_{t_{q-1}}^{l-1})}\!\!\!\!
P_{XY}^{n(l-t_q+1)}(\bx_{t_{q-1}}^{l-1},\by_{t_{q-1}}^{l-1})\!\!\!\!
\sum_{(\bx_l^{\beta_{q-1}},\by_l^{\beta_{q-1}})}\!\!\!\!
P_{XY}^{n(\beta_q-l+1)}(\bx_l^{\beta_{q-1}},\by_l^{\beta_{q-1}})
\min\left\{1,q(\bx_{t_{q-1}}^{\beta_{q-1}},\by_{t_{q-1}}^{\beta_{q-1}})\right\} \label{eqn:split}\\
\nn&=\sum_{(\bx_{t_{q-1}}^{l-1},\by_{t_{q-1}}^{l-1})}\!\!\!\!
P_{XY}^{n(l-t_q+1)}(\bx_{t_{q-1}}^{l-1},\by_{t_{q-1}}^{l-1})\!\!\!\!
\sum_{\substack{ P_1\in\calP_{n(m-l)}(\calY)\\P_2\in\calP_{n(\beta_{q-1}-m+1)}(\calY)\\ V_1\in\calV_{n(m-l)}(\calX;P_1)\\V_2\in\calV_{n(\beta_{q-1}-m+1)}(\calX;P_2)}}\sum_{\substack{\by_l^{m-1}\in\calT_{P_1}\\ \by_m^{\beta_{q-1}}\in\calT_{P_2}}}\sum_{\substack{\bx_l^{m-1}\in\calT_{V_1}(\by_l^{m-1})\\\bx_m^{\beta_{q-1}}\in\calT_{V_2}(\by_m^{\beta_{q-1}})}}\\
&\qquad \qquad \qquad \qquad \qquad P_{XY}^{n(\beta_{q-1}-l+1)}(\bx_l^{\beta_{q-1}},\by_l^{\beta_{q-1}})
\min\left\{1,q(\bx_{t_{q-1}}^{\beta_{q-1}},\by_{t_{q-1}}^{\beta_{q-1}}))\right\} \label{eqn:split_types}\\
\nn&\leq \left(\sum_{(\bx_{t_{q-1}}^{l-1},\by_{t_{q-1}}^{l-1})}\!\!\!\!
P_{XY}^{n(l-t_q+1)}(\bx_{t_{q-1}}^{l-1},\by_{t_{q-1}}^{l-1})\right)\!\!\!\!
\sum_{\substack{P_1\in\calP_{n(m-l)}(\calY)\\P_2\in\calP_{n(\beta_{q-1}-m+1)}(\calY)}}\sum_{\substack{V_1\in\calV_{n(m-l)}(\calX;P_1)\\V_2\in\calV_{n(\beta_{q-1}-m+1)}(\calX;P_2)}}\\*
\nn&\qquad\exp\left\{-n(m-l)D(P_1\times V_1\|P_{XY})-n(\beta_{q-1}-m+1)D(P_2\times V_2\|P_{XY})\right\}\\*
&\qquad \times\min\left\{1,\frac{ (n(\beta_{q-1}-l+1)+1)^{3|\calX||\calY|}\exp(n(\beta_{q-1}-l+1)H_{\rmS}(P_1,P_2,V_1,V_2))}{(N_1N_2)^{\Psi-\Omega+1}}\right\}\label{complicated}\\
\nn&=\sum_{\substack{P_1\in\calP_{n(m-l)}(\calY)\\P_2\in\calP_{n(\beta_{q-1}-m+1)}(\calY)}}\sum_{\substack{V_1\in\calV_{n(m-l)}(\calX;P_1)\\V_2\in\calV_{n(\beta_{q-1}-m+1)}(\calX;P_2)}}\\*
\nn&\qquad\exp\left\{-n(m-l)D(P_1\times V_1\|P_{XY})-n(\beta_{q-1}-m+1)D(P_2\times V_2\|P_{XY})\right\}\\
&\qquad \times(n(\beta_{q-1}-l+1)+1)^{3|\calX||\calY|} \min\left\{1,\frac{ \exp\left(n(\beta_{q-1}-l+1)H_{\rmS}(P_1,P_2,V_1,V_2)\right)}{(N_1N_2)^{\Psi-\Omega+1}}\right\} \label{eqn:one}\\
&\leq (n(\beta_{q-1}-l+1)+1))^{7|\calX||\calY|}\exp\left\{-n(\beta_{q-1}-l+1) E_X\left(R_{X,n}^l,R_{Y,n}^m,\frac{m-l}{\beta_q-l+1}\right)\right\}\label{defexgamma2}\\
&\leq (n(\beta_{q-1}-l+1)+1)^{7|\calX||\calY|}\exp\left\{-n(\beta_{q-1}-l+1) \inf_{\gamma\in[0,1]}E_X(R_{X,n}^l,R_{Y,n}^m,\gamma)\right\}\label{defexgamma},
\end{align}
where in~\eqref{eqn:split} we split the product distribution $P_{XY}^{n\Psi}$ and the sum into two parts (from $t_{q-1}$ to $l-1$ and from $l$ to $\beta_{q-1}$);  in~\eqref{eqn:split_types} we split the inner sum into types and conditional types; \eqref{complicated} follows from using standard results from the method of types \cite{csiszar2011information} and the upper bound on $q(\bx_{t_{q-1}}^{\beta_{q-1}},\by_{t_{q-1}}^{\beta_{q-1}})$ in \eqref{draper8287};  \eqref{eqn:one} follows by using the fact that the first term in parentheses in  \eqref{complicated} is unity; \eqref{defexgamma2} follows from upper bound on the number of (conditional) types similarly as in \eqref{draper8287} (also see the calculations in \eqref{eqn:num_types1}--\eqref{eqn:num_types4}), the definition  of $E_X(R_X,R_Y,\gamma)$ in \eqref{gammax}, the codeword sizes $N_1$ in \eqref{m1size} and $N_2$ in \eqref{m2size}, and the rates $R_{X,n}^l$ in \eqref{rxnl} and $R_{Y,n}^m$ in \eqref{rynm}; and \eqref{defexgamma} follows since $\exp(-a)$ is decreasing in $a$. The subsequent analyses for other error events are similar to \eqref{defexgamma} and thus we present the results only without giving detailed proofs.

\item $l\in[t_{q-1}:\alpha_{q-1}]$ and $m\in[\alpha_{q-1}+1:\beta_{q-1}]$: Similarly to the  steps leading to \eqref{defexgamma}, we obtain
\begin{align}
\Pr\left(\calE_{\mathrm{sw}}^1(l,m)\right)&\leq (n(\beta_{q-1}-l+1)+1)^{7|\calX||\calY|}\exp\left\{-n(\beta_{q-1}-l+1) \inf_{\gamma\in[0,1]}E_X(R_{X,n}^l,R_{Y,n},\gamma)\right\}\label{sw1:step2}.
\end{align}

\item $l,m\in[t_{q-1}:\alpha_{q-1}]$ and $l\geq m$: Similarly to the  steps leading to \eqref{defexgamma}, we obtain
\begin{align}
\Pr\left(\calE_{\mathrm{sw}}^1(l,m)\right)&\leq (n(\beta_{q-1}-m+1)+1)^{7|\calX||\calY|}\exp\left\{-n(\beta_{q-1}-m+1) \inf_{\gamma\in[0,1]}E_Y(R_{X,n}^l,R_{Y,n}^m,\gamma)\right\}\label{sw1:step3}.
\end{align}

\item $m\in[t_{q-1}:\alpha_{q-1}]$ and $l\in[\alpha_{q-1}+1:\beta_{q-1}]$: Similarly to the  steps leading to \eqref{defexgamma}, we obtain
\begin{align}
\Pr\left(\calE_{\mathrm{sw}}^1(l,m)\right)&\leq (n(\beta_{q-1}-m+1)+1)^{7|\calX||\calY|}\exp\left\{-n(\beta_{q-1}-m+1) \inf_{\gamma\in[0,1]}E_Y(R_{X,n},R_{Y,n}^m,\gamma)\right\}\label{sw1:step4}.
\end{align}
\end{enumerate}
Therefore, combining the results in \eqref{defexgamma}, \eqref{sw1:step2}, \eqref{sw1:step3} and \eqref{sw1:step4}, we obtain that for all $n\in\bbN$,
\begin{align}
\nn&\Pr\left(\bigcup_{\substack{l,m\in[t_{q-1}:\beta_{q-1}]:\\\min\{l,m\}\leq \alpha_{q-1}}}\calE_{\mathrm{sw}}^1(l,m)\right)\leq \sum_{\substack{l,m\in[t_{q-1}:\beta_{q-1}]:\\\min\{l,m\}\leq \alpha_{q-1}}}\Pr\left(\calE_{\mathrm{sw}}^1(l,m)\right)\\
\nn&\leq \sum_{l=t_{q-1}}^{\alpha_{q-1}}\sum_{m=l}^{\alpha_{q-1}} (n(\beta_{q-1}-l+1)+1)^{7|\calX||\calY|}\exp\left\{-n(\beta_{q-1}-l+1) \inf_{\gamma\in[0,1]}E_X(R_{X,n}^l,R_{Y,n}^m,\gamma)\right\}\\*
\nn&\qquad+\sum_{m=t_{q-1}}^{\alpha_{q-1}}\sum_{l=m}^{\alpha_{q-1}} (n(\beta_{q-1}-m+1)+1)^{7|\calX||\calY|}\exp\left\{-n(\beta_{q-1}-m+1) \inf_{\gamma\in[0,1]}E_Y(R_{X,n}^l,R_{Y,n}^m,\gamma)\right\}\\*
\nn&\qquad+\sum_{l=t_{q-1}}^{\alpha_{q-1}}\sum_{m=\alpha_{q-1}+1}^{\beta_{q-1}}(n(\beta_{q-1}-l+1)+1)^{7|\calX||\calY|}\exp\left\{-n(\beta_{q-1}-l+1) \inf_{\gamma\in[0,1]}E_X(R_{X,n}^l,R_{Y,n},\gamma)\right\}\\*
&\qquad+\sum_{m=t_{q-1}}^{\alpha_{q-1}}\sum_{l=\alpha_{q-1}+1}^{\beta_{q-1}} (n(\beta_{q-1}-m+1)+1)^{7|\calX||\calY|}\exp\left\{-n(\beta_{q-1}-m+1) \inf_{\gamma\in[0,1]}E_Y(R_{X,n},R_{Y,n}^m,\gamma)\right\}\label{manysum}\\
\nn&\leq 2\Psi^2(n\Psi+1)^{7|\calX||\calY|}\exp\left\{-n(\Psi-\Omega+1)\min\left\{\inf_{\gamma\in[0,1]}E_X(R_{X,n}^l,R_{Y,n}^m,\gamma),\inf_{\gamma\in[0,1]}E_Y(R_{X,n}^l,R_{Y,n}^m,\gamma)\right\}\right\}\\*
&\qquad+2\Psi^2(n\Psi+1)^{7|\calX||\calY|}\exp\left\{-n(\Psi-\Omega+1)\min\left\{\inf_{\gamma\in[0,1]}E_X(R_{X,n}^l,R_{Y,n},\gamma),\inf_{\gamma\in[0,1]}E_Y(R_{X,n},R_{Y,n}^m,\gamma)\right\}\right\},
\label{upcalesw1}
\end{align}
where \eqref{upcalesw1} follows since using the definitions of $\alpha_q$ in \eqref{def:alphaq}, $\beta_q$ in \eqref{def:betaq}, $t_q$ in \eqref{def:tq} and the conditions that $\Psi>2\Omega$, $\Omega\geq T$, we find that $l,m\in[t_{q-1}:\alpha_{q-1}]$,
\begin{align}
\alpha_{q-1}-t_{q-1}+1&=\Omega<\Psi,\label{ipsw1step0}\\
\beta_{q-1}-(\alpha_{q-1}+1)+1&=\Psi-\Omega<\Psi,\label{ipsw1step1}\\
\Psi-\Omega+1&\leq \beta_{q-1}-l+1\leq \Psi,\label{ipsw1step2}\\
\Psi-\Omega+1&\leq \beta_{q-1}-m+1\leq \Psi\label{ipsw1step3}.
\end{align}

Specifically, using the bound in \eqref{ipsw1step0}, we can upper bound the number of the summands in the first two sums in \eqref{manysum} by $2\Psi^2$. Furthermore, using \eqref{ipsw1step2} and \eqref{ipsw1step3}, each summand in the first two sums in \eqref{manysum} can be upper bounded by the first term in \eqref{upcalesw1} divided by $2\Psi^2$. Hence, the sum of first two sums in \eqref{manysum} is upper bounded by the the first term in \eqref{upcalesw1}. Similarly, using \eqref{ipsw1step0} to \eqref{ipsw1step3}, we can upper bound the sum of last two sums in \eqref{manysum} by the second term in~\eqref{upcalesw1}.
 
\item $\calE_{\mathrm{sw},j}^2(l,m)$ defined in  \eqref{calsw2}: Similarly to the steps leading to \eqref{defexgamma}, we obtain
\begin{enumerate}
\item $l\leq m$
\begin{align}
\Pr\left(\calE_{\mathrm{sw},j}^2(l,m)\right)
&\leq (n(\beta_{q-1}-l+1)+1)^{7|\calX||\calY|}\exp\left\{-n(\beta_{q-1}-l+1) \inf_{\gamma\in[0,1]}E_X(R_{X,n},R_{Y,n},\gamma)\right\}
\end{align}
\item $m\leq l$
\begin{align}
\Pr\left(\calE_{\mathrm{sw},j}^2(l,m)\right)
&\leq (n(\beta_{q-1}-m+1)+1)^{7|\calX||\calY|}\exp\left\{-n(\beta_{q-1}-m+1) \inf_{\gamma\in[0,1]}E_Y(R_{X,n},R_{Y,n},\gamma)\right\}
\end{align}
\end{enumerate}
Therefore, for all $n\in\bbN$,
\begin{align}
\nn&\Pr\left(\bigcup_{\substack{j\in[\alpha_{q-1}+1:t_q-1]\\l,m\in[j:\beta_{q-1}]\\\min\{l,m\}=j}}\calE_{\mathrm{sw},j}^2(l,m)\right)\\
\nn&=\sum_{j=\alpha_{q-1}+1}^{t_q-1}\Bigg(\sum_{m=j}^{\beta_{q-1}} (n(\beta_{q-1}-j+1)+1)^{7|\calX||\calY|}\exp\left\{-n(\beta_{q-1}-j+1) \inf_{\gamma\in[0,1]}E_X(R_{X,n},R_{Y,n},\gamma)\right\}\\
&\qquad+\sum_{l=j}^{\beta_{q-1}} (n(\beta_{q-1}-j+1)+1)^{7|\calX||\calY|}\exp\left\{-n(\beta_{q-1}-j+1) \inf_{\gamma\in[0,1]}E_Y(R_{X,n},R_{Y,n},\gamma)\right\}\Bigg)\label{upcalesw2ini}\\
&\leq 2\Psi^2(n\Psi+1)^{7|\calX||\calY|}\exp\Bigg\{-n\Omega\min\Big\{\inf_{\gamma\in[0,1]}E_X(R_{X,n},R_{Y,n},\gamma),\inf_{\gamma\in[0,1]}E_Y(R_{X,n},R_{Y,n},\gamma)\Big\}\Bigg\}\label{upcalesw2},
\end{align}
where \eqref{upcalesw2} follows by invoking the defintions of $\alpha_q$, $\beta_q$ and $t_q$ in \eqref{def:alphaq}, \eqref{def:betaq} and \eqref{def:tq} and concluding that for $j\in[\alpha_{q-1}+1:t_q-1]$ and $l,m\in[j:\beta_{q-1}]$,
\begin{align}
(t_q-1)-(\alpha_{q-1}+1)+1&\leq t_q-1-\alpha_{q-1}=\Psi-2\Omega+1<\Psi,\label{ipsw2step0}\\
\beta_{q-1}-l+1&\leq \beta_{q-1}-\alpha_{q-1}=\Psi-\Omega<\Psi,\label{ipsw2step1}\\
\beta_{q-1}-m+1&\leq \beta_{q-1}-\alpha_{q-1}=\Psi-\Omega<\Psi,\label{ipswstep2}\\
\beta_{q-1}-j+1&\geq \beta_{q-1}-t_q+2=\Omega,\label{ipsw2step3}\\
\beta_{q-1}-j+1&\leq \beta_{q-1}-(\alpha_{q-1}+1)+1=\Psi-\Omega<\Psi\label{ipsw2step4}.
\end{align}

Specifically, using the bounds in \eqref{ipsw2step0}, \eqref{ipsw2step1} and \eqref{ipswstep2}, we can upper bound the number of summands in the two sums in \eqref{upcalesw2ini} by $2\Psi^2$. Furthermore, using the bounds in  \eqref{ipsw2step3} and \eqref{ipsw2step4}, each summand in the sums of \eqref{upcalesw2ini} can be upper bounded by \eqref{upcalesw2} divided by $2\Psi^2$.

We remark the reasoning for last step in \eqref{upcalesw2} holds similarly in the analyses of other error events for the last steps of the bounding of the error probabilities in the sequel. See \eqref{upcalesw3}, \eqref{upcalesw5} and \eqref{upcalesw6} where we omit the reasonings  for them for the sake of brevity.

\item $\calE_{\mathrm{sw},j}^3(l,m)$ defined in \eqref{calsw3}: Similarly to the  steps leading to \eqref{defexgamma}, we obtain
\begin{enumerate}
\item $l\leq m$
\begin{align}
\Pr\left(\calE_{\mathrm{sw},j}^3(l,m)\right)
&\leq (n(\beta_q-l+1)+1)^{7|\calX||\calY|}\exp\left\{-n(\beta_q-l+1) \inf_{\gamma\in[0,1]}E_X(R_{X,n},R_{Y,n},\gamma)\right\}
\end{align}
\item $m\leq l$
\begin{align}
\Pr\left(\calE_{\mathrm{sw},j}^3(l,m)\right)
&\leq (n(\beta_q-m+1)+1)^{7|\calX||\calY|}\exp\left\{-n(\beta_q-m+1) \inf_{\gamma\in[0,1]}E_Y(R_{X,n},R_{Y,n},\gamma)\right\}
\end{align}
\end{enumerate}
Therefore, for all $n\in\bbN$,
\begin{align}
\nn&\Pr\left(\bigcup_{\substack{j\in[t_q:\beta_{q-1}]\\l,m\in[j : \beta_q]\\\min\{l,m\}=j}}\calE_{\mathrm{sw},j}^3(l,m)\right)\\
\nn&=\sum_{j=t_q}^{\beta_{q-1}}\Bigg(\sum_{m=j}^{\beta_q} (n(\beta_q-j+1)+1)^{7|\calX||\calY|}\exp\left\{-n(\beta_q-j+1) \inf_{\gamma\in[0,1]}E_X(R_{X,n},R_{Y,n},\gamma)\right\}\\
&\qquad+\sum_{l=j}^{\beta_q} (n(\beta_q-j+1)+1)^{7|\calX||\calY|}\exp\left\{-n(\beta_q-j+1) \inf_{\gamma\in[0,1]}E_Y(R_{X,n},R_{Y,n},\gamma)\right\}\Bigg)\\
&\leq 2\Psi^2(n\Psi+1)^{7|\calX||\calY|}\exp\Bigg\{-n(\Psi-\Omega+2)\min\Big\{\inf_{\gamma\in[0,1]}E_X(R_{X,n},R_{Y,n},\gamma),\inf_{\gamma\in[0,1]}E_Y(R_{X,n},R_{Y,n},\gamma)\Big\}\Bigg\}.\label{upcalesw3}
\end{align}

\item $\calE_{\mathrm{sw},j}^5(l,m)$ defined in \eqref{calsw5}: Similarly to the  steps leading to \eqref{defexgamma}, we obtain
\begin{enumerate}
\item $l\leq m$
\begin{align}
\Pr\left(\calE_{\mathrm{sw},j}^5(l,m)\right)
&\leq (n(\beta_q-l+1)+1)^{7|\calX||\calY|}\exp\left\{-n(\beta_q-l+1) \inf_{\gamma\in[0,1]}E_X(R_{X,n},R_{Y,n},\gamma)\right\}
\end{align}
\item $m\leq l$
\begin{align}
\Pr\left(\calE_{\mathrm{sw},j}^5(l,m)\right)
&\leq (n(\beta_q-m+1)+1)^{7|\calX||\calY|}\exp\left\{-n(\beta_q-m+1) \inf_{\gamma\in[0,1]}E_Y(R_{X,n},R_{Y,n},\gamma)\right\}
\end{align}
\end{enumerate}
Therefore, for all $n\in\bbN$,
\begin{align}
\nn&\Pr\left(\bigcup_{\substack{j\in[\alpha_q:t_{q+1}-1]\\l,m\in[j:\beta_q]\\\min\{l,m\}=j}}\calE_{\mathrm{sw},j}^5(l,m)\right)\\
\nn&=\sum_{j=\alpha_q}^{t_{q+1}-1}\Bigg(\sum_{m=j}^{\beta_q} (n(\beta_q-j+1)+1)^{7|\calX||\calY|}\exp\left\{-n(\beta_q-j+1) \inf_{\gamma\in[0,1]}E_X(R_{X,n},R_{Y,n},\gamma)\right\}\\
&\qquad+\sum_{l=j}^{\beta_q} (n(\beta_q-j+1)+1)^{7|\calX||\calY|}\exp\left\{-n(\beta_q-j+1) \inf_{\gamma\in[0,1]}E_Y(R_{X,n},R_{Y,n},\gamma)\right\}\Bigg)\\
&\leq 2\Psi^2(n\Psi+1)^{7|\calX||\calY|}\exp\Bigg\{-n\Omega\min\Big\{\inf_{\gamma\in[0,1]}E_X(R_{X,n},R_{Y,n},\gamma),\inf_{\gamma\in[0,1]}E_Y(R_{X,n},R_{Y,n},\gamma)\Big\}\Bigg\}\label{upcalesw5}.
\end{align}

\item $\calE_{\mathrm{sw},j}^6(l,m)$ defined in \eqref{calsw6}: Similarly to the  steps leading to \eqref{defexgamma}, we obtain
\begin{enumerate}
\item $l\leq m$
\begin{align}
\Pr\left(\calE_{\mathrm{sw},j}^6(l,m)\right)
&\leq (n(T_k-l+1)+1)^{7|\calX||\calY|}\exp\left\{-n(T_k-l+1) \inf_{\gamma\in[0,1]}E_X(R_{X,n},R_{Y,n},\gamma)\right\}
\end{align}
\item $m\leq l$
\begin{align}
\Pr\left(\calE_{\mathrm{sw},j}^6(l,m)\right)
&\leq (n(T_k-m+1)+1)^{7|\calX||\calY|}\exp\left\{-n(T_k-m+1) \inf_{\gamma\in[0,1]}E_Y(R_{X,n},R_{Y,n},\gamma)\right\}
\end{align}
\end{enumerate}
Therefore, for all $n\in\bbN$,
\begin{align}
\nn&\Pr\left(\bigcup_{\substack{j\in[t_{q+1}:k]\\l,m\in[j:T_k]\\\min\{l,m\}=j}}\calE_{\mathrm{sw},j}^6(l,m)\right)\\
\nn&=\sum_{j=t_{q+1}}^k \Bigg(\sum_{m=j}^{T_k} (n(T_k-j+1)+1)^{7|\calX||\calY|}\exp\left\{-n(T_k-j+1) \inf_{\gamma\in[0,1]}E_X(R_{X,n},R_{Y,n},\gamma)\right\}\\
&\qquad+\sum_{l=j}^{T_k} (n(T_k-j+1)+1)^{7|\calX||\calY|}\exp\left\{-n(T_k-j+1) \inf_{\gamma\in[0,1]}E_Y(R_{X,n},R_{Y,n},\gamma)\right\}\Bigg)\\
&\leq 2\Psi^2(n\Psi+1)^{7|\calX||\calY|}\exp\Bigg\{-nT\min\Big\{\inf_{\gamma\in[0,1]}E_X(R_{X,n},R_{Y,n},\gamma),\inf_{\gamma\in[0,1]}E_Y(R_{X,n},R_{Y,n},\gamma)\Big\}\Bigg\}\label{upcalesw6}.
\end{align}
\end{enumerate}

We are now ready to bound the error probability in decoding $(\bX_k,\bY_k)$ at time $T_k$ for $k\in\calS(q)$ with $q\geq 2$. Recalling the analyses of error events in Section \ref{analysisep} (cf. \eqref{calsw1} to \eqref{calsw6}), we conclude that
\begin{align}
 \nn&\Pr\left((\hat{\bX}_k,\hat{\bY}_k)\neq (\bX_k,\bY_k)\right)\\
 \nn&\leq \Pr\left(\bigcup_{\substack{l,m\in[t_{q-1}:\beta_{q-1}]:\\\min\{l,m\}\leq \alpha_{q-1}}}\calE_{\mathrm{sw}}^1(l,m)\right)+\Pr\left(\bigcup_{\substack{j\in[\alpha_{q+1}:t_q-1]\\l,m\in[j:\beta_{q-1}]\\\min\{l,m\}=j}}\calE_{\mathrm{sw},j}^2(l,m)\right)+\Pr\left(\bigcup_{\substack{j\in[t_q:\beta_{q-1}]\\l,m\in[j:\beta_q]\\\min\{l,m\}=j}}\calE_{\mathrm{sw},j}^3(l,m)\right)\\
&\qquad+\Pr\left(\bigcup_{\substack{j\in[\alpha_q:t_{q+1}-1]\\l,m\in[j:\beta_q]\\\min\{l,m\}=j}}\calE_{\mathrm{sw},j}^5(l,m)\right)+\Pr\left(\bigcup_{\substack{j\in[t_{q+1}:k]\\l,m\in[j:T_k]\\\min\{l,m\}=j}}\calE_{\mathrm{sw},j}^6(l,m)\right)\label{uppfinalerror}.
\end{align}

\subsection{Asymptotic Behavior of the Exponents}

We choose the largest and smallest memories to  be
\begin{align}
\Psi&=n^{\frac{1}{2}+\delta}\label{defpsi}\\
\Omega&=2T\label{defomega},
\end{align}
where $\delta\in(0,\frac{1}{2})$. Note these choices of $\Psi$ and $\Omega$ satisfy the two conditions   $\Psi>2\Omega$ and $\Omega\geq T$ for $n$ large enough since $T$ is a constant.

Define the doubly-indexed sequence 
\begin{align}
\kappa_{n,l}:=\left(1-\frac{\Psi-\Omega+1}{\beta_{q-1}-l+1}\right)(R_X^*+\theta_1\xi_n)\label{def:taunl}.
\end{align}
Invoking the definitions in \eqref{m1size} and \eqref{rxnl}, we obtain 
\begin{align}
R_{X,n}^l
&=\frac{\Psi-\Omega+1}{n(\beta_{q-1}-l+1)}\log N_1\\
&=R_X^*+\theta_1\xi_n-\left(1-\frac{\Psi-\Omega+1}{\beta_{q-1}-l+1}\right)(R_X^*+\theta_1\xi_n)\\
&=R_X^*+\theta_1\xi_n-\kappa_{n,l}=R_{X,n}-\kappa_{n,l}\label{deviaterxn}.
\end{align}

Using the definitions of $\alpha_q$ in \eqref{def:alphaq}, $\beta_q$ in \eqref{def:betaq} and $t_q$ in \eqref{def:tq}, we obtain 
\begin{align}
\beta_q-\alpha_q&=\Psi-\Omega,\\
\beta_q-t_q&=\Psi-1.
\end{align}
Hence, for $l\in[t_{q-1} : \alpha_{q-1}]$, we have
\begin{align}
\Psi-\Omega+1\leq \beta_{q-1}-l+1 \leq \Psi\label{betaqlm}.
\end{align}
Further, invoking \eqref{def:taunl}, we conclude that for every $l\in[t_{q-1}:\alpha_{q-1}]$,
\begin{align}
\kappa_{n,l}
&\leq \left(1-\frac{\Psi-\Omega+1}{\Psi}\right)(R_X^*+\theta_1\xi_n)\\
&\leq \frac{\Omega-1}{\Psi}(R_X^*+\theta_1\xi_n)\\
&=o(\xi_n)\label{oxin},
\end{align}
where \eqref{oxin} holds by using the values of $\Psi$ in \eqref{defpsi} and $\Omega$ in \eqref{defomega}, and the asymptotic conditions on the sequence $\{\xi_n\}_{n\in\bbN}$ (See Definition \ref{defmdcsw}).  We remark that \eqref{oxin} holds for all $l\in[t_{q-1} : \alpha_{q-1}]$, i.e., for all such $l$ and for any $\varepsilon>0$, there exists $N=N_\varepsilon\in\bbN$ such that $\kappa_{n,l}/\xi_n<\varepsilon$ for all $n> N$. 

Define 
\begin{align}
\zeta_{n,m}:=\left(1-\frac{\Psi-\Omega+1}{\beta_{q-1}-m+1}\right)(R_Y^*+\theta_2\xi_n)\label{def:tannm}.
\end{align}
Similarly as \eqref{deviaterxn} and \eqref{oxin}, we can show that
\begin{align}
R_{Y,n}^m=R_Y^*+\theta\xi_n-\zeta_{n,m}=R_{Y,n}-\zeta_{n,m},
\end{align}
and for $m\in[t_{q-1} : \alpha_{q-1}]$,
\begin{align}
\zeta_{n,m}=o(\xi_n)\label{oxiny}.
\end{align}
We remark that  \eqref{oxiny} also holds for all $m \in [t_{q-1} : \alpha_{q-1}]$.

The next lemma presents the asymptotic behavior of the exponents of the error probabilities in  \eqref{upcalesw1}, \eqref{upcalesw2}, \eqref{upcalesw3}, \eqref{upcalesw5} and \eqref{upcalesw6}.
\begin{lemma}
\label{mdcasymp}
For $l,m\in[t_{q-1}:\alpha_{q-1}]$,
 we have
\begin{align}
\liminf_{n\to\infty} \frac{\min\left\{\inf_{\gamma\in[0,1]}E_X(R_{X,n}^l,R_{Y,n}^m,\gamma),\inf_{\gamma\in[0,1]}E_Y(R_{X,n}^l,R_{Y,n}^m,\gamma) \right\}}{\xi_n^2}&\geq L(R_X^*,R_Y^*)\label{eesw1}\\
\liminf_{n\to\infty} \frac{\min\left\{\inf_{\gamma\in[0,1]}E_X(R_{X,n}^l,R_{Y,n},\gamma),\inf_{\gamma\in[0,1]}E_Y(R_{X,n},R_{Y,n}^m,\gamma) \right\}}{\xi_n^2}&\geq L(R_X^*,R_Y^*)\\
\liminf_{n\to\infty} \frac{\min\left\{\inf_{\gamma\in[0,1]}E_X(R_{X,n},R_{Y,n},\gamma),\inf_{\gamma\in[0,1]}E_Y(R_{X,n},R_{Y,n},\gamma) \right\}}{\xi_n^2}&\geq L(R_X^*,R_Y^*),\label{eesw3}
\end{align}
where $L(R_X^*,R_Y^*)$ is defined as follows:
\begin{enumerate}
\item Case (i): $R_X^*=H(P_{X|Y}|P_Y)$ and $R_Y^*>H(P_Y)$
\begin{align}
L(R_X^*,R_Y^*)=\frac{\theta_1^2}{2\rmV(P_{X|Y}|P_Y)},
\end{align}
\item Case (ii): $R_X^*=H(P_{X|Y}|P_Y)$ and $R_Y^*=H(P_Y)$
\begin{align}
L(R_X^*,R_Y^*)=\min\left\{\inf_{\gamma\in[0,1]}\frac{(\theta_1+(1-\gamma)\theta_2)^2}{2\left(\gamma \rmV(P_{X|Y}|P_Y)+(1-\gamma)\rmV(P_{XY})\right)},
\frac{(\theta_1+\theta_2)^2}{2\rmV(P_{XY})}\right\}\end{align}
\item Case (iii): $R_X^*+R_Y^*=H(P_{XY})$, $H(P_{X|Y}|P_Y)<R_X^*<H(P_X)$ and $H(P_{Y|X}|P_X)<R_Y^*<H(P_Y)$
\begin{align}
L(R_X^*,R_Y^*)=\frac{(\theta_1+\theta_2)^2}{2\rmV(P_{XY})},
\end{align}
\end{enumerate}
and for Cases (iv) and (v), $L(R_X^*,R_Y^*)$ is defined similarly as Cases (ii) and (i) with $X$ and $Y$, $\theta_1$ and $\theta_2$ interchanged.
\end{lemma}

The proof of Lemma \ref{mdcasymp} is presented in Appendix \ref{prooflemmamdc}. We emphasize  that the same lower bounds for the limits in  \eqref{eesw1}--\eqref{eesw3} hold  regardless of the specific choices of $l,m\in[t_{q-1}:\alpha_{q-1}]$ due to the estimates $\kappa_{n,l}=o(\xi_n)$ (see \eqref{oxin}) and $\zeta_{n,m}=o(\xi_n)$ (see \eqref{oxiny}) and the fact that \eqref{oxin} and \eqref{oxiny} hold for all $l$ and $m$ in the specified range.

Using Lemma \ref{mdcasymp}, the bounds in \eqref{upcalesw1}, \eqref{upcalesw2}, \eqref{upcalesw3}, \eqref{upcalesw5} and \eqref{upcalesw6}, and the facts that $\Psi = n^{\frac{1}{2}+\delta}$ (cf.~\eqref{defpsi})  and $\Omega=2T$ (cf.~\eqref{defomega}), we conclude that the upper bound of the error probability in \eqref{uppfinalerror} is dominated by the final term in  \eqref{upcalesw6}.  Therefore, we obtain
\begin{align}
\nn&\liminf_{n\to\infty}-\frac{\log \Pr\left((\hat{\bX}_k,\hat{\bY}_k)\neq (\bX_k,\bY_k)\right)}{n\xi_n^2}\\
\nn&\geq \liminf_{n\to\infty} T\frac{\min\left\{\inf_{\gamma\in[0,1]}E_X(R_{X,n},R_{Y,n},\gamma),\inf_{\gamma\in[0,1]}E_Y(R_{X,n},R_{Y,n},\gamma) \right\}}{\xi_n^2}\\
&\qquad+\frac{\log 2+2\log \Psi}{n\xi_n^2}+\frac{7|\calX||\calY|\log (n\Psi+1)}{n\xi_n^2}\\
&\geq TL(R_X^*,R_Y^*)\label{acefinal},
\end{align}
where \eqref{acefinal} holds by invoking Lemma \ref{mdcasymp} and noting that $\frac{\log n}{n\xi_n^2}\to 0$ as $n\to\infty$.

\section{Conclusion}
\label{conc}
In this paper, we have derived an achievable moderate deviations constant for the blockwise streaming version of SW coding. We showed that the moderate deviations constant is enhanced by a multiplicative factor of at least $T$ (over the non-streaming setting) in many instances. 

A natural next step is to attempt to derive a converse (cf.~\cite{LeeTanKhisti2016}), possibly leveraging on feedforward decoders~\cite{sahai2008block,chang2007price}.  However, we envision significant challenges in obtaining a matching converse. The only tight result for streaming source coding was proved by Chang and Sahai in \cite{chang2006error} where they derived the optimal error exponent for symbolwise lossless point-to-point streaming. They proved the direct part by using  {\em fixed-to-variable-length} codes, coupled with a FIFO  encoder. We consider  {\em fixed-to-fixed-length} codes in this paper. Chang \emph{et al.}~\cite{chang2007price} made attempts to establish a converse result in the large deviations regime for symbolwise lossless streaming source coding with decoder side information  using feedforward decoding. However, the derived bounds on the error exponents differ significantly (from the achievability) except for some very pathological sources. If we adopt vanilla feedforward decoders~\cite{sahai2008block,chang2007streaming} for our  problem setting, we are able to derive a converse moderate deviations result  in which the moderate deviations constant is infinity, which is vacuous. This is because of the suboptimality of the bounds on the error exponents. Other possible research topics may include streaming lossy source coding \cite[Chapter 3]{chang2007streaming} and streaming  versions of other multi-terminal coding problems~\cite[Part~II]{el2011network}.

\appendix
\subsection{Proof of Proposition~\ref{ttime}}
\label{proofttime}

Note that for $\bm{\theta}\in\Theta_{\mathrm{(ii)}}$ (cf.~\eqref{def:Theta2}), $\theta_1>0$ and $\theta_1+\theta_2>0$. The conditions on $\bm{\theta}$ are essentially only concerning the ratio of $\theta_1$ and $\theta_2$. For ease of notation, we denote $f(P_{XY},\gamma,\bm{\theta})$ as $f(P_{XY},\gamma)$ to suppress the dependency on $\bm{\theta}$.  To further simplify notation, we also use $\rmV_{\rmc}:=\rmV(P_{X|Y}|P_Y)$ and $\rmV_\rmj:= \rmV(P_{XY})$ to denote the conditional and joint varentropies.

The first and second derivatives of $f(P_{XY},\gamma)$ with respect to $\gamma$ are
\begin{align}
f'(P_{XY},\gamma)&=\frac{-(\rmV_{\rmc}-\rmV_{\rmj})(\theta_1+(1-\gamma)\theta_2)^2}{2\left(\left(\gamma \rmV_{\rmc}+(1-\gamma)\rmV_{\rmj}\right)\right)^2}-\frac{\theta_2(\theta_1+(1-\gamma)\theta_2)}{\left(\gamma \rmV_{\rmc}+(1-\gamma)\rmV_{\rmj}\right)}\\
f''(P_{XY},\gamma)&=\frac{\left((\theta_1+\theta_2)\rmV_{\rmc}-\theta_1\rmV_{\rmj}\right)^2}{\left(\left(\gamma \rmV_{\rmc}+(1-\gamma)\rmV_{\rmj}\right)\right)^3}.
\end{align}
Hence, $f(P_{XY},\gamma)$ is a convex function in $\gamma$ since $f''(P_{XY},\gamma)\geq 0$.

Define
\begin{align}
\gamma^*=\argmin_{\gamma\in[0,1]} f(P_{XY},\gamma).
\end{align}

We first prove that \eqref{condtheta1} and \eqref{condtheta2} are sufficient conditions by considering the scenarios where $\gamma^*=1$ and $\gamma^*=0$. Then we prove that \eqref{condtheta1} and \eqref{condtheta2} are also necessary by considering the scenario where $\gamma^*\in(0,1)$.

\begin{enumerate}
\item $\gamma^*=1$

In order to achieve the infimum at $\gamma^*=1$, we need $f'(\gamma^*)\leq 0$, i.e.,
\begin{align}
f'(P_{XY},1)=\frac{-(\rmV_{\rmc}-\rmV_{\rmj})\theta_1^2}{2 \rmV_{\rmc}^2}-\frac{\theta_1\theta_2}{\rmV_{\rmc}}\leq 0
\end{align}
Hence, we have
\begin{align}
\frac{\theta_2}{\theta_1}\geq \frac{\rmV_{\rmj}-\rmV_{\rmc}}{2\rmV_{\rmc}}\label{thetacond1}.
\end{align}

Note that 
\begin{align}
f(P_{XY},1)=\frac{\theta_1^2}{2\rmV_{\rmc}}.
\end{align}
Hence, we have proved \eqref{tgain} holds when \eqref{condtheta1} is satisfied.

\item $\gamma^*=0$

In order to achieve the infimum at $\gamma^*=0$, we need $f'(\gamma^*)\geq 0$, i.e.,
\begin{align}
f'(P_{XY},0)=\frac{-(\rmV_{\rmc}-\rmV_{\rmj})(\theta_1+\theta_2)^2}{2  \rmV_{\rmj}^2}-\frac{\theta_2(\theta_1+\theta_2)}{\rmV_{\rmj}}\geq 0.
\end{align} 
Note that $\theta_1+\theta_2>0$ and $\rmV_\rmj +\rmV_\rmc >0$. After performing some algebra, we obtain
\begin{align}
\frac{\theta_2}{\theta_1}&\leq \frac{\rmV_{\rmj}-\rmV_{\rmc}}{\rmV_{\rmj}+\rmV_{\rmc}}\label{thetalower}
\end{align}

Further note that
\begin{align}
f(P_{XY},0)=\frac{(\theta_1+\theta_2)^2}{2\rmV_{\rmj}}.
\end{align}
Hence, in order to for \eqref{tgain} to hold, invoking \eqref{caseiinonstreaming}, we obtain
\begin{align}
\frac{(\theta_1+\theta_2)^2}{2\rmV_{\rmj}}\leq \frac{\theta_1^2}{2\rmV_{\rmc}}\label{require1}.
\end{align}
Invoking the constraint that $\theta_2>-\theta_1$ and $\theta_1>0$ (refer to  the definition of $\Theta_{( \mathrm{ii})}$ in \eqref{def:Theta2}), we obtain
\begin{align}
-1<\frac{\theta_2}{\theta_1}\leq \sqrt{\frac{\rmV_{\rmj}}{\rmV_{\rmc}}}-1\label{thetacond2}.
\end{align}

Combining~\eqref{thetalower} and \eqref{thetacond2}, we have proved that \eqref{tgain} holds when \eqref{condtheta2} holds.

\item $\gamma^*\in(0,1)$.

For this case, in order for \eqref{tgain} to hold, we need \eqref{require1} and 
\begin{align}
f(P_{XY},\gamma^*)\geq \frac{(\theta_1+\theta_2)^2}{2\rmV_{\rmj}}\label{caseiiicond}.
\end{align}

Since $\gamma^*$ minimizes $f(\gamma)$, we obtain
\begin{align}
f'(P_{XY},\gamma^*)=\frac{-(\rmV_{\rmc}-\rmV_{\rmj})(\theta_1+(1-\gamma^*)\theta_2)^2}{2\left(\left(\gamma^*\rmV_{\rmc}+(1-\gamma^*)\rmV_{\rmj}\right)\right)^2}-\frac{\theta_2(\theta_1+(1-\gamma^*)\theta_2}{\left(\gamma^* \rmV_{\rmc}+(1-\gamma^*)\rmV_{\rmj}\right)}=0\label{gammastar}.
\end{align}

Solving \eqref{gammastar}, we obtain
\begin{align}
\gamma^*
&=\frac{(\rmV_{\rmc}+ \rmV_{\rmj})\theta_2-(\rmV_{\rmj}-\rmV_{\rmc})\theta_1}{\theta_2\left(\rmV_{\rmj}-\rmV_{\rmc}\right)}\\
&=-\frac{\theta_1}{\theta_2}+\frac{\rmV_{\rmc}+\rmV_{\rmj}}{\rmV_{\rmj}-\rmV_{\rmc}}
\end{align}
Therefore, we obtain
\begin{align}
f(P_{XY},\gamma^*)=\frac{2\theta_2\left(\theta_1\rmV_{\rmj}-(\theta_1+\theta_2)\rmV_{\rmc}\right)}{\left(\rmV_{\rmj}-\rmV_{\rmc}\right)^2}.
\end{align}

Since \eqref{thetacond1} and \eqref{thetalower}  are, respectively, the conditions on $\theta_2/\theta_1$ that ensure that $\gamma^*=0$ and $\gamma^*=1$, when $\gamma^*\in(0,1)$, $\theta_2/\theta_1$ cannot satisfy either of \eqref{thetacond1} or \eqref{thetalower}, we conclude that $\gamma^*\in(0,1)$ implies
\begin{align}
 \frac{\rmV_{\rmj}-\rmV_{\rmc}}{\rmV_{\rmj}+\rmV_{\rmc}}<\frac{\theta_2}{\theta_1}< \frac{\rmV_{\rmj}-\rmV_{\rmc}}{2\rmV_{\rmc}}\label{thetainterval}.
\end{align}
Hence, $f(P_{XY},\gamma^*)>0$. In order to satisfy \eqref{caseiiicond}, we have
\begin{align}
\frac{2\theta_2\left(\theta_1\rmV_{\rmj}-(\theta_1+\theta_2)\rmV_{\rmc}\right)}{\left(\rmV_{\rmj}-\rmV_{\rmc}\right)^2}\geq \frac{(\theta_1+\theta_2)^2}{2\rmV_{\rmj}}\label{inerequire}.
\end{align}
Solving \eqref{inerequire}, we obtain
\begin{align}
-\left(\theta_1\left(\rmV_{\rmc}-\rmV_{\rmj}\right)+\theta_2\left(\rmV_{\rmc}+\rmV_{\rmj}\right)\right)^2\geq 0.
\end{align}
Thus, the only possible case is
\begin{align}
\frac{\theta_2}{\theta_1}=\frac{\rmV_{\rmj}-\rmV_{\rmc}}{\rmV_{\rmc}+\rmV_{\rmj}}
\end{align}
However, using the condition for $\gamma^*\in (0,1)$ given in~\eqref{thetainterval}, we conclude that it is impossible to satisfy \eqref{caseiiicond} if $\gamma^*\in(0,1)$.
\end{enumerate}
Therefore, we have proved that \eqref{condtheta1} and \eqref{condtheta2} are both sufficient and necessary conditions for \eqref{tgain} to hold.

\subsection{Preliminaries for the Proof of Lemma \ref{mdcasymp}}
\label{appprelim}
Define
\begin{align}
E_{XY}(\rho)&=(1+\rho)\log\sum_{x,y}P_{XY}(x,y)^{\frac{1}{1+\rho}}\label{eefunction1},\\
E_{X|Y}(\rho)&=\log \sum_{y}P_{Y}(y)\Big(\sum_{x}P_{X|Y}(x|y)^{\frac{1}{1+\rho}}\Big)^{1+\rho}\label{eefunction2},\\
E_X(\gamma,\rho)&=\gamma E_{X|Y}(\rho)+(1-\gamma)E_{XY}(\rho),\label{eefunction3}
\end{align}
and $E_{Y|X}(\rho)$ is similar to $E_{X|Y}(\rho)$ with $X$ and $Y$ interchanged; $E_Y(\gamma,\rho)$ is similar to $E_X(\gamma,\rho)$ with $X$ and $Y$ interchanged.

\begin{lemma}
\label{derivative:joint}
  For any pmf $P_{XY} \in \calP(\calX\times\calY)$ where $\calX$ and $\calY$ are finite sets,
  \begin{align}
    E_{XY}'(\rho)\Big|_{\rho=0}&=H(P_{XY}),\\
    E_{XY}''(\rho)\Big|_{\rho=0}&=\rmV(P_{XY}),\\
    \inf_{\rho\in[0,1]} E_{XY}''(\rho)&\geq 0.
   \end{align}

   Furthermore, there exists a finite positive number $M_{XY}$ such that
  \begin{align}
    \sup_{\rho\in[0,1]}\left|E_{XY}'''(\rho)\right|&\leq M_{XY}.
  \end{align}
\end{lemma}

\begin{lemma}
\label{derivative:conditional}
For any pmf $P_{XY} \in \calP(\calX\times\calY)$  where $\calX$ and $\calY$ are finite sets,
\begin{align}
  E_{X|Y}'(\rho)\Big|_{\rho=0}&=H(P_{X|Y}|P_{Y}),\\
  E_{X|Y}''(\rho)\Big|_{\rho=0}&=\rmV(P_{X|Y}|P_{Y}),\\
  \inf_{\rho\in[0,1]} E_{X|Y}''(\rho)&\geq 0.
\end{align}
   Furthermore, there exists a finite  positive number $M_X$ such that
\begin{align}
  \sup_{\rho\in[0,1]}\left|E_{X|Y}'''(\rho)\right|&\leq M_{X}.
\end{align}
The derivatives and properties of $E_{Y|X}(\rho)$ are exactly the same as $E_{X|Y}(\rho)$ with $X$ and $Y$ interchanged.
\end{lemma}

The proofs of Lemmas \ref{derivative:joint} and \ref{derivative:conditional} are  standard and thus omitted. See for example \cite{draper2010lossless} and \cite{altugwagner2014}.

Invoking Lemmas \ref{derivative:joint} and \ref{derivative:conditional}, the Taylor series expansions of \eqref{eefunction3} is as follows:
\begin{align}
  E_X(\gamma,\rho)
  \nn&=\gamma\rho H(P_{X|Y}|P_{Y})+(1-\gamma)\rho H(P_{XY})+\frac{\gamma\rho^2}{2}\rmV(P_{X|Y}|P_{Y})+\frac{(1-\gamma)\rho^2}{2}\rmV(P_{XY})\\
  &\qquad+\frac{\gamma\rho^3}{6}E_{X|Y}'''(\overline{\rho})+\frac{(1-\gamma)\rho^3}{6}E_{XY}'''(\overline{\rho})\label{taylorx},
\end{align}
for some $\overline{\rho}\in[0,\rho]$. The Taylor expansion of $E_Y(\gamma,\rho)$ is obtained by interchanging $X$ and $Y$ in \eqref{taylorx}.

\begin{lemma}
\label{convergegamma}
Let $f:[0,1]\to\bbR$, $g:[0,1]\to\bbR$ and $g_i:[0,1]\to\bbR$ where $i\in[1 : m]$ be continuous functions. Consider arbitrary positive sequences $\{a_n\}_{n=1}^{\infty}$, $\{b_n\}_{n=1}^{\infty}$
$\{c_{i,n}\}_{n=1}^{\infty}$ satisfying $b_n=o(a_n)$ and $c_{i,n}=o(b_n)$.
Then we have
\begin{align}
\inf_{\gamma\in[0,1]}\left(a_nf(\gamma)+b_n g(\gamma)+\sum_{i=1}^m c_{i,n}g_i(\gamma)\right)=a_nf^*+b_n g^*+o(b_n),
\end{align}
where 
\begin{align}
f^*&=\inf_{\gamma\in[0,1]}f(\gamma)\\
g^*&=\inf_{\gamma:f(\gamma)=f^*}g(\gamma).
\end{align}
\end{lemma}
Lemma \ref{convergegamma} can be proved similarly as \cite[Lemma 48]{polyanskiy2010thesis}.

\subsection{Proof of Lemma \ref{mdcasymp}}
\label{prooflemmamdc}
We prove \eqref{eesw1} only since others can be done similarly. 
Let $(R_X^*,R_Y^*)$ be fixed boundary point of the SW region (cf. Cases (i)-(v) in Figure \ref{rateregion}). We consider  an arbitrary sequence $\{\xi_n\}_{n=1}^{\infty}$ satisfying $\xi_n\to 0$ and $\frac{\log n}{n\xi_n^2}\to0$ as $n\to\infty$. 

Recall the definition of $R_{X,n}^l$, $R_{Y,n}^m$, $\kappa_{n,l}$ and $\zeta_{n,m}$ in \eqref{rxnl}, \eqref{rynm}, \eqref{def:taunl} and \eqref{def:tannm} respectively.

Define
\begin{align}
\rho_n&=\frac{|\theta_1+(1-\gamma)\theta_2|\xi_n}{\gamma \rmV(P_{X|Y}|P_Y)+(1-\gamma)\rmV(P_{XY})}\label{defrhon},
\end{align}
where $|x|$ is the absolute value of $x$. Note that for $n$ large enough, $\rho_n\in[0,1]$.  Hence, for any $\gamma\in[0,1]$,
\begin{align}
\nn&E_X(R_{X,n}^l,R_{Y,n}^m,\gamma)\\
&=\max_{\rho\in[0,1]} \gamma\rho R_{X,n}^l+(1-\gamma)\rho(R_{X,n}^l+R_{Y,n}^m)-E_{X}(\gamma,\rho)\\
&\geq \gamma\rho_n R_{X_n}^l-\gamma\left(\rho_n H(P_{X|Y}|P_{Y})
+\frac{\rho_n^2}{2}\rmV(P_{X|Y}|P_{Y})+\frac{\rho_n^3}{6}E_{X|Y}'''(\overline{\rho}_n)\right)\\
&\qquad+(1-\gamma)\rho_n(R_{X,n}^l+R_{Y,n}^m)-(1-\gamma)\left(\rho_n H(P_{XY})+\frac{\rho_n^2}{2}\rmV(P_{XY})+\frac{\rho_n^3}{6}E_{XY}'''(\overline{\rho}_n)\right)\label{taylorexpand}\\
&\nn=\rho_n\left(\gamma R_{X,n}^l+(1-\gamma)(R_{X,n}^l+R_{Y,n}^m)-\gamma H(P_{X|Y}|P_Y)-(1-\gamma)H(P_{XY})\right)\\
&\qquad-\gamma\left(\frac{\rho_n^2}{2}\rmV(P_{X|Y}|P_{Y})+\frac{\rho_n^3}{6}M_X\right)-(1-\gamma)\left(\frac{\rho_n^2}{2}\rmV(P_{XY})+\frac{\rho_n^3}{6}M_{XY}\right)\\
&\nn=\rho_n\left(\gamma R_X^*+(1-\gamma)(R_X^*+R_Y^*)-\gamma H(P_{X|Y}|P_Y)-(1-\gamma)H(P_{XY})\right)-\rho_n(\theta_1\kappa_{n,l}+(1-\gamma)\theta_2\zeta_{n,m})\\
&\qquad+\left(\rho_n(\theta_1+(1-\gamma)\theta_2)\xi_n-\rho_n^2\left(
\frac{\gamma \rmV(P_{X|Y}|P_{Y})+(1-\gamma)\rmV(P_{XY})}{2}\right)\right)-\frac{\rho_n^3(\gamma M_X+(1-\gamma)M_{XY})}{6}\\
\nn&=\rho_n\left(\gamma R_X^*+(1-\gamma)(R_X^*+R_Y^*)-\gamma H(P_{X|Y}|P_Y)-(1-\gamma)H(P_{XY})\right)-\theta_1\rho_n\kappa_{n,l}-(1-\gamma)\theta_2\rho_n\zeta_{n,m}\\
&\qquad+\frac{(\theta_1+(1-\gamma)\theta_2)^2\xi_n^2}{2\left(\gamma \rmV(P_{X|Y}|P_Y)+(1-\gamma)\rmV(P_{XY})\right)}-\frac{\rho_n^3(\gamma M_X+(1-\gamma)M_{XY})}{6}\\
\nn&\geq \rho_n\left(\gamma R_X^*+(1-\gamma)(R_X^*+R_Y^*)-\gamma H(P_{X|Y}|P_Y)-(1-\gamma)H(P_{XY})\right)-\theta_1\rho_n\kappa_{n,l}-(1-\gamma)\theta_2\rho_n\zeta_{n,m}\\
&\qquad+\frac{(\theta_1+(1-\gamma)\theta_2)^2\xi_n^2}{2\left(\gamma \rmV(P_{X|Y}|P_Y)+(1-\gamma)\rmV(P_{XY})\right)}-\frac{\rho_n^3}{6}\max\{M_{XY},M_X\}\label{exlbduniform},
\end{align}
where \eqref{taylorexpand} follows from \eqref{taylorx} and $\bar{\rho}_n\in[0,\rho_n]$.

Define
\begin{align}
\rho_n'&=\frac{|(1-\gamma)\theta_1+\theta_2|\xi_n}{\gamma \rmV(P_{Y|X}|P_X)+(1-\gamma)\rmV(P_{XY})}\label{defrhon2}.
\end{align}
Similarly, we obtain
\begin{align}
E_Y(R_{X,n}^l,R_{Y,n}^m,\gamma)
\nn&\geq \rho_n'\left(\gamma R_Y^*+(1-\gamma)(R_X^*+R_Y^*)-\gamma H(P_{Y|X}|P_X)-(1-\gamma)H(P_{XY})\right)-(1-\gamma)\theta_1\rho_n'\kappa_{n,l}-\theta_2\rho_n'\zeta_{n,m}\\
&\qquad+\frac{((1-\gamma)\theta_1+\theta_2)^2\xi_n^2}{2\left(\gamma \rmV(P_{Y|X}|P_X)+(1-\gamma)\rmV(P_{XY})\right)}-\frac{\rho_n'^3}{6}\max\left\{M_{XY},M_Y\right\}\label{eylbduniform}.
\end{align}

We then deal with different cases. Here we only prove the result for Cases (i)-(iii) because Case (iv) is symmetric to Case (ii) and Case (v) is symmetric to Case (i).
\begin{enumerate}
\item Case (i): $R_X^*=H(P_{X|Y}|P_Y)$ and $R_Y>H(P_Y)$

For this case, invoking \eqref{exlbduniform} and \eqref{eylbduniform}, we obtain
\begin{align}
E_X(R_{X,n}^l,R_{Y,n}^m,\gamma)
\nn&\geq \rho_n(1-\gamma)\left((R_X^*+R_Y^*)-H(P_{XY})\right)-\theta_1\rho_n\kappa_{n,l}-(1-\gamma)\theta_2\rho_n\zeta_{n,m}\\
&\qquad+\frac{(\theta_1+(1-\gamma)\theta_2)^2\xi_n^2}{2\left(\gamma \rmV(P_{X|Y}|P_Y)+(1-\gamma)\rmV(P_{XY})\right)}-\frac{\rho_n^3}{6}\max\{M_{XY},M_X\}\label{excase1} \end{align}
and\begin{align}
E_Y(R_{X,n}^l,R_{Y,n}^m,\gamma)
\nn&\geq \rho_n'\left(\gamma R_Y^*+(1-\gamma)(R_X^*+R_Y^*)-\gamma H(P_{Y|X}|P_X)-(1-\gamma)H(P_{XY})\right)\\
&\qquad-(1-\gamma)\theta_1\rho_n'\kappa_{n,l}-\theta_2\rho_n'\zeta_{n,m}+\frac{((1-\gamma)\theta_1+\theta_2)^2\xi_n^2}{2\left(\gamma \rmV(P_{Y|X}|P_X)+(1-\gamma)\rmV(P_{XY})\right)}   \nn\\ 
&\qquad -\frac{\rho_n'^3}{6}\max\left\{M_{XY},M_Y\right\}\label{eycase1}.
\end{align}

In order to evaluate the right hand side of \eqref{excase1}, we define the following functions:
\begin{align}
f(\gamma)&:=\frac{(1-\gamma)\left(R_X^*+R_Y^*-H(P_{XY})\right)|\theta_1+(1-\gamma)\theta_2|}{\gamma \rmV(P_{X|Y}|P_Y)+(1-\gamma)\rmV(P_{XY})}\\
g(\gamma)&:=\frac{(\theta_1+(1-\gamma)\theta_2)^2}{2\left(\gamma \rmV(P_{X|Y}|P_Y)+(1-\gamma)\rmV(P_{XY})\right)}\\
g_1(\gamma)&:=-\frac{|\theta_1+(1-\gamma)\theta_2|}{\gamma \rmV(P_{X|Y}|P_Y)+(1-\gamma)\rmV(P_{XY})}\\
g_2(\gamma)&:=-\frac{(1-\gamma)\theta_2|\theta_1+(1-\gamma)\theta_2|}{\gamma \rmV(P_{X|Y}|P_Y)+(1-\gamma)\rmV(P_{XY})}\\
g_3(\gamma)&:=-\frac{\max\left\{M_{XY},M_Y\right\}}{6}\left(\frac{|\theta_1+(1-\gamma)\theta_2|}{\gamma \rmV(P_{X|Y}|P_Y)+(1-\gamma)\rmV(P_{XY})}\right)^3.
\end{align}
Invoking the definition of $\rho_n$ in~\eqref{defrhon}, we conclude that minimizing the right hand side of \eqref{excase1} is equivalent to minimizing
\begin{align}
F(n,\gamma)
&:=f(\gamma)\xi_n+g(\gamma)\xi_n^2+g_1(\gamma)\rho_n\kappa_{n,l}+g_2(\gamma)\rho_n\zeta_{n,m}+g_3(\gamma)\xi_n^3.
\end{align}
Note that $\rho_n\kappa_{n,l}=o(\xi_n^2)$ (see \eqref{oxin} and \eqref{defrhon}) and similarly $\rho_n\zeta_{n,m}=o(\xi_n^2)$. Invoking Lemma \ref{convergegamma}, we obtain 
\begin{align}
\inf_{\gamma\in[0,1]}E_X(R_{X,n},R_{Y,n},\gamma)
&\geq \inf_{\gamma\in[0,1]}F(n,\gamma)\\
&=\frac{T\theta_1^2\xi_n^2}{2 \rmV(P_{X|Y}|P_Y)}+o(\xi_n^2)\label{useconverge}.
\end{align}
However, for any $\gamma\in[0,1]$, the right hand side of \eqref{eycase1} is dominated by the first term which is of order $\Theta(\xi_n)$ (see~\eqref{defrhon2}).

Therefore, we obtain
\begin{align}
\nn&\liminf_{n\to\infty} \frac{\min\left\{\inf_{\gamma\in[0,1]}E_X(R_{X,n}^l,R_{Y,n}^m,\gamma),\inf_{\gamma\in[0,1]}E_Y(R_{X,n}^l,R_{Y,n}^m,\gamma) \right\}}{\xi_n^2}\\
&\geq \min\left\{\frac{1}{2 \rmV(P_{X|Y}|P_Y)},\infty\right\}=\frac{1}{2 \rmV(P_{X|Y}|P_Y)}\label{swmdccase1}.
\end{align}

\item Case (ii): $R_X^*=H(P_{X|Y}|P_Y)$ and $R_Y=H(P_Y)$

Invoking \eqref{exlbduniform} and \eqref{eylbduniform}, we obtain
\begin{align}
E_X(R_{X,n}^l,R_{Y,n}^m,\gamma)
&\geq -\theta_1\rho_n\kappa_{n,l}-(1-\gamma)\theta_2\rho_n\zeta_{n,m}+\frac{(\theta_1+(1-\gamma)\theta_2)^2\xi_n^2}{2\left(\gamma \rmV(P_{X|Y}|P_Y)+(1-\gamma)\rmV(P_{XY})\right)} \nn\\*
&\qquad -\frac{\rho_n^3}{6}\max\{M_{XY},M_X\}\label{excase2}
\end{align}
and 
\begin{align}
E_Y(R_{X,n}^l,R_{Y,n}^m,\gamma)
&\nn\geq \rho_n'\gamma I(P_{X|Y},P_Y)-(1-\gamma)\theta_1\rho_n'\kappa_{n,l}-\theta_2\rho_n'\zeta_{n,m}\\
&\qquad+\frac{((1-\gamma)\theta_1+\theta_2)^2\xi_n^2}{2\left(\gamma \rmV(P_{Y|X}|P_X)+(1-\gamma)\rmV(P_{XY})\right)}-\frac{\rho_n'^3}{6}\max\left\{M_{XY},M_Y\right\}\label{eycase2}.
\end{align}
Following the argument leading to \eqref{swmdccase1}, we obtain
\begin{align}
\nn&\liminf_{n\to\infty} \frac{\min\left\{\inf_{\gamma\in[0,1]}E_X(R_{X,n}^l,R_{Y,n}^m,\gamma),\inf_{\gamma\in[0,1]}E_Y(R_{X,n}^l,R_{Y,n}^m,\gamma) \right\}}{\xi_n^2}\\
&\geq \min\left\{\inf_{\gamma\in[0,1]}\frac{(\theta_1+(1-\gamma)\theta_2)^2}{2\left(\gamma \rmV(P_{X|Y}|P_Y)+(1-\gamma)\rmV(P_{XY})\right)},\frac{(\theta_1+\theta_2)^2}{2\rmV(P_{XY})}\right\}
\end{align}

\item Case (iii): $H(P_{X|Y}|P_Y)<R_X^*<H(P_X)$, $H(P_{Y|X}|P_X)<R_Y^*<H(P_Y)$ and $R_X^*+R_Y^*=H(P_{XY})$.

Invoking \eqref{exlbduniform} and \eqref{eylbduniform}, we obtain
\begin{align}
E_X(R_{X,n}^l,R_{Y,n}^m,\gamma)
\nn&\geq \rho_n\gamma\left(R_X^*-H(P_{X|Y}|P_Y)\right)-\theta_1\rho_n\kappa_{n,l}-(1-\gamma)\theta_2\rho_n\zeta_{n,m}\\
&\qquad+\frac{(\theta_1+(1-\gamma)\theta_2)^2\xi_n^2}{2\left(\gamma \rmV(P_{X|Y}|P_Y)+(1-\gamma)\rmV(P_{XY})\right)}-\frac{\rho_n^3}{6}\max\{M_{XY},M_X\}\label{excase3}\end{align}
and
\begin{align}
E_Y(R_{X,n}^l,R_{Y,n}^m,\gamma)
&\geq \rho_n'\gamma\left(R_Y^*-H(P_{Y|X}|P_X)\right)-(1-\gamma)\theta_1\rho_n'\kappa_{n,l}-\theta_2\rho_n'\zeta_{n,m}\nn\\
&\qquad+\frac{((1-\gamma)\theta_1+\theta_2)^2\xi_n^2}{2\left(\gamma \rmV(P_{Y|X}|P_X)+(1-\gamma)\rmV(P_{XY})\right)}-\frac{\rho_n'^3}{6}\max\left\{M_{XY},M_Y\right\}\label{eycase3}.
\end{align}
Following similar steps leading to \eqref{swmdccase1}, we obtain
\begin{align}
 \liminf_{n\to\infty} \frac{\min\left\{\inf_{\gamma\in[0,1]}E_X(R_{X,n}^l,R_{Y,n}^m,\gamma),\inf_{\gamma\in[0,1]}E_Y(R_{X,n}^l,R_{Y,n}^m,\gamma) \right\}}{\xi_n^2}\geq \frac{(\theta_1+\theta_2)^2}{2\rmV(P_{XY})}.
\end{align}
\end{enumerate}

\subsection*{Acknowledgments}
The authors would like to acknowledge Associate Editor Prof. Sandeep Pradhan and two anonymous reviewers for useful comments which helped to improve the quality of the current paper.
 
\bibliographystyle{IEEEtran}
\bibliography{IEEEfull_lin}
\end{document}